\documentclass{article}
\usepackage{fullpage}

\usepackage{amsmath, amssymb,bm}
\usepackage{hyperref}
\usepackage{enumerate,mathtools}
\usepackage{amsthm}
\usepackage{cite}
\usepackage{microtype}
\usepackage{color}
\usepackage{subcaption}

\newtheorem{theorem}{Theorem}
\newtheorem{cor}{Corollary}
\newtheorem{lemma}{Lemma}

\newtheorem{prop}{Proposition}

\theoremstyle{definition}
\newtheorem{definition}{Definition}

\newtheorem{remark}{Remark}

\newtheorem{example}{Example}

\usepackage{pgfplots}
\usetikzlibrary{positioning}
\pgfplotsset{compat=1.12}

\newcommand{\bE}{\bm{E}}

\newcommand{\bO}{\bm{O}}

\newcommand{\bW}{\bm{W}}
\newcommand{\bX}{\bm{X}}
\newcommand{\bY}{\bm{Y}}
\newcommand{\bZ}{\bm{Z}}

\newcommand{\bu}{\bm{u}}

\newcommand{\bx}{\bm{x}}
\newcommand{\by}{\bm{y}}
\newcommand{\bz}{\bm{z}}

\newcommand{\cA}{\mathcal{A}}

\newcommand{\cC}{\mathcal{C}}
\newcommand{\cD}{\mathcal{D}}

\newcommand{\cG}{\mathcal{G}}
\newcommand{\cH}{\mathcal{H}}
\newcommand{\cI}{\mathcal{I}}

\newcommand{\cM}{\mathcal{M}}

\newcommand{\cP}{\mathcal{P}}
\newcommand{\cQ}{\mathcal{Q}}
\newcommand{\cR}{\mathcal{R}}

\newcommand{\cT}{\mathcal{T}}
\newcommand{\cU}{\mathcal{U}}
\newcommand{\cV}{\mathcal{V}}

\newcommand{\cZ}{\mathcal{Z}}

\newcommand{\dd}{\mathsf{d}}

\newcommand{\psd}{\mathbb{S}_+}

\newcommand{\sym}{\mathbb{S}}

\newcommand{\bEx}{\ensuremath{\mathbb{E}}}

\newcommand{\ex}[1]{\ensuremath{\mathbb{E}\left[ #1\right]}}

\newcommand{\Vertiii}[1]{{\left\vert\kern-0.25ex\left\vert\kern-0.25ex\left\vert #1 
    \right\vert\kern-0.25ex\right\vert\kern-0.25ex\right\vert}}
\newcommand{\vertiii}[1]{{\vert\kern-0.25ex \vert\kern-0.25ex\vert #1 
    \vert\kern-0.25ex\vert\kern-0.25ex\vert}}

\DeclareMathOperator{\cov}{\mathsf{Cov}}
\DeclareMathOperator{\mmse}{\sf mmse}
\DeclareMathOperator{\var}{\sf Var}

\DeclareMathOperator{\diag}{\mathrm diag}

\newcommand{\fro}{\ensuremath{F}}

\DeclareMathOperator{\gtr}{tr}

\DeclareMathOperator{\gvec}{\mathsf{vec}}
\DeclareMathOperator{\vech}{\mathsf{vech}}

\DeclareMathOperator{\MMSE}{\mathsf{MMSE}}
\DeclareMathOperator{\FP}{\mathsf{FP}}

\newcommand{\reals}{\mathbb{R}}
\newcommand{\integers}{\mathbb{N}}
\newcommand{\eps}{\epsilon}
\newcommand{\normal}{\mathcal{N}}

\makeatletter
\newcommand{\opnorm}{\@ifstar\@opnorms\@opnorm}
\newcommand{\@opnorms}[1]{%
  \left|\mkern-1.5mu\left|\mkern-1.5mu\left|
   #1
  \right|\mkern-1.5mu\right|\mkern-1.5mu\right|
}
\newcommand{\@opnorm}[2][]{%
  \mathopen{#1|\mkern-1.5mu#1|\mkern-1.5mu#1|}
  #2
  \mathclose{#1|\mkern-1.5mu#1|\mkern-1.5mu#1|}
}
\makeatother

\let\originalleft\left
\let\originalright\right
\renewcommand{\left}{\mathopen{}\mathclose\bgroup\originalleft}
\renewcommand{\right}{\aftergroup\egroup\originalright}

\allowdisplaybreaks

\title{Information-Theoretic Limits for the\\ Matrix Tensor Product} 

\author{Galen Reeves%~\IEEEmembership{Member,~IEEE}
\thanks{G. Reeves is with the Department of Electrical and Computer Engineering and the Department of Statistical Science, Duke University, Durham, NC 27708 USA. This research was supported in part by the National Science Foundation under Grant 1750362.}
}

\begin{document}

%\clearpage

\maketitle

\begin{abstract}
This paper studies a high-dimensional inference problem involving the matrix tensor product of  random matrices. This problem generalizes a number of contemporary data science problems including the spiked matrix models used in sparse principal component analysis and covariance estimation and the stochastic block model used in network analysis. The main results are single-letter formulas  (i.e., analytical expressions that can be approximated numerically)  for the mutual information and the minimum mean-squared error (MMSE) in the Bayes optimal setting where the distributions of all random quantities are known. We provide non-asymptotic bounds and show that our formulas describe exactly the leading order terms in the mutual information and MMSE in the high-dimensional regime where the number of rows $n$ and number of columns $d$ scale with $d = O(n^\alpha)$ for some $\alpha < 1/20$. 

On the technical side, this paper introduces some new techniques for the analysis of  high-dimensional matrix-valued signals. Specific contributions include a novel extension of the adaptive interpolation method that uses order-preserving positive semidefinite interpolation paths, and a variance inequality between the overlap and the free energy that is based on continuous-time I-MMSE relations.
\end{abstract}

%\begin{IEEEkeywords}  Adaptive interpolation, continuous-time I-MMSE relations,  information-theoretic limits, low-rank matrix estimation, matrix factorization, minimum mean square error (MMSE), mutual information, spiked Wigner model,  spiked Wishart model,  stochastic block model.
%\end{IEEEkeywords}

\tableofcontents

\section{Introduction}

%1)  start with different models people care about. They have been analyzed in slightly different ways. 

%2) now there is a unified approach that can handle both of these models at the same time as well as a number of more general ones. It was not obvious that this was the case before. 

%3) convexity had played a key role but it is not necessary. 

%4) these insights not limited to this problem. 

% 5) move of theorem earlier to describe results

% 6) why is exact value cool. Because you can take derivatives and perturbations have operational meaning. 

% 7) something about Eric Lander. key idea is to understand reduction down a few numbers . 

% 8) one the mapping from linear to bilinear gives the general principles

%
%The basic problem is to disentangle the product of unknown random matrix from noisy observations. Two models have been studies widely. The first spiked Wishart model is given by
%\begin{align}
%\bY =  \sqrt{ \frac{\lambda}{n} } \bX \bX^\top + \bW
%\end{align}
%The asymetric version is 
%\begin{align}
%\bY =  \sqrt{ \frac{\lambda}{n} } \bU \bV^\top + \bW
%\end{align}
%Suppose that the rows of $\bX  = [X_1, \dots, X_n]^\top$ are i.i.d.\ from a distribution $P_0$ on $\reals^d$. The main result of BLA are the following limits: 
%\begin{align}
%\frac{1}{n} I(\bX; \bY) & \to  \cI^*\\
%\frac{1}{n} \ex{ \| \bX \bX^\top - \ex{  \bX \bX^\top \mid \bY}  \|_\fro^2}   & \to  \cM^*
%\end{align} 
%
%Differen setup 
%\begin{align}
%\mmse(\bU^\top \bU \mid \bY)   & \to  \cM_{U}^*\\
%\mmse(\bV^\top \bV \mid \bY)   & \to  \cM_{V}^*
%\end{align} 
%  
%  

%  
  
Inference problems involving the estimation and factorization of large structured matrices play a central role in the data sciences. The last decade has witnessed significant progress on theory and algorithms for a variety of models involving low-rank structure, such as the spiked matrix models used in low-rank covariance estimation\cite{johnstone:2001a,johnstone:2009},  sparse principal component analysis (PCA) \cite{zou:2006}, and clustering \cite{lesieur:2016,banks:2018}, as well as related models involving sparse graphical structures, such as the stochastic block model (SBM) for community detection \cite{decelle:2011a,abbe:2018}. 

To understand the fundamental (or information-theoretic) limits of inference a recent line of work \cite{korada:2010, lesieur:2015,reeves:2016a, barbier:2016, reeves:2019c,  barbier:2019,deshpande:2017, krzakala:2016,barbier:2016a, lelarge:2018, miolane:2017a,reeves:2017e, alaoui:2018, mourrat:2018,mourrat:2020,reeves:2019d,reeves:2019ab,reeves:2019a,mayya:2019,barbier:2019c,aubin:2019,barbier:2019e} has focused on exact characterizations of  the asymptotic mutual information and minimum-mean squared error (MMSE). This body of work focuses on problem settings where the probability distributions  of all unknown quantities are known and there are no constraints on computational complexity. The limits are described by single-letter formulas (i.e., analytical expressions that can be approximated numerically) that characterize the leading order terms in the mutual information and MMSE. 

The value of exact formulas is that they provide detailed information about the role of different model parameters (e.g., the amount or quality of data). Moreover, they can describe phase transitions that delineate between problem regimes with drastically different behaviors with respect to the quality of inference and computational complexity.  The development of this type of theory often follows a two-part sequence where the formulas are first conjectured using heuristic approaches, such as the replica method, and then proven rigorously using very different techniques \cite{mezard:2009,talagrand:2011,talagrand:2011a,lesieur:2015b,lesieur:2017}.  While the emphasis on single-letter formulas is standard in areas such as information theory and statistical physics, it differs from some of the other approaches used in the data sciences, which focus instead on order-optimal bounds or rates of convergence. 

In this context, the contribution of this paper is a rigorous analysis of the fundamental limits for a broad class of problems related to the matrix tensor product  (or Kronecker product) of large random matrices with low-dimensional structure. Specifically we introduce a model that generalizes existing models in the literature and provide single-letter formulas for the mutual information and MMSE. These formulas recover a number of existing results  as special cases and also provide new results for settings that could not be analyzed previously. 

On the technical side, a further contribution of this paper is the development of techniques for the analysis of high-dimensional matrix-valued signals. Our formulas leverage functional properties of entropy under Gaussian convolution, using ideas introduced in analysis of multi-layer networks~\cite{reeves:2017e} and, more recently, for matrix factorization problems motivated by community detection~\cite{reeves:2019ab,reeves:2019a,mayya:2019}. The proof builds upon the adaptive interpolation method  \cite{barbier:2018,barbier:2019f}, which represents an evolution of earlier interpolation techniques in the statistical physics literature \cite{guerra2002,guerra:2003a}. In this direction, one of the key steps in this paper is the introduction of matrix-valued interpolation paths that satisfy a certain order-preserving property that is easier to work with than the regularity conditions used in previous work. Further innovations include  a new method for controlling overlap concentration using perturbations parameterized by positive semidefinite matrices, and a variance inequality between overlap and the free energy that is based on continuous-time I-MMSE relationships~\cite{weissman:2010, venkat:2012}.

\subsection{The matrix tensor product and related models}

This paper studies a model consisting of the matrix tensor product (or Kronecker product)  observed in additive Gaussian noise with an arbitrary coupling matrix. Throughout this paper, we use $\bX$ to denote a real $n \times d$ random signal matrix with known distribution. 

\begin{definition} 
The matrix tensor product model with coupling matrix $ B \in \reals^{d_1 d_2 \times m}$ is given by 
\begin{align}
\bY  = (\bX_1 \otimes \bX_2) B + \bW  \label{eq:mtp_model}
\end{align}
where $(\bX_1, \bX_2) \in \reals^{n_1 \times d_2} \times \reals^{n_2 \times d_2}$ is a pair of jointly random signal matrices, $ (\bX_1 \otimes \bX_2)$ denotes the $n_1 n_2 \times d_1 d_2$ matrix obtained by the Kronecker product,  and $\bW \in \reals^{n_1n_2 \times m}$ is a noise matrix that has i.i.d.\ standard Gaussian entries. Equivalently, the rows of $\bY$ are in one-to-one correspondence with the set of $m$-dimensional observations $\{Y_{ij} \, : \, 1 \le i \le n_1, 1 \le j \le n_2 \}$ given by
\begin{align}
Y_{ij}  =   ( X_{1i} \otimes X_{2j} ) B   + W_{ij}, 
\end{align}
where $X_{1i}$ and $X_{2j}$ denote the $i$-th and $j$-th rows of $\bX_1$ and $\bX_2$ and each $W_{ij}$ is an independent $m$-dimensional standard Gaussian vector. 
\end{definition}

The matrix tensor model defined above is general in the sense that it allows for any linear combination of the rows of the tensor product. The spiked matrix models used for sparse PCA and covariance estimation can be represented by this model with specific choice of the coupling matrix $B$. To see this correspondence, suppose that $d_1 = d_2 = d$ and let  $B$ be the $d^2 \times 1$ vector obtained by stacking the columns of $(s/n) I_d$ with $n = \max(n_1,n_2)$ and $s> 0$. Then, \eqref{eq:mtp_model} is equivalent to the asymmetric low-rank matrix estimation model (or spiked Wishart model) given by
\begin{align}
\bY = \sqrt{\frac{s}{n}} \bX_1 \bX_2^\top  + \bW, \label{eq:spiked_wishart}
\end{align}
where $s$ parameterizes the signal-to-noise ratio and $\bW$ is an $n_1 \times n_2$ noise matrix with i.i.d.\ standard Gaussian entries. The case where $\bX_2= \bX_1 =\bX$ recovers the symmetric spiked-matrix model (or spiked Wigner model), which is given by
\begin{align}
\bY = \sqrt{\frac{s}{n}} \bX \bX^\top  + \bW. \label{eq:spiked_wigner}
\end{align}

A generalization of the usual spiked matrix model was introduced to study the  $k$-community degree-balanced SBM \cite{reeves:2019ab,reeves:2019a}. In the setting of diverging average degree, the community detection problem can be modeled using
\begin{align}
\bY =  \frac{1}{\sqrt{n}}  \bX  B \bX^\top  + \bW, \label{eq:spiked_matrix}
\end{align}
where $B$ is a $d \times d$ coupling matrix $(d = k-1)$ that models the interactions between different communities. The case where  $B$ is positive semidefinite models networks with assortative behaviors (i.e., individuals in the same community are more likely to form an edge) and the case where $B$ is negative semidefinite models networks with dissassortative behaviors (i.e., individuals in different communities are more likely to form an edge). The ability to model networks with both assortative and dissassortative requires a coupling matrix with positive and negative values. 

A further generalization of \eqref{eq:spiked_matrix} is the multiview spiked matrix model used to study community detection with multiple correlated networks \cite{mayya:2019}. This model is given by 
\begin{align}
\bY_l =  \frac{1}{\sqrt{n}}  \bX  B_l \bX^\top  + \bW_l, \quad l=1,\dots ,L, \label{eq:multiview}
\end{align}
where each $B_\ell$ is associated with a different network. By vectorization, this model is an instance of \eqref{eq:mtp_model} with $\bX_1 = \bX_2 = \bX$ and $B = \frac{1}{\sqrt{n}} [\gvec(B_1), \dots, \gvec(B_L)].$

The main result of this paper is an approximation formula for the mutual information and MMSE associated with the matrix tensor product model. Rather than focusing on \eqref{eq:mtp_model} directly we consider a symmetric form that is easier to work with. 

\begin{definition}\label{def:sym_mtp_model} The symmetric matrix tensor product model with  $d^2 \times d^2$ positive semidefinite matrix $S$ is given by
\begin{align}
\bY = \frac{1}{\sqrt{n} }  \bX^{\otimes 2} S^{1/2}  + \bW, \label{eq:sym_mtp_model}
\end{align}
where $\bX^{\otimes 2} = \bX \otimes \bX$ is $n^2 \times d^2$ and  $\bW$ is a noise matrix with i.i.d.\ standard Gaussian entries. 
\end{definition} 

In fact, there is no loss in generality in restricting our attention to the symmetric case.

\begin{prop}\label{lem:model_equiv}
The matrix tensor product model \eqref{eq:mtp_model} and the symmetric form \eqref{eq:sym_mtp_model} are equivalent in the sense that either one can be used to represent the other. 
\end{prop} 
\begin{proof}
The mapping from \eqref{eq:sym_mtp_model}  to \eqref{eq:mtp_model} is obvious. The reverse direction is based on two steps. First, any pair $(\bX_1, \bX_2)$ can be embedded into a single $(n_1 + n_2) \times (d_1 + d_2)$  block-diagonal matrix $\bX$ via the matrix direct product $\bX_1 \oplus \bX_2 = \diag(\bX_1, \bX_2)$. The tensor product $\bX \otimes \bX$ then includes $\bX_1 \otimes \bX_2$ as  subset of the columns and rows. Second, by orthogonal invariance of the standard Gaussian distribution, the projected observation $\bY B^{\top}$ is a sufficient statistic for inference about $\bX$, and thus there is no loss in generality in considering only positive semidefinite coupling matrices. 
\end{proof}

\subsection{Overview of contributions}

This paper is the first to consider the general setting of the matrix tensor product \eqref{eq:mtp_model}. Some special cases of our general result that could not be addressed previously are the following: %include the following settings:
\begin{itemize}
\item The spiked Wishart model \eqref{eq:spiked_wishart} where there is dependence between the factors $\bX_1$ and $\bX_2$. %with an arbitrary distribution on the matrices $(\bX_1, \bX_2)$. 
%Note that this includes the spiked Wigner model as a special case. 

\item The spiked matrix model \eqref{eq:spiked_matrix} and the multiview version \eqref{eq:multiview} with arbitrary coupling matrices. 

\item The stochastic block model with an indefinite coupling matrix $B$ that models both assortative and dissasortative behaviors in the same network. 

\end{itemize}

The main results are theoretical guarantees for optimal estimation in the matrix tensor product model. We derive a single-letter approximation formula for the mutual information and provide a rigorous non-asymptotic bound on the approximation error (Theorem~\ref{thm:general_bound}). Similar in spirit to recent work focusing on spiked matrix models with generative priors~\cite{aubin:2019}, this result is stated under general conditions that depend only on the variance of the free energy. Combining this result with a more common (and more restrictive) independence assumption leads to a bound (Theorem~\ref{thm:MI_bound}) in which the error is proportional to $d^4\,  n^{-1/5}$, and thus converges to zero whenever $d = O(n^\alpha)$ for some $\alpha < 1/20$ . Although it is unlikely that this is the optimal rate of convergence it is, to the best of our knowledge, the strongest rate that has been proven in the literature. Some implications of these results for the MMSE are given in Theorem~\ref{thm:MMSE_limit}.

Our proof builds upon the method of adaptive interpolation used in previous work on spiked matrix models~\cite{barbier:2019c, barbier:2018,barbier:2019f,barbier:2019e}  and shares the same basic outline of  1) specifying an interpolation path via a first-order ordinary differential equation and 2) introducing a perturbation to the problem to guarantee  overlap concentration (almost everywhere) and thus cancel a remainder term. The main novelty in our approach is that we are able to perform these steps in a more general setting that allows us to address the matrix tensor product model. 

Two new ideas are the following: 
\begin{itemize}

\item \textbf{Order-preserving  interpolation path:} The interpolation paths are parameterized by positive semidefinite matrices and satisfy an order-preserving property. This property can be verified thanks to a comparison inequality for  differential equations (Lemma~\ref{lem:ODE_comp}).  As a consequence, matrix-valued perturbations  can be introduced directly along the interpolation path.

\item \textbf{Continuous time variance inequality:}  Inspired by pointwise I-MMSE relationships \cite{weissman:2010, venkat:2012} we introduce a continuous time coupling of the noise in the perturbation model. Leveraging the power of Ito  calculus, we derive a variance inequality (Theorem~\ref{thm:var_inq}) that provides a direct link between the variance of the free energy and the variation in the overlap matrix. This approach simplifies certain aspects of the analysis and allows for concentration under weaker assumptions on the scaling of the problem dimension. 
\end{itemize}

We expect that many of the techniques introduced in this paper can also be applied to higher-order tensor products.

\subsection{Related work}

A large body of work has focused on the information-theoretic \cite{barbier:2019,deshpande:2017, krzakala:2016,barbier:2016a, lelarge:2018, miolane:2017a,reeves:2017e, alaoui:2018, mourrat:2018,mourrat:2020,reeves:2019d,reeves:2019ab,reeves:2019a,mayya:2019,barbier:2019c,aubin:2019,barbier:2019e}  and algorithmic \cite{baik:2005,amini:2009,fletcher:2012b, fletcher:2018,deshpande:2014, deshpande:2014b,perry:2018} limits of recovery in spiked matrix models. Two special cases are now well understood:  the spiked Wigner model \eqref{eq:spiked_wigner} with i.i.d.\ rows and the spiked Wishart model \eqref{eq:spiked_wishart} in the setting where $\bX_1$ and $\bX_2$ are independent matrices with i.i.d.\ rows. Using the replica method, Lesieur et al. \cite{lesieur:2015,lesieur:2016, lesieur:2017} derived conjectured formulas for the asymptotic mutual information and MMSE. The case of the rank-one $(d = 1)$  spiked Wigner model was proven rigorously by  Krzakala et al.~\cite{krzakala:2016} and Barbier et al.~\cite{barbier:2016a}. The  case of low-rank matrices $(d \ge 1)$ was proven for the spiked Wigner model by Lelarge and  Miolane~\cite{lelarge:2018}  and for the spiked  Wishart model by Miolane~\cite{miolane:2017a}.

The role of the coupling matrix in the symmetric model \eqref{eq:spiked_matrix}  was investigated by Reeves et al.~\cite{reeves:2019ab,reeves:2019a}, who provided one-sided bounds on the formulas for the mutual information and MMSE. Subsequent work by Mayya and Reeves~\cite{mayya:2019} extended this approach to the multiview version \eqref{eq:multiview}. Concurrent to the work in the present paper, Barbier and Reeves~\cite{barbier:2020a}  provide an exact characterization for a special case of the multiview model \eqref{eq:multiview} that satisfies an additional convexity property. Specifically it is assumed that $B_1, \dots, B_L$ are such that
\begin{align}
\sum_{\ell=1}^L \left\{ B_l \otimes B_l)  + (B_l \otimes B_l)^\top \right\} \succeq 0 .
\end{align}
This assumption is restrictive in the sense that  it precludes the model equivalence in Lemma~\ref{lem:model_equiv}; see Section~\ref{sec:alternate_formulas} for more discussion about the role of this assumption.

The approach in this paper builds upon the method of adaptive interpolation, which has been applied primarily in the case of rank-one estimation problems~\cite{barbier:2019c, barbier:2018,barbier:2019f} and more recently to the setting of low-rank even order tensors \cite{barbier:2019e}. Some different approaches include the cavity method \cite{lelarge:2018, miolane:2017a}, large deviations techniques \cite{alaoui:2018,aubin:2019}, and the analysis of Hamilton-{J}acobi equations \cite{mourrat:2018,mourrat:2020}.

\subsection{Notation} 
We use  $\mathbb{S}^d$, $\psd^d$ and $\mathbb{S}_{++}^d$ to denote the space of $d \times d$ symmetric, positive semidefinite, and positive definite matrices, respectively.  For matrices  $A \in \reals^{m \times n}$ and $B \in \reals^{m \times n}$ the trace inner product is given by $\langle A , B \rangle : = \gtr(A^\top B)$. Unless specified otherwise, $\|A\| = \sigma_\text{max}(A)$ denotes is spectral norm and $\|A\|_\fro = \sqrt{ \langle A, A \rangle}$ denotes the Frobenius norm. For  $A \in \reals^{m \times n}$, we use $\gvec(A)$ to denote the $mn \times 1$ vector obtained by stacking the columns in $A$.  The Kronecker product (or tensor product) is denoted by $\otimes$ and the Kronecker power is denoted by $A^{\otimes k}  = \otimes_{\ell = 1}^k A $.    

We use upper case letters such as $X$ and $Y$ to denote random variables and random vectors and upper case bold letters such as $\bX$ and $\bY$ to denote random matrices. The probability measure of jointly random variables $(X, Y)$ is denoted by $P_{X,Y}$,   the marginals are denoted by $P_{X}$ and $P_{Y}$, and the product measure of the marginals is $P_X \otimes P_Y$. The Gaussian distribution of mean $\mu$ and covariance $\Sigma$ is denoted by $\normal(\mu,\Sigma)$ and a random matrix is a called a standard Gaussian matrix if its entries are i.i.d.\  $\normal(0, 1)$ variables.

\section{Statement of main results} 

We assume throughout that $\bX$ is an $n \times d$ random matrix with finite fourth moments.  For a pair of positive semidefinite matrices $(R, S) \in \psd^d \times \psd^{d^2}$, observations  $\bY_{R,S}$ are generated according to the model
\begin{align}
\bY_{R,S}  = \begin{dcases} 
\bX R^{1/2} + \bW\\
\frac{1}{\sqrt{n}} \bX^{\otimes 2}  S^{1/2} + \bW'
\end{dcases} \label{eq:Y_RS}
\end{align}
where $\bW$ and $\bW'$ are independent standard Gaussian matrices. The special case  $\bY_{R,0}$ is the linear model with matrix-valued input  and the special case $\bY_{0, S}$ is the matrix tensor product model given in Definition~\ref{def:sym_mtp_model}.  We also allow for the additional observations called ``side-information'' which are represented by a random variable $\bZ$ taking values in a set $\cZ$.  It is assumed throughout that $\bY_{R,S}$ and $\bZ$ are conditionally independent given $\bX$ and that the joint distribution of all random variables is known. 

%
%\begin{figure*}[!b]
%% ensure that we have normalsize text
%\normalsize
%% Store the current equation number.
%\setcounter{MYtempeqncnt}{\value{equation}}
%% Set the equation number to one less than the one
%% desired for the first equation here.
%% The value here will have to changed if equations
%% are added or removed prior to the place these
%% equations are referenced in the main text.
%\hrulefill
%\setcounter{equation}{11}
%\begin{align}\label{eq:Ih_RS}
%\hat{\cI}_n(R, S)  
%&\coloneqq    \adjustlimits   \min_{Q \in \cQ} \sup_{\tilde{R} \in \psd^d} \left\{ \cI_n(R+ \tilde{R})  + \frac{1}{2}  \langle S, \ex{ (  \tfrac{1}{n} \bX^\top \bX)^{\otimes 2}}  - Q^{\otimes 2} \rangle -  \frac{1}{2} \langle \tilde{R},  \ex{   \tfrac{1}{n} \bX^\top \bX}  -  Q \rangle   \right\}, 
%\end{align}
%\setcounter{equation}{13}
%\begin{align}
%\hat{\cI}_0(R, S) 
%&\coloneqq    \adjustlimits   \min_{Q \in  \cQ} \sup_{S \in \psd^d} \left\{ \cI_0(R + \tilde{R})   + \frac{1}{2}  \langle S, \ex{X_0 X_0^\top}^{\otimes 2}  - Q^{\otimes 2} \rangle -  \frac{1}{2} \langle \tilde{R}, \ex{ X_0 X_0^\top} -  Q \rangle   \right\}, \label{eq:I0} 
%\end{align}
%% Restore the current equation number.
%\setcounter{equation}{\value{MYtempeqncnt}}
%% IEEE uses as a separator
%%\hrulefill
%% The spacer can be tweaked to stop underfull vboxes.
%%\vspace*{4pt}
%\end{figure*}
%

Our  main results describe the fundamental limits of estimating $\bX$ from the observations  $(\bY_{R,S}, \bZ)$. Specifically, we describe exactly the leading order terms in the mutual information and as well as the MMSE matrix.

\subsection{Mutual information} \label{sec:mutual_information}

 The primary object of interest is the normalized mutual information function $\cI_n \colon \psd^d \times \psd^{d^2} \to \reals$ defined by
%Our first object of interest is the dimension normalized mutual information: 
\begin{align}
\cI_n(R,S) %& \coloneqq  \frac{1}{n} I(\bX ; \bY_{R, S})  = \frac{1}{n} \langle R,  \ex{ \bX^\top \bX} \rangle  + \frac{ 1}{n^2} \langle S, \ex{ (\bX^{\otimes 2})^\top (\bX^{\otimes 2}) }  \rangle  - \frac{1}{n} D(R, \frac{1}{n} S) \\
& \coloneqq  \frac{1}{n} I(\bX ; \bY_{R, S} \mid \bZ) , \label{eq:I_RS}
% = \frac{1}{n} \langle R,  \ex{ \bX^\top \bX} \rangle  + \frac{ 1}{n^2} \langle S, \ex{ (\bX^\top \bX)^{\otimes 2}} \rangle  - \frac{1}{n} D(R, \frac{1}{n} S) 
\end{align}
where $I(\bX; \bY_{R,S} \mid \bZ)$ is the conditional mutual information between $\bX$ and $\bY_{R,S}$ given $\bZ$. 
This function has a number of important properties. It is concave \cite{lamarca:2009} and non-decreasing with respect to the Loewner order \cite{reeves:2018a}. Furthermore, it is infinitely differentiable on the interior of its domain and its gradients are related to the mean-squared error in estimating both $\bX$ and $\bX^{\otimes 2}$ from $(\bY_{R,S}, \bZ)$  \cite{guo:2005a, palomar:2006, payaro:2009,guo:2011}.  These relationships are described in more detail in Section~\ref{sec:MMSE_matrix}.

For the special case of the linear model $(S = 0)$ we use the short-hand notation $\cI_n(R) \coloneqq \cI_n(R,0)$. The linear model has been studied extensively in the context of the linear Gaussian channel in information theory. In many cases the leading order terms in mutual information can be described exactly in terms of so-called ``single-letter" or ``multi-letter'' formulas,  which are analytic expressions involving low-dimensional integrals. This is particularly true when the rows of $\bX$ are independent and thus the estimation problem decouples into $n$ independent problems, each of dimension $d$.

\begin{example}
Suppose that $\bX = (X_1, \dots, X_n)^\top$ and $\bZ = (Z_1, \dots, Z_n)^\top $ where $(X_i, Z_i)_i$ are i.i.d.\ according to a distribution $P_0$.  Then, the  mutual information in the linear model is given by the single-letter formula 
 \begin{align}
 \cI_{n}(R) & \coloneqq I(X_0 ; R^{1/2} X_0 + W_0 \mid Z_0 ), %\\
% \MMSE( \bX \mid \bY_{R,S} ) & =  \MMSE( X_0 \mid R^{1/2} X_0 + W_0 \mid Z_0 )
 \end{align}
 where $W_0 \sim \normal(0, I_d)$ and $(X_0, Z_0 ) \sim P_0$ are independent. Under  regularity conditions on $P_0$, this expression can be approximated numerically with complexity that scales exponentially in $d$. 
 \end{example}

In comparison to the linear model, the general setting of nonzero $S$ is much more challenging. Part of the difficulty is that the tensor product introduces non-trivial dependencies  between the rows of the observations. Another issue is that  there can exist non-identifiabilities when the distribution of $\bX$ has certain invariances (e.g., invariance to a change in sign). 

The high-level idea behind our approach is that the mutual information associated with an arbitrary pair $(R,S)$ can be related to the linear model evaluated at $(R +\tilde{R} , 0)$ where  $\tilde{R}$  is a positive semidefinite matrix that depends on $(R,S)$ as well as the distribution of $(\bX, \bZ)$. Based on this comparison, the complexity of analyzing the tensor product model is then reduced to the complexity of the linear model. The nature of this comparison is made precise in the following definition.

\begin{definition}\label{def:Ih} The ``approximation formula'' for the normalized mutual information $\cI_{n}(R,S)$ is given by % \eqref{eq:Ih_RS} 
\begin{align}
\label{eq:Ih_RS}
\hat{\cI}_n(R, S)  
&\coloneqq    \adjustlimits   \min_{Q \in \cQ} \sup_{\tilde{R} \in \psd^d} \left\{ \cI_n(R+ \tilde{R})  + \frac{1}{2}  \langle S, \ex{ (  \tfrac{1}{n} \bX^\top \bX)^{\otimes 2}}  - Q^{\otimes 2} \rangle -  \frac{1}{2} \langle \tilde{R},  \ex{   \tfrac{1}{n} \bX^\top \bX}  -  Q \rangle   \right\}, 
\end{align}
where $\cQ \coloneqq \{ Q \in \psd^d \, : \, Q \preceq \ex{\frac{1}{n}  \bX^\top \bX}\}$. 
This function depends only on the mutual information associated with the linear model. 
%\stepcounter{equation}
\end{definition}

\begin{remark}
If $S= 0$, then it follows from the  concavity of the mapping $\tilde{R} \mapsto \cI_n(R+ \tilde{R})$ that the approximation formula is exact, i.e.,   $\hat{\cI}_n(R) = \cI_n(R)$ for all $R$.
\end{remark}

\begin{remark}
The approximation formula $\hat{\cI}_n$ has many of the same function properties as $\cI_n$, including concavity and monotonicity with respect to Loewner order. One important difference however, is that the gradients can have discontinuities at points where the extremizers are non-unique. These discontinuities correspond to phase transitions that occur in the limit of high dimension. 
\end{remark}

 Our first main result is a non-asymptotic bound on the difference between the mutual information and the approximation formula.  The proof, which is given in Section~\ref{proof:thm:MI_bound}, follows from a more general result given in  Section~\ref{sec:general_bounds}.

\begin{theorem}\label{thm:MI_bound} Suppose that  $\bX = (X_1, \dots, X_n)^\top$  and $\bZ = (Z_1, \dots, Z_n)$  where  $ (X_i, Z_i)_i$ are independent and the rows satisfy the bound $\|X_i\| \le \sqrt{d}\, \rho$ almost surely for some $\rho \ge 1$. There exists a universal constant $C$ such that
\begin{align}
 | \cI_n(R,S)   - \hat{\cI}_n(R,S)   |  \le   \frac{  C \left( d^2  + d^4 \rho^4 \|S\|  \right) }{ n^{1/5} }.
\end{align}
\end{theorem}

Theorem~\ref{thm:MI_bound} shows that the approximation formula describes exactly the leading order terms in the mutual information in the high-dimensional setting where  dimensions $n$ and $d$ scale with  $d = O(n^{\alpha})$ for some $\alpha <1/20$. We remark that the non-asymptotic nature of this bound is significantly different from typical results in the literature, which are usually asymptotic or have constants that depend implicitly on the underlying parameters.

Focusing on the setting where $d$ is fixed and the rows and side information are identically distributed provides an exact characterization of the limit.

\begin{cor}\label{cor:MI_limit} 
Suppose that  $\bX = (X_1, \dots, X_n)^\top$  and $\bZ = (Z_1, \dots, Z_n)$  where  $(X_i, Z_i)_i$  are i.i.d.\ according to a distribution $P_0$ whose $X$-marginal has finite fourth moments. Then,  $\cI_n$ converges pointwise to the limit $\hat{\cI}_0$ given by % \eqref{eq:I0} \stepcounter{equation} 
\begin{align}
\hat{\cI}_0(R, S) 
&\coloneqq    \adjustlimits   \min_{Q \in  \cQ} \sup_{S \in \psd^d} \left\{ \cI_0(R + \tilde{R})   + \frac{1}{2}  \langle S, \ex{X_0 X_0^\top}^{\otimes 2}  - Q^{\otimes 2} \rangle -  \frac{1}{2} \langle \tilde{R}, \ex{ X_0 X_0^\top} -  Q \rangle   \right\}, \label{eq:I0} 
\end{align}
where $\cQ = \{ Q \in \psd^d \, : \, Q \preceq \ex{X_0 X_0^\top}\}$ and  $ \cI_{0}(R)  \coloneqq I(X_0 ; R^{1/2} X_0 + W_0 \mid Z_0 )$ with  $W_0 \sim \normal(0, I_d)$ independent of  $(X_0, Z_0) \sim P_0$.   
\end{cor}
\begin{proof}
If the support of $X_0$ is bounded then this follows immediately from Theorem~\ref{thm:MI_bound}. Otherwise, one can first establish the limit for an approximation to $P_0$ under which $X_0$ has bounded support and then use the Lipschitz continuity of the mapping $P_{\bX, \bX^{\otimes2} } \mapsto I(\bX ; \bY_{R,S} \mid \bZ)$ with respect to the quadratic Wasserstein distance \cite{polyanskiy:2016} to show the same limit holds for $P_0$. 
\end{proof}

Corollary~\ref{cor:MI_limit} applies generally for an arbitrary pair $(R,S)$ and any source of side information satisfying the independence constraints. 
Applied to the special case of the low-rank spiked matrix models where $S$ is rank one,  $R = 0$ and there is no side information, Corollary~\ref{cor:MI_limit} recovers the results given in~\cite{krzakala:2016, barbier:2016a, lelarge:2018, miolane:2017a}. Some other special cases of Corollary~\ref{cor:MI_limit}  have also been established in recent work focusing on a multiview observation model where the coupling matrices satisfy a coupling constraint~\cite{reeves:2019ab,reeves:2019a,barbier:2020a}.

\subsection{MMSE matrix} \label{sec:MMSE_matrix}

Having established the exact behavior of the leading order terms in the mutual information we now turn our attention to estimation-theoretic quantities. Consider the problem of estimating a matrix-valued random variable $\bX$ based on observations $\bZ$.  One natural measure of performance is given by the MMSE:
\begin{align}
\mmse(\bX \mid \bZ) &\coloneqq \ex{ \| \bX - \ex{ \bX \mid \bZ} \|_\fro^2}.
 \label{eq:MMSE_scalar}
\end{align} 
Note the expectation is taken with respect to the joint distribution of $(\bX, \bZ)$ and so this is a non-random quantity. Because the Frobenius norm is an entrywise norm,  the MMSE does not change if we replace $\bX$ by its transpose or by the $nd \times 1$ vector $\gvec(\bX)$ obtained by stacking the columns. 

This paper focuses on the MMSE matrix introduced in \cite{reeves:2018a, reeves:2019a,reeves:2019ab}, which is given by the $d \times d$ positive semidefinite matrix
\begin{align}
\MMSE( \bX \mid \bZ) \coloneqq \ex{ (\bX - \ex{\bX\mid \bZ} )^\top (\bX - \ex{ \bX \mid \bZ} ) }.
\end{align}
This matrix provides a multivariate generalization of \eqref{eq:MMSE_scalar}.  The diagonal entries correspond to the MMSE of the columns of $\bX$ and  thus $\gtr(\MMSE( \bX \mid \bZ) )$ is equal to \eqref{eq:MMSE_scalar}. The off-diagonal entries provide information about the correlation between the rows of $\bX$.  Note that the dimensions of the matrix are treated differently. For example,  $\MMSE( \bX^\top \mid \bZ)$ is $n \times n$ and $\MMSE( \gvec(\bX)^\top \mid \bZ)$ is $nd \times nd$.

%
%\begin{figure*}[!b]
%% ensure that we have normalsize text
%\normalsize
%% Store the current equation number.
%\setcounter{MYtempeqncnt}{\value{equation}}
%% Set the equation number to one less than the one
%% desired for the first equation here.
%% The value here will have to changed if equations
%% are added or removed prior to the place these
%% equations are referenced in the main text.
%\hrulefill
%\setcounter{equation}{18}
%\begin{subequations}
%\label{eq:Ih_RS_gradient}
%\begin{align}
%\partial_{R} \hat{\cI}_n(R,S) &= \left\{\frac{1}{2}  U \in \sym_+^d \, : \,   \hat{\cI}_n( \tilde{R},S)  - \hat{\cI}_n(R,S) \le \frac{1}{2} \langle \tilde{R} - R  , U \rangle ,  \text{ for all $\tilde{R} \in \psd^{d}$}  \right \} \\
%\partial_{S} \hat{\cI}_n(R,S) & =\left\{\frac{1}{2}  V \in \sym_+^{d^2} \, : \,   \hat{\cI}_n(R, \tilde{S} )  - \hat{\cI}_n(R,S) \le \frac{1}{2} \langle \tilde{S} - S, V \rangle ,  \text{ for all $\tilde{S} \in \psd^{d^2}$} \right \}
%%\frac{1}{2n^2} \MMSE(\bX^{\otimes 2}  \mid \bY_{R, S}) .
%\end{align}
%\end{subequations}
%% Restore the current equation number.
%\setcounter{equation}{\value{MYtempeqncnt}}
%% IEEE uses as a separator
%%\hrulefill
%% The spacer can be tweaked to stop underfull vboxes.
%%\vspace*{4pt}
%\end{figure*}

The I-MMSE relationship \cite{guo:2005a} and its multivariate extensions \cite{palomar:2006, lamarca:2009,payaro:2009,guo:2011,reeves:2018a} provide a link between the gradient of the mutual information in a Gaussian linear model and the MMSE. As a consequence of the matrix version of the I-MMSE relation, 
\begin{subequations}
\label{eq:I_RS_gradient}
\begin{align}
\nabla_{R} \cI_n(R,S) &= \frac{1}{2n} \MMSE(\bX \mid \bY_{R, S} , \bZ ) \label{eq:I_R_gradient} \\
\nabla_{S} \cI_n(R,S) & = \frac{1}{2n^2} \MMSE(\bX^{\otimes 2}  \mid \bY_{R, S}, \bZ ) . \label{eq:I_S_gradient} 
\end{align}
\end{subequations}
Because $\cI_n$ is defined on the space of positive definite matrices, these gradients should be regarded  as linear mappings from the space of symmetric matrices to the reals. Also, these gradients are defined on the interior of the domain of  $\cI_{n}$, which is the space of positive definite matrices. For matrices that lie on the boundary, the gradient can be defined unambiguously via the limit of a sequence of positive definite matrices $\{(R_k, S_k)\}_{k \in \integers}$  that converges to $(R,S)$. For example, the MMSE matrix in the matrix tensor product model defined in Definition~\ref{def:sym_mtp_model}  is given by the formula:
\begin{align}
\frac{1}{n^2} \MMSE(\bX^{\otimes 2}  \mid \bY_{0, S} ,  \bZ )  & =  2 \nabla_{R} \cI_n(0,S)  \notag \\
& = \lim_{\eps \downarrow 0}  2 \nabla_{R} \cI_n(\eps I,S)  .  
\label{eq:MMSE_R0}
\end{align}

In view of the I-MMSE relations in \eqref{eq:I_RS_gradient}, the basic idea of our approach is to approximate the MMSE matrix  using the gradient of the  approximation formula $\hat{\cI}_n$ given in \eqref{eq:Ih_RS}. There are however some technical details that need to be addressed. For one,  the approximation formula is not continuously differentiable. A generalization of the derivative,  called the superdifferential,  is given by the set-valued mappings % $\partial_{R} \hat{\cI}_n(R,S) $ and $\partial_{S} \hat{\cI}_n(R,S)$ defined in \eqref{eq:Ih_RS_gradient}.  \stepcounter{equation}
\begin{subequations}
\label{eq:Ih_RS_gradient}
\begin{align}
\partial_{R} \hat{\cI}_n(R,S) &= \left\{\frac{1}{2}  U \in \sym_+^d \, : \,   \hat{\cI}_n( \tilde{R},S)  - \hat{\cI}_n(R,S) \le \frac{1}{2} \langle \tilde{R} - R  , U \rangle ,  \text{ for all $\tilde{R} \in \psd^{d}$}  \right \} \\
\partial_{S} \hat{\cI}_n(R,S) & =\left\{\frac{1}{2}  V \in \sym_+^{d^2} \, : \,   \hat{\cI}_n(R, \tilde{S} )  - \hat{\cI}_n(R,S) \le \frac{1}{2} \langle \tilde{S} - S, V \rangle ,  \text{ for all $\tilde{S} \in \psd^{d^2}$} \right \}.
%\frac{1}{2n^2} \MMSE(\bX^{\otimes 2}  \mid \bY_{R, S}) .
\end{align}
\end{subequations}
These are closed convex sets that are single-valued at every point where  $\hat{\cI}_n$ is differentiable.  The fact that  the superdifferentials are defined with respect to the space of positive semidefinite matrices, as opposed to the set of symmetric matrices, is justified by the fact that $\hat{\cI}_n$ is non-decreasing with respect to the Loewner order.

With these definitions in hand,  we can can  now provide non-asymptotic bounds on the MMSE matrix. 

\begin{prop}\label{prop:MMSE_bound}
Consider any pair $(R, S) \in \psd^d \times \psd^{d^2}$. 
%\begin{enumerate}[(i)]
%\item 
For all $\tilde{R} \in \sym^d$ such that $R + \tilde{R} \succeq 0$, 
\begin{align*}
 \frac{1}{n}   \langle   \tilde{R} , \MMSE( \bX \mid \bY_{R,S} ,  \bZ) \rangle &\ge \max_{U \in 2\partial \hat{\cI}_n(R  +  \tilde{R} ,S)} \langle \tilde{R} , U\rangle -  \delta_n,%\\
%\frac{1}{n}   \langle   \tilde{R} , \MMSE( \bX \mid \bY_{R + \tilde{R} ,S} ) \rangle  & \le  \min_{U \in 2\partial_R \hat{\cI}_n(R ,S)} \langle \tilde{R} , U\rangle +  \delta_n
\end{align*}
where $ \delta_n =  2 \sum_{ \xi \in \{0,1 \}}  | \cI_n(R +\xi \tilde{R} , S) - \hat{\cI}_n(R + \xi \tilde{R} , S) |$. 
%\item 
Furthermore, for all $\tilde{S} \in \sym^{d^2}$ such that $S + \tilde{S} \succeq 0$, 
\begin{align*}
\frac{1}{n^2}   \langle   \tilde{S} , \MMSE( \bX^{\otimes 2}  \mid \bY_{R,S}  , \bZ) \rangle  & \ge \min_{  V \in 2 \partial_S  \hat{\cI}_n(R, S+ \tilde{S}) } \langle  \tilde{S}, V  \rangle  -  \delta'_n,%\\
%\frac{1}{n}   \langle   \tilde{S} , \MMSE( \bX \mid \bY_{R  ,S + \tilde{S} } ) \rangle  & \le  \min_{V \in 2\partial_S \hat{\cI}_n(R ,S)} \langle \tilde{S} , V\rangle +  \delta_n
\end{align*}
where $\delta'_n = 2 \sum_{\xi \in \{0,1 \}}  | \cI_n(R, S  + \xi \tilde{S}  ) - \hat{\cI}(R, S+\xi  \tilde{S})|$. 
%\end{enumerate}
\end{prop} 
\begin{proof}
%We prove this result for the gradients with respect to $R$. The corresponding result for the gradients with respect to $S$ follows from the same arguments.
 By concavity of  $\cI_n$ and the I-MMSE relations in \eqref{eq:I_R_gradient},
\begin{align*}
\langle  \tilde{R} , \MMSE(\bX \mid \bY_{R+ \tilde{R} , S} )  \rangle   
 &  \le  2n  \cI_n(R +\tilde{R}, S) -  2n \cI_n(R,S)  & \le  \langle  \tilde{R} ,\MMSE(\bX \mid \bY_{R}, S )  \rangle.
\end{align*}
Similarly, by concavity of $\hat{\cI}_n$ it follows that for all $U \in  2 \partial_R \hat{\cI}_n(R, S) $ and $\tilde{U} \in  2 \partial_R  \hat{\cI}_n(R+ \tilde{R} , S) $, 
\begin{align*}
 \langle  \tilde{R} , \tilde{U}  \rangle \le  2 \hat{\cI}_n(R +\tilde{R}, S) - 2 \hat{\cI}_n(R,S)  \le   \langle  \tilde{R} , U \rangle.
\end{align*}
Comparing the above displays leads to $ \langle  \tilde{R} ,\MMSE(\bX \mid \bY_{R}, S )  \rangle \ge \langle  \tilde{R} , \tilde{U}  \rangle  - \delta_n$ and $\langle  \tilde{R} ,\MMSE(\bX \mid \bY_{R+ \tilde{R} }, S )  \rangle   \le \langle  \tilde{R} , U  \rangle  +  \delta_n$. 
%\begin{align}
% \langle  \tilde{R} ,\MMSE(\bX \mid \bY_{R}, S )  \rangle &\ge \langle  \tilde{R} , \tilde{U}  \rangle  - \delta_n\\
% \langle  \tilde{R} ,\MMSE(\bX \mid \bY_{R+ \tilde{R} }, S )  \rangle  & \le \langle  \tilde{R} , U  \rangle  +  \delta_n.
%\end{align}
Optimizing over $U$ and $\tilde{U}$ establishes part (i). The proof of part (ii) follows similarly. 
\end{proof}

Proposition~\ref{prop:MMSE_bound} shows that the MMSE matrix can be bounded from above and below by the supergradients of the approximation formula. Note that the non-asymptotic bounds on the difference  $|\cI_n - \hat{\cI}_n|$ given in Theorem~\ref{thm:MI_bound}  provide conditions under which the terms $\delta_n, \delta_n'$ are small.

To provide a more explicit characterization of the  MMSE matrix it is useful to consider the limiting behavior associated with the large-$n$ limit.  A basic result in convex analysis, which is known as Griffiths' lemma in the statistical physics literature,  says that if a sequence of differentiable convex functions $f_n$ converges  pointwise  to limit $f$, then the limit is convex, and hence differentiable almost everywhere, and the gradients  $\nabla f_n$ converge to $\nabla f$ at all points where $f$ is differentiable. A multivariate version of this result is provided in Lemma~\ref{lem:convex_gradients} in the Appendix. For the MMSE matrix, this means that pointwise convergence of the normalized mutual information  implies convergence of the normalized MMSE matrix almost everywhere. 

Before we state our next result, there is an additional technical detail that needs to be addressed, which is that the superdifferentials  of $\hat{\cI}_n$ may be multivalued on the boundary.  We define $\cP_R \colon \psd^d \to \psd^d$ to be the projection operator in Frobenius norm generated by the trace inner product onto the cone of positive definite matrices given by $\{ \tilde{R} \in \psd^d \, : \, \tilde{R} \preceq   t R, \, t \ge 0  \}$. This projection is the identity operator whenever $R$ is full rank. A simple upper bound on the MMSE matrix provided by the mean-squared error of the maximum likelihood estimator reveals that% the projection of the MMSE matrix satisfies
\begin{subequations}
\begin{align}
\cP_R( \MMSE(\bX \mid \bY_{R,S}, \bZ )) \preceq R^\dagger\\
\cP_S( \MMSE(\bX^{\otimes 2}  \mid \bY_{R,S}, \bZ)) \preceq S^\dagger,
\end{align}
\end{subequations}
where $(\cdot)^\dagger$ denotes the Moore--Penrose inverse of a matrix $M$ and  $\cP_S \colon \psd^{d^2} \to \psd^{d^2}$ is defined analogously to $\cP_R$. The next result shows that the existence of a limit for  $\hat{\cI}_n$ leads to tight bounds on the projected MMSE matrix at all points where the projection of the superdifferential  of $\hat{\cI}_n$ is single-valued.

\begin{prop}\label{prop:MMSE_lim}
Let $\cI_n$ be the mutual information function associated with a sequence of problems where the number of columns $d$ is fixed as the number of rows $n$ increases. If $\cI_n$ converges pointwise to a limit $\cI$, then $\cI$ is concave and thus differentiable almost everywhere. Furthermore:
\begin{enumerate}[(i)]

\item For all $R, \tilde{R} \in \psd^d$ and $S , \tilde{S} \in \psd^{d^2}$ 
\begin{align*}
\liminf_{n \to \infty} \frac{1}{n}   \langle   \tilde{R} , \MMSE( \bX \mid \bY_{R,S}, \bZ ) \rangle  &\ge \min_{U \in 2 \partial_R  \cI(R, S) }  \langle  \tilde{R}, U  \rangle \\
\liminf_{n \to \infty} \frac{1}{n}  \langle   \tilde{S} , \MMSE( \bX^{\otimes 2}  \mid \bY_{R,S}, \bZ ) \rangle  & \ge \min_{V \in 2  \partial_{S} \cI(R, S) }  \langle \tilde{S} , V  \rangle .
\end{align*}

\item For all $R, \tilde{R} \in \psd^d$ and $S , \tilde{S} \in \psd^{d^2}$ such that $\tilde{R}  \preceq \eps R $ and $\tilde{S} \preceq \eps  S$ for some $\eps > 0$, 
\begin{align*}
\limsup_{n \to \infty} \frac{1}{n}   \langle   \tilde{R} , \MMSE( \bX \mid \bY_{R,S}, \bZ ) \rangle  &\le \max_{U \in 2 \partial_R  \cI(R, S) }  \langle  \tilde{R}, U  \rangle \\
\limsup_{n \to \infty}  \frac{1}{n} \langle   \tilde{S} , \MMSE( \bX^{\otimes 2}  \mid \bY_{R,S}, \bZ ) \rangle  & \le \max_{V \in 2  \partial_{S} \cI(R, S) }  \langle \tilde{S} , V  \rangle .
\end{align*}

\item  If $\cP_{R} ( \partial_R \cI(R,S) ) = \{\frac{1}{2} U\}$ is single-valued then 
\begin{align*}
\lim_{n \to \infty}   \frac{1}{n}  \cP_{R} \left(  \MMSE( \bX \mid \bY_{R,S}, \bZ)  \right)   = U.
\end{align*}

\item If $\cP_{S} ( \partial_S \cI(R,S) ) = \{\frac{1}{2} V\}$ is single-valued then 
\begin{align*}
\lim_{n \to \infty} \frac{1}{n^2}   \cP_{S} \left(  \MMSE( \bX^{\otimes 2}  \mid \bY_{R,S}, \bZ)  \right)   = V.
\end{align*}
\end{enumerate}
\end{prop}
\begin{proof} Parts (i) and (ii) follow from Lemma~\ref{lem:convex_gradients} by noting that $-\cI_n$ and $- \hat{\cI}_n$ are convex. Parts (ii) and (iv)  follow straightforwardly from parts (i) and (ii). 
\end{proof}

Combining Corollary~\ref{cor:MI_limit}  with Proposition~\ref{prop:MMSE_lim}  leads to  the following result. 

\begin{theorem}\label{thm:MMSE_limit} 
Consider the assumptions of Corollary~\ref{cor:MI_limit}. If for a given pair $(R, S)$ the minimum in \eqref{eq:I0} is attained at a unique point $Q^\star \in \cQ$,   then the MMSE matrix satisfies 
\begin{align*}
\lim_{n \to \infty} \frac{1}{n} \cP_{R} \left(  \MMSE(\bX \mid \bY_{R,S}, \bZ )  \right)  &=  \cP_R\left(M_1 \right) \\
\lim_{n \to \infty} \frac{1}{n^2} \cP_{S} \left( \MMSE(\bX^{\otimes 2}  \mid \bY_{R,S} , \bZ) \right)  & =  \cP_{S}\left( M_2 \right ),
%\lim_{n \to \infty} \frac{1}{n} \cP_{R} \left(  \MMSE(\bX \mid \bY_{R,S}, \bZ )  \right)  &=  \cP_R\left(\ex{ X_0 X_0^\top} - Q^\star \right) \\
%\lim_{n \to \infty} \frac{1}{n^2} \cP_{S} \left( \MMSE(\bX^{\otimes 2}  \mid \bY_{R,S} , \bZ) \right)  & =  \cP_{S}\left( \ex{X_0 X_0^\top}^{\otimes 2}  - (Q^\star)^{\otimes 2} \right ).
\end{align*}
where 
\[M_1 = \ex{ X_0 X_0^\top} - Q^\star , \quad M_2 =  \ex{X_0 X_0^\top}^{\otimes 2}  - (Q^\star)^{\otimes 2}.\] 
\end{theorem}

Let us comment briefly on the  restriction to the projection of the MMSE. First, note that if $R$ or $S$ are positive definite, then there is no restriction and Theorem~\ref{thm:MMSE_limit} establishes the exact limit of the MMSE matrix. However,  settings where $R$ and $S$ are degenerate are also important to consider because they correspond to previously studied problem settings such as the  spiked Wigner and spiked Wishart models.  The following examples shed some light on what can and cannot be said in these settings. 

\begin{example}[Sign Invariance] Suppose that there is no side information and the distribution of $\bX$  is invariant to a sign change,  that is  $\bX$ is equal to $-\bX$ in distribution. If $R =0$,  the conditional distribution of $\bX$ given $\bY_{0,S}$ is also invariant to a change in sign. Consequently, the conditional expectation is zero almost surely and $\MMSE(\bX \mid \bY_{0, S}) = \ex{ \bX^\top \bX}$ for all $S$. In other words,  the MMSE is constant and does not decrease with an increase in the signal strength. This surprising behavior shows that the MMSE evaluated at a single point does not necessarily provide the whole story.  
 
In the context of our problem formulation, the linear model parameterized by $R$ provides a way to address the non-identifiability issue. Specifically, an arbitrarily small but nonzero value of $R$ is sufficient to resolve the sign ambiguity with high probability such that the MMSE convergences to zero as $\lambda_\mathrm{min}(S)$ diverges.  This discrepancy between $R=0$ and the small $R$ limit shows that the technical details involving the upper bounds on the MMSE matrix in part (ii) of Proposition~\ref{prop:MMSE_lim} are fundamental to the problem and are not an artifact of our proof technique. 
\end{example}

\begin{example}[Spiked Wigner Model]
The spiked Wigner model is a special case of $\eqref{eq:Y_RS}$ where $R$ is zero and $S$ is given by the rank-one matrix $s \gvec(I_d)\gvec(I_d)$ for some $s > 0$. In this case, the projection of the MMSE matrix onto the cone generated by $S$ satisfies
\begin{align}
 \gtr( \cP_S( \MMSE(\bX^{\otimes 2} \mid \bY_{0, S}))) = \ex{ \left\| \bX^\top \bX - \ex{ \bX^\top \bX \mid \bY_{0, S} }\right \|_\fro^2}.
\end{align}
This is the same performance metric that appears in previous work focusing on the spiked Wigner model~\cite{deshpande:2017, lesieur:2015,lesieur:2016, krzakala:2016, barbier:2016a,lelarge:2018,  lesieur:2017}. 
\end{example}

\subsection{Overlap concentration}\label{sec:overlap}

The mutual information and MMSE terms studied in the previous sections provide measures of the average difference between the distribution of a random matrix and the  conditional distribution given observations.  In this section, we focus on measures that provide further information about the variability in the conditional distribution itself. The results in this section are not specific to the matrix tensor product model and can be applied generally to any inference problem with matrix-structured variables.  To lighten the notation, we use $\bZ$ to denote a generic set of observations, which may or may  not include the observations from the joint model \eqref{eq:Y_RS}.   
 
 The overlap matrix $\bO$  associated with observations $\bZ$   is the $d \times d$ random matrix defined by
\begin{align}
\bO \coloneqq (\bX')^\top \bX'',
\end{align}
where $\bX'$ and $\bX''$ are conditionally independent draws from the conditional distribution of $\bX$ given $\bZ$. Note that the conditional expectation of the overlap matrix  is a $d \times d$ positive semidefinite matrix given by $\ex{ \bO \mid \bZ} = \ex{ \bX \mid \bZ}^\top \ex{ \bX \mid \bZ}$. Accordingly, there is a one-to-one correspondence between the expectation of the overlap matrix and the MMSE matrix given by 
\begin{align}
\ex{ \bO} + \MMSE(\bX \mid \bZ) = \ex{ \bX^\top \bX}. \label{eq:OtoMMSE}
\end{align}

A fundamental question of interest is whether the overlap matrix concentrates about its expectation. When such concentration occurs, it means that the rows of the matrix are weakly correlated with high probability and also, that the magnitude of the squared error does not fluctuate significantly with the randomness in the observations. 

The main results in this section provide bounds on the variation in the overlap matrix. These bounds depend on the variation in a related random quantity, which is often called the ``free energy" in the statistical physics literature. Specifically, consider the random variable $F$ given by 
\begin{align}
F \coloneqq  \log \frac{ \dd P}{ \dd Q}( Y), \quad Y \sim P ,
\end{align} 
where $P$ and $Q$ are probability measures defined on the same space such that $P$ is absolutely continuous with respect to $Q$. The mean and variance of $F$ correspond to the relative entropy  and the relative entropy variance, respectively, which are given by 
\begin{align}
D( P \, \| Q) & \coloneqq  \bEx_{P}\left[ \log \frac{ \dd P}{ \dd Q} \right]\\
V(P \, \| \, Q)  & \coloneqq \bEx_{P}\left[ \left(\log \frac{ \dd P}{ \dd Q} - D(P \, \| \, Q) \right)^2\right].
\end{align} 

The first main result  in this section establishes an inequality between the squared deviation of the conditional expectation of the overlap matrix and the variance of the free energy. The link between these quantities is provided by the linear Gaussian model
 \begin{align}
\bY_{R} = \bX R^{1/2} + \bW, \label{eq:Y_R}
\end{align}
where $ R\in \psd^d$ and  $\bW$ is a standard Gaussian matrix that is independent of $(\bX, \bZ)$. Note that this is equivalent to the special case $(R,0)$ of the joint model in \eqref{eq:Y_RS}. For $R \in \psd^d$, we use $\bO_R$ to denote the overlap matrix associated with the augmented observations $(\bY_R, \bZ)$.

Furthermore, let $B_1, \dots, B_{d^2}$ be the collection of $d \times d$  positive semidefinite matrices defined by
\begin{align} \label{eq:def_basis}
B_{\pi(a,b) } &\coloneqq \begin{dcases}  e_a e_b^\top , & a = b\\
\frac{1}{2}  (e_a  + e_b) (e_a  + e_b)^\top   & a  <  b \\
\frac{1}{2}  (e_a  - e_b) (e_a - e_b)^\top     & a  >  b 
  \end{dcases},
\end{align}
for $a,b =  1, \dots d,$ where $\pi(a,b) = a + (d-1) b$ and $e_a$ denotes the $a$-th standard basis vector in $d$-dimensions. Note that each $B_k$ has rank one and unit norm. 

The following result is proved in Section~\ref{proof:thm:var_inq},  based on pointwise I-MMSE relations \cite{weissman:2010, venkat:2012} associated with a continuous-time version of the Gaussian model.

\begin{theorem}[Variance Inequality]\label{thm:var_inq}
Let $\bO$ be the overlap matrix associated with observations $\bZ$ and, for each $R \in \psd^d$, let  $\bO_R$ be the overlap matrix associated with the observation pair $(\bY_R, \bZ)$.  For all  $\delta > 0$, 
\begin{align}
\ex{  \| \ex{ \bO \mid  \bZ } - \bEx[\bO]  \|_\fro^2}  \le  2 \sum_{k=1}^{d^2}   \langle B_k ,  \ex{  \bO }   \rangle  \langle B_k , \ex{\bO_{ \delta B_k} } -  \ex{\bO  }    \rangle
+ \frac{4}{ \delta^2} \sum_{k=1}^{d^2} V\left( P_{\bY_{\delta B_k}, \bZ}  \,  \| \, P_{\bW} \otimes P_{\bZ}   \right), 
  \label{eq:var_inq}
\end{align} 
where $B_1, \dots, B_{d^2}$ are defined as in \eqref{eq:def_basis}.
\end{theorem}

By \eqref{eq:OtoMMSE} and the data processing inequality for the MMSE matrix (see Proposition~\eqref{prop:DPI} ahead), the mapping $R \mapsto \ex{ \bO_{R}}$ is order-preserving and so the second term in \eqref{eq:var_inq} is non-negative and non-decreasing in $\delta$.

The next result provides an upper bound on the squared deviation averaged over a matrix-valued perturbation. For $\beta > 0$, let $\mu_\beta$ be the probability measure of the $d \times d$ positive semidefinite random matrix
\begin{align}
%\cE = 
\frac{\beta }{2} U_0  I_d + \frac{\beta}{2d} \sum_{k=1}^{d^2} U_k B_k,  \label{eq:eps_random}
\end{align}
where $B_1, \dots, B_{d^2}$ are defined as in \eqref{eq:def_basis}  and $U_0, \dots, U_{d^2}$ are i.i.d.\ according to the uniform distribution on $[0,1]$. 
%and $u = (u_0,u_2, \dots, u_{d^2})$ is distributed uniformly on  $[0,1]^{d^2 + 1}$. 
Note that $\sum_{k=0}^{d^2} B_k =d \,  I_d$ and thus $\mu_\beta$ is supported on the bounded set $\{ \eps \in \psd^d \, : \,   \eps \preceq  \beta I_d\}$ . 

The following result is proved in Section~\ref{proof:thm:overlap_concentration}.

\begin{theorem} \label{thm:overlap_concentration} Assume that $(\ex{ \|\bX\|^4})^{1/4}  \le \sqrt{n} \, \rho$ and, for each $R \in \psd^d$, let  $\bO_R$ be the overlap matrix associated with the observation pair $(\bY_R, \bZ)$. Further, suppose that $\psi \colon \psd^d \to \psd^d$ is a function that is order-preserving with respect to the Loewner order, i.e., $\eps_1 \preceq \eps_2 \implies \psi(\eps_1)\preceq \psi(\eps_2)$, and bounded according to $\|\psi(\eps) \|\le \gamma$ for all $\eps \in \psd^d$. Then, for all $\beta, \delta > 0$, 
\begin{align}
\frac{1}{n^2}  \int  \ex{ \|  \bO_{ \eps + \psi(\eps)    }   - \ex{ \bO_{ \eps + \psi(\eps)  } } \|_\fro^2} \, \mu_{\beta}(\eps)  
  \le  2 \sqrt{ \frac{2   d^3 \rho^6 }{ \beta n}   }  + \frac{4 \delta d^2 \rho^4 }{\beta } + \frac{ 4 d^2 }{ \delta^2 n^2}  \overline{V},  \label{eq:overlap_concentration}
\end{align}
where
\begin{align}
\overline{V} &\coloneqq    \max_{\eps, \tilde{\eps}   \in \psd^d \, : \, \|\eps\| \le \beta  + \gamma, \|\tilde{\eps} \| \le \delta }   V\left( P_{\tilde{\bY}_{\tilde{\eps}}, \bY_{\eps   }, \bZ}  \,  \| \,  P_{\tilde{\bW}, \bY_{\eps  } , \bZ}  \right).
\end{align}
In this expression,  $\tilde{\bY}_{\tilde{\eps}} \coloneqq \bX \tilde{\eps}^{1/2} + \tilde{\bW}$ is an additional observation where  $\tilde{\bW}$ is a standard Gaussian matrix that is independent of everything else.
\end{theorem}

To appreciate the significance of Theorem~\ref{thm:overlap_concentration}, it is useful to first consider the case where $\psi$ is identically zero such that the left-hand side is simply the $\mu_\beta$-averaged squared deviation of the overlap. If the upper bound on the relative entropy variance $\overline{V}$ scales at a rate slower then $n^2$ for some fixed pair $(\beta_0, \delta_0)$, then there exists a sequence $(\beta_n, \delta_n)_n$ decreasing to zero such that the right-hand side of \eqref{eq:overlap_concentration} converges to zero in the large-$n$ limit. In other words, concentration of the free energy is sufficient to ensure concentration of the conditional overlap under a vanishingly small perturbation. 

A further contribution of Theorem~\ref{thm:overlap_concentration} is that the perturbation argument still holds under the change of measure given by $\eps \mapsto \eps + \psi(\eps)$,  provided that $\psi$ is order-preserving. This property is crucial for the adaptive interpolation used in Section~\ref{proof:prop:approx_bound}, where the specification of the interpolation path depends on the perturbation. 

Theorem~\ref{thm:overlap_concentration} is related to recent work by  Barbier~\cite{barbier:2020c}, who also  uses perturbation arguments to establish overlap concentration in terms of the variance of the free energy. While there are a number of similarities in the approaches there are also some important differences. First, the results in \cite{barbier:2020c} are based on a componentwise decomposition of the overlap matrix which differs with the frame-based approach in this paper. As a consequence, it is unclear whether the resulting bounds in  \cite{barbier:2020c} still hold for an  arbitrary order-preserving transformation $\psi$. Second, our use of the  variance inequality in Theorem~\ref{thm:var_inq} provides tighter control on the error terms. For example, we do not require that the rows have bounded support. For the applications considered in this paper, the improvements provided by the variance inequality lead to better convergence rates with respect to the dimensions $n$ and $d$.

\subsection{General bound based on relative entropy variance}\label{sec:general_bounds}

The bounds on the difference between the mutual information and the approximation formula given in Section~\ref{sec:mutual_information}  require an independence assumption on the rows of $\bX$ and the side information. In this section, we provide a more general result that depends only on relative entropy variance.

For $R, \eps  \in \psd^d$ and $S \in \psd^{d^2}$ define the normalized relative entropy variance
\begin{align}
\cV_n(\eps  \mid R, S) & \coloneqq \frac{1}{n^2} V\left( P_{\tilde{\bY}_{\eps } , \bY_{R,S} , \bZ} \, \| \, P_{\tilde{\bW},  \bY_{R,S} , \bZ} \right), \label{eq:Vn} 
\end{align}
where   $\tilde{\bY}_{\eps} \coloneqq \bX \eps^{1/2} + \tilde{\bW}$ is an additional observation and $\tilde{\bW}$ is a standard Gaussian matrix that is independent of everything else. In words, this is the relative entropy variance between the conditional distribution of a new observation $\tilde{\bY}_\eps$ given $(\bY_{R,S}, \bZ)$ and the standard Gaussian measure on $\reals^{n \times d}$, i.e., 
\begin{align}
\cV_n(\eps  \mid R, S) & = \frac{1}{n^2}  \var\left( \log \frac{ \dd  P_{\tilde{\bY}_{\eps} \mid \bY_{R,S}, \bZ}}{\dd P_{\tilde{\bW}}} \right).
\end{align}

The following result provides a non-asymptotic bound on the difference between the normalized mutual information $\cI_n$  function defined by \eqref{eq:I_RS} and the approximation formula  $\hat{\cI}_n$ defined by \eqref{eq:Ih_RS}. The proof is  given in Section~\ref{proof:thm:general_bound}. 

\begin{theorem}\label{thm:general_bound}
Assume that $(\ex{ \|\bX\|^4})^{1/4} \le \sqrt{n}\,  \rho$. For a pair $(R,S) \in \psd^d \times \psd^{d^2}$ and  $\beta, \delta > 0$, 
define
\begin{align}
\overline{\cV}_n \coloneqq  \max_{\eps, \tilde{R}, \tilde{S}  } \cV_n(\eps\mid  R+ \tilde{R} , \tilde{S} ),
\end{align}
where the maximum is over positive semidefinite matrices $\eps, \tilde{R} , \tilde{S}$ subject to the constraints
\begin{align}
& \| \eps \| \le \delta,   \quad  \|\tilde{R} \| \le    \beta + 2 d \rho^2 \|S\|, \quad  \|\tilde{S}\|  \le \|S\|  .
\end{align}
Then, the error in the approximation formula satisfies 
\begin{align}
 \big | \cI_n(R, S)   -   \hat{\cI}_n(R, S) \big |  \le  \frac{\beta d\rho^2}{2}  +\|S\| \left[  \sqrt{ \frac{2   d^3 \rho^6 }{ \beta n}   }  + \frac{2 \delta d^2 \rho^4 }{ 2\beta} +  \frac{2 d^2}{ \delta^2 }  \overline{\cV}_n \right]  . \label{eq:general_bound}
 \end{align}
\end{theorem}

Similar to  the bound on the overlap concentration given  in Theorem~\ref{thm:overlap_concentration}, one of the main takeaways from Theorem~\ref{thm:general_bound} is that convergence of the upper bound on the entropy variance $\overline{\cV}_n$ is sufficient to ensure that the difference $| \cI_n(R, S)   -   \hat{\cI}_n(R, S)  |$ converges to zero. Specifically,  note that $\overline{\cV}_n$ is nondecreasing as a function of $(\beta, \delta)$. Thus, if $\overline{\cV}_n$ converges to zero for some fixed pair $(\beta_0, \delta_0)$, then there exists a sequence $(\beta_n, \delta_n)_n$  decreasing to zero such that the right-hand side of \eqref{eq:general_bound} converges to zero. 

One setting where convergence of $\bar{\cV}_n$ can be established directly is the setting considered in Theorem~\ref{thm:MI_bound}. Establishing the convergence of $\overline{\cV}_n$ under more general conditions is an interesting direction for future work.

\section{Properties and examples} \label{sec:prop_approx}

This section studies properties of the mutual information function and the approximation formula given in Definition~\ref{def:Ih} and also provides some examples of how this formula can be applied to the asymmetric and multiview models.

\subsection{Properties of mutual information and MMSE} 

The mutual information and MMSE associated with the joint observation model \eqref{eq:Y_RS} have a number of functional properties that play a prominent role in our analysis. The observations associated with a pair $(R,S)$ can be split  into two separate observations given by $(\bY_{R}, \bY_{0,S})$, where we use the convention that $\bY_{R} = \bY_{R,0}$.  The chain rule for mutual information then yields
\begin{align}
I(\bX ; \bY_{R,S} \mid \bZ) & = I(\bX ; \bY_{R} \mid \bZ)   +  I(\bX ; \bY_{0,S} \mid  \bY_{R} , \bZ). \label{eq:MI_reduction} 
\end{align}
The first term does not depend on $S$ and in the second term,  the observations $\bY_R$ appear only  in the conditioning argument and can be viewed as additional side information. Because our basic formulation does not place constraints on the side information, this means the case of general $(R,S)$ can be reduced to the case $(0,S)$ with augmented side information. Note  that in the setting of Theorem~\ref{thm:MI_bound}, the additional observations $\bY_{R}$ satisfy the independence assumption required on the side information and so there is no loss of generality in this reduction. 

Next,   the data processing inequality for mutual information says that if variables $(X,Y,Z)$ form a Markov chain, that is $Z$ and $X$ are conditionally independent given $Y$, then $I(X;Y) \ge I(X ; Z)$. An estimation-theoretic version of the data processing inequality \cite[Proposition~5]{rioul:2011} applies to the MMSE: if $X$ is a random vector and $(X,Y,Z)$ form a Markov chain then $\mmse( X \mid Y) \le \mmse( X \mid Z)$.  When the mutual information and MMSE are viewed as functions of the matrices $(R,S)$ in the joint model \eqref{eq:Y_RS}, the data processing inequality implies an order preserving property.

\begin{prop}[Data processing inequality]\label{prop:DPI} The mapping $(R,S) \mapsto I(\bX ; \bY_{R,S} \mid \bZ)$ is order-preserving with respect to the Loewner order and the mapping $(R,S) \mapsto \MMSE(\bX \mid \bY_{R,S} ,  \bZ)$ is order-reversing in the Leowner order, i.e., if $R_1 \preceq R_2$ and $S_1 \preceq S_2$ then 
\begin{subequations} 
\begin{align}
I(\bX ; \bY_{R_1,S_1} \mid \bZ) &\le I(\bX ; \bY_{R_2,S_2} \mid \bZ)\\
 \MMSE(\bX \mid \bY_{R_1,S_1} ,  \bZ)  &\succeq  \MMSE(\bX \mid \bY_{R_2,S_2} ,  \bZ).
\end{align}
\end{subequations} 
\end{prop}

\subsection{Alternative characterizations}\label{sec:alternate_formulas}

It is convenient to consider a different characterization of the approximation formula where the mutual information is replaced by a relative entropy term. Specifically, let $\cD_n \colon \psd^d \times \psd^{d^2} \to \reals$ be defined as%according to 
\begin{align}
\cD_n(R,S) \coloneqq \frac{1}{n} D( P_{\bY_{R,S}, \bZ} \, \| \, P_{\bW , \bW'} \otimes P_{\bZ} ).
\end{align}
This function is related to the mutual information via the identity 
\begin{align}
\cI_n(R,S) + \cD_n(R,S)  
=  \frac{1}{2} \langle R , \ex{ \tfrac{1}{n} \bX^\top \bX} \rangle +  \frac{1}{2} \langle S , \ex{ (\tfrac{1}{n} \bX^\top \bX)^{\otimes 2} } \rangle, \label{eq:InDn}
\end{align}
which follows from recognizing the right-hand side as $\frac{1}{n} D(P_{\bX, \bY_{R,S}, \bZ } \, \| \, P_{\bX , \bW, \bW' , \bZ})$ and then applying the chain rule for relative entropy.  
In view of \eqref{eq:InDn}, the relative entropy approximation formula $\hat{\cD}_n \colon  \psd^d \times \psd^{d^2} \to \reals$ is defined according to  
\begin{align}
\label{eq:Dh_RS}
\hat{\cD}_n(R, S)  
\coloneqq    \adjustlimits   \max_{Q \in \cQ} \inf_{\tilde{R} \in \psd^d} \left\{ \cD_n(R+ \tilde{R})  + \frac{1}{2}  \langle S,Q^{\otimes 2} \rangle -  \frac{1}{2} \langle \tilde{R},  Q \rangle   \right\}, 
\end{align}
where $\cQ \coloneqq \{ Q \in \psd^d \, : \, Q \preceq \ex{\frac{1}{n}  \bX^\top \bX}\}$ and $\cD_n(R) \coloneqq \cD_n(R,0)$. Note that  $\cI_n- \hat{\cI}_n = \hat{\cD}_n - \cD_n$ and so the error of the approximation formula in the relative entropy formulation is the same as in the mutual information formulation. 

The next result provides a characterization of the bilinear form $(Q_1, Q_2) \mapsto  \langle S , Q_1 \otimes Q_2\rangle$ that is used to provide alternative representations of the approximation formula. 

\begin{lemma}\label{lem:lin_op}
For each $S \in \sym^{d^2}$ there exist unique linear operators $\cT , \cT^* \colon \sym^d \to \sym^d$ such that 
\begin{align*}
\langle S, M_1 \otimes M_2 \rangle  = \langle \cT(M_1) , M_2 \rangle = \langle M_1, \cT^*(M_2) \rangle,
\end{align*}
for all $M_1,M_2 \in \reals^{d \times d}$. % where $\cT^*$ is the adjoint of $\cT$. %Tor any $\lambda_1, \dots, \lambda_{m} \in \reals$  and $V_1, \dots, V_{m} \in  \reals^{d \times d}$ such that  
For any decomposition of the form
\begin{align}
S = \sum_{k=1}^{m} \lambda_k \gvec(V_k) \gvec(V_k)^\top, \label{eq:S_decomp}
\end{align}
with $\lambda_1, \dots, \lambda_{m} \in \reals$  and $V_1, \dots, V_{m} \in  \reals^{d \times d}$, these operators can be expressed as
\begin{align*}
\cT(M)  &= \sum_{k=1}^m \lambda_k V_k M V_k^\top, \quad \cT^*(M) = \sum_{k=1}^m \lambda_k V^\top_k M V_i,
\end{align*}
Furthermore,  $\|\cT(M) \| , \|\cT^*(M)\| \le \|S\| \|M\|_*$ for all $M \in \reals^{d \times d}$, 
where $\|\cdot\|_*$ denotes the nuclear norm,  and if $S \succeq 0$ then both $\cT$ and $\cT^*$ are order-preserving with respect to the Leowner order. 
\end{lemma} 
\begin{proof}
We prove these results for $\cT$; the corresponding proof for $\cT^*$ follows analogously. 
 For any decomposition of the form given in \eqref{eq:S_decomp}, 
\begin{align*}
\langle S,  M_1 \otimes M_2 \rangle %& = \sum_{k=1}^m \lambda_k   \gvec(V_k) ( M_1 \otimes M_2)  \gvec(V_k)\\
& = \sum_{k=1}^m \lambda_k \left  \langle  \gvec(V_k),  ( M_1 \otimes M_2)  \gvec(V_k) \right \rangle 
 = \sum_{k=1}^m \lambda_k \left  \langle  V_k,   M_2 V_k  M^\top_1  \right \rangle
  = \Big \langle  \sum_{k=1}^m \lambda_k  V_k M_1 V_k^\top,  M_2   \Big \rangle ,
\end{align*}
where the second step follows from the basic property $\gvec(A X B) = (B^\top \otimes A) \gvec(X)$. 

To establish the bound on the spectral norm of $\cT(M)$, observe that for any unit vector $u \in\reals^d$, 
\begin{align*}
u^\top \cT(M) u & = \langle \cT(M), u u^\top \rangle = \langle S, M \otimes u u^\top \rangle  \le  \|S\| \| M \otimes u u^\top\|_{*}  =  \|S\| \| M \|_{*},
\end{align*}
where we have used the matrix H\"{o}lder inequality and the fact that the nonzero singular values of $M \otimes u u^\top$ are the same as the nonzero singular values of $M$. 

Finally, by linearity, the order preserving property is equivalent to  $\cT(M) \succeq 0$ for all $M \in \psd^d$. If $S \succeq 0$ then the eigenvalue decomposition has the form given in \eqref{eq:S_decomp} with $\lambda_k \ge 0$. Hence, $u^\top \cT(M) u= \sum_{k=1}^m \lambda_k u^\top V_k M V_k^\top u \ge 0$ for all $M \in \psd^d$ and $u \in \reals^d$, and so $\cT$ is order preserving.  
\end{proof}

\iffalse
{\color{blue} 
The following is a variant of the  Fenchel-Moreau Theorem~\cite[Theorem~13.37]{bauschke:2017} which says that a proper, lower semi-continuous, and convex function, is equal to its biconjugate (i.e., the  conjugate of the convex conjugate). We provide a short  proof adapted to the setting of this paper. 

\begin{lemma}
For all $R \in \psd^d$, 
\begin{align}
  \cD_n(R)& = \adjustlimits  \max_{Q \in \cQ}  \inf_{\tilde{R}\in  \psd^d }  \left\{ \cD_n(\tilde{R}) + \frac{1}{2} \langle R  -\tilde{R} , Q \rangle\right\}. 
\end{align}
\end{lemma}
\begin{proof}
 One direction of the identity follows immediately from
 \begin{align}
 D_n(R) \ge  \inf_{\tilde{R}\in  \psd^d }  \left\{ \cD_n(\tilde{R}) + \frac{1}{2} \langle R  -\tilde{R} , Q \rangle\right\}  
   \end{align}
 for all $Q \in \cQ$. 
For the other direction, let $Q_R =  2 \nabla \cD_n(R)$ and observe that, by convexity, 
  \begin{align}
 \cD_n(R)  \le  \cD_n(\tilde{R})  -    \frac{1}{2} \langle   R - \tilde{R}   ,  Q_{R} \rangle , 
  \end{align}
  for all $\tilde{R} \in \psd^d$.    Since $Q_R \in \cQ$, it follows that
\begin{align}
 \cD_n(R) & \le \inf_{\tilde{R}\in  \psd^d }  \left\{ \cD_n(\tilde{R}) + \frac{1}{2} \langle R  -\tilde{R} , Q_R \rangle\right\} \\
  & \le  \adjustlimits  \max_{Q \in \cQ}  \inf_{\tilde{R}\in  \psd^d }  \left\{ \cD_n(\tilde{R}) + \frac{1}{2} \langle R  -\tilde{R} , Q \rangle\right\}. 
\end{align}
\end{proof} 
}
\fi

The next result shows that the approximation formula can be expressed as a max-min problem defined on a compact set. This formulation is useful for numerical approximation. 

\begin{prop}\label{prop:Dh_RS_alt}
The approximation formula $\hat{\cD}_n(R,S)$ given in \eqref{eq:Dh_RS} satisfies 
\begin{align*}
 \hat{\cD}_n(R,S)  % \notag \\
& =   \adjustlimits \max_{Q \in \cQ} \min_{Q \in \cQ } \left\{ D(R + \cT(Q) + \cT^*(\tilde{Q}) )  - \frac{1}{2} \langle \cT(Q),  \tilde{Q} \rangle \right\} ,  
\end{align*}
where the linear operators $\cT, \cT^* \colon \psd^d \to \psd^d$ are defined as a function of $S$ according to  Lemma~\ref{lem:lin_op}.  
\end{prop}
\begin{proof}
By the reduction in \eqref{eq:MI_reduction}  we may assume without loss of generality that $R = 0$.  For $Q \in \cQ$ define the compact set $\cR_{Q} = \{  \cT(Q) + \cT^*(\tilde{Q})\, : \,  \tilde{Q} \in \cQ  \}$. Since $\cR_Q \subset \psd^d$, we have the upper bound:
\begin{align*}
%& 
\hat{\cD}_n(R, S) % \notag \\
& \le  \adjustlimits  \max_{Q \in \cQ}  \min_{\tilde{R}\in  \cR_{Q}  } \left\{ \cD_{n}( \tilde{R} )   -  \frac{1}{2} \langle \tilde{R} - \cT(Q)  , Q \rangle  \right\} \\
&=    \adjustlimits  \max_{Q \in \cQ}  \min_{\tilde{Q} \in \cQ } \left\{ \cD_{n}(  \cT(Q) + \cT^*(\tilde{Q}) )   -  \frac{1}{2} \langle \cT^*(\tilde{Q})     , Q \rangle  \right\} \\
&=    \adjustlimits  \max_{Q \in \cQ}  \min_{\tilde{Q} \in \cQ } \left\{ \cD_{n}(  \cT(Q) + \cT^*(\tilde{Q}) )   -  \frac{1}{2} \langle \cT(Q)     , \tilde{Q} \rangle  \right\} .
\end{align*}

The matching upper bound is established using a duality argument.  Letting $D^*_n(Q) \coloneqq \sup_{R \in \psd^d} \{ \frac{1}{2} \langle R, Q \rangle - D_n(R) \}$ denote the convex conjugate of $\cD_n(R)$, we can write 
\begin{align}
\hat{\cD}_n(R, S) 
&=  \max_{Q \in  \cQ }  \left\{  \frac{1}{2} \langle \cT(Q)  , Q \rangle  - \cD^*_{n}( Q )   \right\}.  \label{eq:Dhat_alt_c}
%\notag \\
%& \le  \sup_{Q \in  \psd^d }  \left\{  \frac{1}{2} \langle \cT(Q)  , Q \rangle  - \cD^*_{n}( Q )   \right\}. \label{eq:Dhat_alt_c}
\end{align}
Next, for every $\tilde{R} \in \psd^d$ we have the basic inequality
\begin{align*}
\langle \cT(Q)  , Q \rangle %& =\langle \tilde{R}  , Q \rangle + \langle \cT(Q) - \tilde{R}  , Q \rangle\\
%&\ge  \min_{\tilde{Q} \in \cQ} \left\{ \langle \tilde{R}  , Q \rangle + \langle \cT(Q) - \tilde{R}  , \tilde{Q} \rangle \right\}\\
& \ge  \min_{\tilde{Q} \in \cQ} \left\{ \langle \tilde{R}  + \cT^*(\tilde{Q} )  , Q \rangle  -  \langle \tilde{R}  , \tilde{Q} \rangle \right\} .
\end{align*} 
Plugging this inequality into \eqref{eq:Dhat_alt_c} shows that $\hat{\cD}_n(R, S)$ is bounded from below by
\begin{align}
%\hat{\cD}_n(R, S)  \ge \\
 \adjustlimits \max_{Q \in  \psd^d } \min_{\tilde{Q} \in \cQ}   \left\{  \frac{1}{2} \langle \tilde{R}  + \cT^*(\tilde{Q} )  , Q \rangle   - \frac{1}{2} \langle \tilde{R} , \tilde{Q} \rangle  - \cD_n^*(Q) \right\}. \label{eq:Dhat_alt_d}
\end{align}
The expression inside the brackets is upper semicontinuous and concave in $Q$,  and  linear in $\tilde{Q}$.  By the compactness of $\cQ$, we can apply the minimax theorem (e.g., \cite[Corollary~37.3.1]{rockafellar:1970}) to swap the order of the maximum  and the minimum over these terms, and this leads to the lower bound
\begin{align*}
 \hat{\cD}_n(R, S) %\notag \\
&\ge
  \adjustlimits  \min_{\tilde{Q} \in \cQ}  \max_{Q \in \psd^d }  \left\{  \frac{1}{2} \langle \tilde{R}  + \cT^*(\tilde{Q} )  , Q \rangle   - \frac{1}{2} \langle \tilde{R} , \tilde{Q} \rangle  - \cD_n^*(Q) \right\}\\
&=   \min_{\tilde{Q} \in \cQ}   \left\{  \max_{Q \in  \psd^d }  \left(  \frac{1}{2} \langle \tilde{R}  + \cT^*(\tilde{Q} )  , Q \rangle -   \cD_n^*(Q) \right)   - \frac{1}{2} \langle \tilde{R} , \tilde{Q} \rangle   \right\} \\
%&=   \min_{\tilde{Q} \in \cQ}   \left\{  \sup_{Q \in  \psd^d }  \left(  \frac{1}{2} \langle \tilde{R}  + \cT^*(\tilde{Q} )  , Q \rangle -   \cD_n^*(Q) \right)   - \frac{1}{2} \langle \tilde{R} , \tilde{Q} \rangle   \right\} \\
&=   \min_{\tilde{Q} \in \cQ}   \left\{ \cD_n(  \tilde{R}  + \cT^*(\tilde{Q} )  )   - \frac{1}{2} \langle \tilde{R} , \tilde{Q} \rangle   \right\}. %\label{eq:Dhat_alt_f}.
\end{align*}
Here the last step follows from the  Fenchel-Moreau Theorem~\cite[Theorem~13.37]{bauschke:2017} which says that a proper, lower semi-continuous, and convex function, is equal to its biconjugate (i.e., the  conjugate of the convex conjugate) and the fact that $2 \nabla_R \cD_n(R) \in \cQ$ for all $R \in \psd^d$.  Since this inequality holds for every choice of $\tilde{R}$, we can take the maximum of both sides with respect to $\tilde{R}  \in \{ \cT^*(Q)\, : \,  Q \in \cQ  \} $, and this gives the matching lower bound. 
\end{proof}

The approximation formula can be simplified further  in the special case where $M \mapsto \langle S, M^{\otimes2} \rangle $ is convex. Necessary and sufficient conditions for convexity are given as follows.

\begin{lemma}Suppose that  $S\in \psd^{d^2}$ has a decomposition of the form   $S = \sum_{l=1}^L \gvec(B_l) \gvec(B_l)^\top$ with $B_l \in \reals^{d \times d}$. Then, the mapping $M \mapsto \langle S, M^{\otimes 2}\rangle$ is convex on $\sym^d$ if and only if 
\begin{align}
\sum_{l=1}^L  \mathrm{D}^\top_n \left\{ (B_l \otimes B_l) + (B_l \otimes B_l)^\top \right\} \mathrm{D}_n \succeq 0, \label{eq:convex_condition}
\end{align}
where $\mathrm{D}_n$  is the duplication matrix, i.e.,  the unique $d^2 \times d (d+1)/2$ matrix such that $\gvec(M)  =\mathrm{D}_n \vech(M)$ for all $M \in \sym^d$ where $\vech(M)$ is the $d(d+1)/2$-dimensional vector obtained by stacking the entries on or below the diagonal in $M$; see \cite[Chapter 3.8]{magnus:2007}.
\end{lemma}
\begin{proof}
For $M\in \sym^d$ we can write $\langle S, M \rangle = g( \vech(M))$ where   $g : \reals^{d(d+1)/2} \to \reals$  is given by \[g(v) =  v^\top \mathrm{D}_n^\top  \sum_{l=1}^L ( B_l \otimes B_l)^\top \mathrm{D}_n v.\] The function $g$ is convex if and only if its Hessian is positive semidefinite, and this leads to the condition given in \eqref{eq:convex_condition}.
\end{proof}

%\begin{figure*}[!b]
%% ensure that we have normalsize text
%\normalsize
%% Store the current equation number.
%\setcounter{MYtempeqncnt}{\value{equation}}
%% Set the equation number to one less than the one
%% desired for the first equation here.
%% The value here will have to changed if equations
%% are added or removed prior to the place these
%% equations are referenced in the main text.
%\hrulefill
%\setcounter{equation}{47}
%%\begin{subequations}
%\begin{alignat}{3}
%\hat{\cD}_n(R,S) & = \adjustlimits \max_{Q \in \cQ} \min_{Q \in \cQ } \left\{ \cD_n\left(R + \sum_{l=1}^L \left\{ B_l  Q B_l^\top + B_l^\top \tilde{Q} B_l \right\}  \right )  - \frac{1}{2} \sum_{l=1}^L  \langle  B_l Q , \tilde{Q} B_l^\top  \rangle   \right\}  \label{eq:Dh_multiview}\\
%\hat{\cD}_n(R,S) 
%& = \max_{Q \in \cQ} \left\{ \cD_n\left(R + \sum_{l=1}^L \left\{ B_l  Q B_l^\top + B_l^\top Q B_l \right\}  \right )  - \frac{1}{2} \sum_{l=1}^L  \langle  B_l Q , Q B_l^\top  \rangle   \right\} &\qquad & \text{if \eqref{eq:convex_condition} holds.}  \label{eq:Dh_multiview_convex}
%\end{alignat}
%%\end{subequations}
%% Restore the current equation number.
%\setcounter{equation}{\value{MYtempeqncnt}}
%% IEEE uses as a separator
%%\hrulefill
%% The spacer can be tweaked to stop underfull vboxes.
%%\vspace*{4pt}
%\end{figure*}

%Recall that by Corollary~\ref{cor:convex_LB}, convexity on $\reals^{d \times d}$ implies that the approximation formula is a lower bound: $D(R,S) \ge \hat{D}(R,S)$. 

The next result shows that convexity allows for a representation of the approximation formula involving only a single matrix variable. 
\begin{prop}\label{prop:Dh_convex}
If $Q \mapsto \langle S, Q^{\otimes2} \rangle $ is convex on $\sym^d$, then the approximation formula $\hat{\cD}_n(R,S)$ given in \eqref{eq:Dh_RS} satisfies 
\begin{align*}
\hat{\cD}_n(R,S) 
& =\max_{Q \in \cQ}  \left\{ D(R + \cT(Q) + \cT^*(Q) )  - \frac{1}{2} \langle \cT(Q) , Q\rangle \right\} ,% \label{eq:Dh_convex}
\end{align*}
where the linear operator $\cT \colon \psd^d \to \psd^d$ is defined as a function of $S$ according to  Lemma~\ref{lem:lin_op}.  
%where $\cT$ is the linear operator described in Lemma~\ref{lem:lin_op}.  
\end{prop}
\begin{proof}
By the reduction in \eqref{eq:MI_reduction}  we may assume without loss of generality that $R = 0$. The fact that this expression is an upper bound on $\hat{\cD}_n$ follows immediately from Proposition~\ref{prop:Dh_RS_alt}. To obtain a matching lower bound, observe that  convexity of $Q \mapsto \langle S, Q^{\otimes 2} \rangle$  implies that for  all $Q, \tilde{Q} \in \psd^d$,
\begin{align*}
\langle S, Q^{\otimes 2 } \rangle  %& =  \langle \cT(Q) , Q  \rangle \\
& \ge   \langle S , \tilde{Q}^{\otimes 2}   \rangle  + \langle Q - \tilde{Q}  , \cT(\tilde{Q})  + \cT^*(\tilde{Q}) \rangle \\
 &= \langle  \cT(\tilde{Q})  + \cT^*(\tilde{Q}) , Q \rangle 
-   \langle \cT(\tilde{Q}) ,\tilde{Q}  \rangle .
\end{align*}
%Furthermore, convexity and differentiability of $R \mapsto \cD_n(R)$ means that for all $\tilde{R}  \in \psd^d$, 
%\begin{multline}
%\cD_n(\tilde{R})   \ge \cD_{n}( \cT( \tilde{Q})  + \cT^*(\tilde{Q}) ) \\
%-  \langle \tilde{R} -   \cT(\tilde{Q})  - \cT^*(\tilde{Q}) , \nabla_R \cD_n(\tilde{R}) \rangle
%\end{multline}
Combining this inequality with \eqref{eq:Dh_RS}, we see that $\hat{\cD}_n(R,S)$ is bounded from below by 
\begin{align*}
%& \hat{\cD}_n(R, S)   \notag\\
%&\ge
%& =
& \adjustlimits   \max_{Q \in \cQ} \inf_{\tilde{R} \in \psd^d} \left\{ \cD_n(\tilde{R}) +   \frac{1}{2} \langle\cT(\tilde{Q})  + \cT^*(\tilde{Q}) -  \tilde{R},  Q \rangle   \right\} \notag \\
& \quad  - \frac{1}{2}   \langle \cT(\tilde{Q}) ,\tilde{Q}  \rangle \\
%& \quad  + \cD_{n}( \cT( \tilde{Q})  + \cT^*(\tilde{Q}) )  - \frac{1}{2}   \langle \cT(\tilde{Q}) ,\tilde{Q}  \rangle \\
%&\ge  \max_{Q \in \cQ} \left\{ \frac{1}{2} \langle\cT(\tilde{Q})  + \cT^*(\tilde{Q}) ,  Q \rangle  - \cD^*_{n}(Q)  \right\} - \frac{1}{2}   \langle \cT(\tilde{Q}) ,\tilde{Q}  \rangle \\
& =   \cD_{n}( \cT( \tilde{Q})  + \cT^*(\tilde{Q}) )- \frac{1}{2}   \langle \cT(\tilde{Q}) ,\tilde{Q}  \rangle
\end{align*}
where the last step follows from the  Fenchel-Moreau Theorem~\cite[Theorem~13.37]{bauschke:2017} and the fact that $2 \nabla_R \cD_n(R) \in \cQ$ for all $R \in \psd^d$. Taking the maximum of this lower bound with respect to $\tilde{Q} \in \cQ$ gives the matching lower bound. 
\end{proof}

\subsection{Examples} 

\begin{example}[Multiview  model]
The multiview spiked matrix model \eqref{eq:multiview} can be represented using the symmetric matrix tensor product model in Definition~\ref{def:sym_mtp_model} using $ S = \sum_{l=1}^L  \gvec(B_l) \gvec(B_l)^\top$. In this case, the linear operators described in Lemma~\ref{lem:lin_op} are given by
\begin{align}
\cT(M)  &= \sum_{l=1}^L  B_l M B_l^\top, \quad \cT^*(M)  = \sum_{l=1}^L  B^\top_l M B_l.
\end{align}
By Proposition~\ref{prop:Dh_RS_alt}, the approximation formula can be expressed as% in \eqref{eq:Dh_multiview}. 
\begin{align}
\hat{\cD}_n(R,S) & = \adjustlimits \max_{Q \in \cQ} \min_{Q \in \cQ } \left\{ D\left(R + \sum_{l=1}^L \left\{ B_l  Q B_l^\top + B_l^\top \tilde{Q} B_l \right\}  \right )  - \frac{1}{2} \sum_{l=1}^L  \langle  B_l Q , \tilde{Q} B_l^\top  \rangle   \right\} .
\end{align}
Furthermore, if  the convexity condition in \eqref{eq:convex_condition} holds then, by Proposition~\ref{prop:Dh_convex}, we obtain the simplified formula  %in \eqref{eq:Dh_multiview_convex}. \addtocounter{equation}{2}
\begin{align}
\hat{\cD}_n(R,S) & = \max_{Q \in \cQ} \left\{ D\left(R + \sum_{l=1}^L \left\{ B_l  Q B_l^\top + B_l^\top Q B_l \right\}  \right )  - \frac{1}{2} \sum_{l=1}^L  \langle  B_l Q , Q B_l^\top  \rangle   \right\},
\end{align}
While the simplified version has appeared in previous work \cite{mayya:2019, barbier:2020a} the general case in  \eqref{eq:Dh_multiview} is new. 

The spiked Wigner model corresponds to the setting $L =1$ and $B_1 = \sqrt{s } I_d$, which yields 
\begin{align}
\hat{\cD}_n(0,S) & = \max_{Q \in \cQ}  \left\{ \cD_n\left(R+ 2 s Q   \right )  - \frac{s}{2} \|Q\|_\fro^2 \right\} .
\end{align}
In the special case of no side information and  $R =0$ this recovers the formula given by Lelarge and Miolane \cite[Theorem~41]{lelarge:2018}. The only difference is a factor of two in the parameter $s$, which arises because our model does not  assume the noise matrix is symmetric. 

\end{example}

\begin{example}[Asymmetric model] Suppose that $d_1 = d_2  = d$. The asymmetric spiked matrix model \eqref{eq:spiked_wishart} can be represented using \eqref{eq:sym_mtp_model} with a $(n_1+ n_2) \times 2d$ signal matrix  $\bX = \bX_1 \oplus \bX_2$ and  $ S = \gvec(B) \gvec(B)^\top$ where  $B$ is the $d \times d$ matrix given by
\begin{align*}
B = \sqrt{ s} \begin{pmatrix} 0 & I_{d} \\ 0  & 0 \end{pmatrix}.
\end{align*}
For each $M \in \reals^{2d \times 2d}$ it follows that
\begin{align*}
\cT\left(\begin{bmatrix} M_{11} & M_{12} \\ M_{21}  & M_{22} \end{bmatrix}\right)  &= s \begin{bmatrix} M_{22} & 0  \\ 0  & 0 \end{bmatrix}  , \qquad \cT^*\left(\begin{bmatrix} M_{11} & M_{12} \\ M_{21}  & M_{22} \end{bmatrix}\right)  = s \begin{bmatrix} 0 & 0  \\ 0  & M_{11} \end{bmatrix} .
\end{align*}
By Proposition~\ref{prop:Dh_RS_alt}, it follows that
\begin{align*}
\begin{multlined}
\hat{\cD}_n(R,S) =  \adjustlimits \max_{Q_1 \in \cQ_1} \min_{Q_2 \in \cQ_2 }  \left\{ \cD_n\left(R +  s \begin{bmatrix} Q_{2} & 0 \\ 0  & Q_{1} \end{bmatrix}   \right )  - \frac{s}{2}   \langle Q_1, Q_2 \rangle  \right\},
\end{multlined}
\end{align*}
where $\cQ_k = \{ Q \in \psd^d \, : \, Q \preceq \ex{ \frac{1}{n} \bX_k^\top \bX_k } \} $ for $k  =1,2$. Note that here, we have not made any assumptions about the dependence between $\bX_1$ and $\bX_2$ and   $\cD_n$ is defined with respect to the block-diagonal matrix $\bX$.  In this case,  the convexity condition in \eqref{eq:convex_condition} does not hold and so the simplification in Proposition~\ref{prop:Dh_convex} does not apply.  

If $\bX_1$ and $\bX_2$ are conditionally independent given the side information $\bZ$, then the relative entropy function evaluated at the block-diagonal matrix $R = \diag(R_1 , R_2)$ decomposes as
\[
\hat{\cD}_n(\diag(R_1,R_2) ) = \cD^{(1)} _n(R_1)  + \cD^{(2)} _n(R_1),
\]
where $\cD^{(k)}_n$ is the relative entropy function associated with matrix $\bX_k$. Specializing to this setting and evaluating at $R =0$ yields
\begin{align*}
\hat{\cD}_n(0,S)  = \adjustlimits \max_{Q_1 \in \cQ_1} \min_{Q_2 \in \cQ_2 } \left\{ \cD^{(1)} _n(s Q_2)  + \cD^{(2)} _n(s Q_1)- \frac{s}{2}  \langle Q_1, Q_2 \rangle  \right\}.
\end{align*}
This last expression is the same as the one given by Miolane \cite[Theorem~2]{miolane:2017a}.
\end{example}

\section{Proofs of main results}\label{sec:proofs}

This section gives the proofs of Theorem~\ref{thm:MI_bound} and Theorem~\ref{thm:general_bound}.

%of the bounds on the difference between mutual information and the approximation formula given in 
\subsection{Proof of Theorem~\ref{thm:MI_bound}}\label{proof:thm:MI_bound}

The proof of Theorem~\ref{thm:MI_bound}  follows from combining the general bound given in Theorem~\ref{thm:general_bound} with an upper bound on the relative entropy variance term $\cV_n(\eps \mid R, S)$ defined in \eqref{eq:Vn}.

To proceed,  partition $\bY_{R,S}$ according to $(\bY_R, \bY'_S)$ and define the augmented side information $\tilde{\bZ}  = (\bY_R, \bZ)$. Then, the free energy defined by the distributions of  $(\tilde{\bY}_{\eps}, \bY_{R,S} , \bZ)$  and $(\tilde{\bW}, \bY_{R,S} , \bZ)$ can be decomposed as
\begin{align*}
F %&  = \log \frac{ \dd  P_{\tilde{\bY}_{\eps}, \bY_{R,S} , \bZ }}{\dd  P_{\tilde{\bW}, \bY_{R,S} , \bZ } } (\tilde{\bY}_\eps, \bY_{R, S} ,\bZ)  \notag\\
 &  = \log \frac{ \dd  P_{\tilde{\bY}_{\eps}, \bY'_{S} , \tilde{\bZ} }}{\dd  P_{\tilde{\bW}, \bY'_{S} , \tilde{\bZ} } } (\tilde{\bY}_\eps, \bY'_{S} ,\tilde{\bZ})   =  \log \frac{ \dd  P_{\tilde{\bY}_{\eps}, \bY'_{S} , \tilde{\bZ} }}{\dd P_{\tilde{\bW}, \bW', \tilde{\bZ} }}(\tilde{\bY}_\eps, \bY'_{S} ,\tilde{\bZ})  -  \log \frac{ \dd  P_{\bY'_{S}, \tilde{\bZ} }}{\dd P_{\bW', \tilde{\bZ} }} (\bY'_{ S} , \tilde{\bZ}) .
\end{align*}
Combing this decomposition with the basic inequality $\var(A+B) \le 2 \var(A) + 2 \var(B)$ for any random variables $(A,B)$, leads to
\begin{align}
%&V\left( P_{\tilde{\bY}_{\eps } , \bY_{R,S} , \bZ} \, \| \, P_{\tilde{\bW},  \bY_{R,S} , \bZ} \right) \notag \\
\cV_n(\eps \mid R, S)
& \le  \frac{2}{n^2}  V\left( P_{ \bY_{\eps,S} , \tilde{\bZ}} \, \| \, P_{ \bY_{0,0} , \tilde{\bZ}} \right)  + \frac{2}{n^2} V\left( P_{ \bY_{0,S} , \tilde{\bZ}} \, \| \, P_{ \bY_{0,0}  , \tilde{\bZ}} \right), \label{eq:Vn_bound}
\end{align}
where we have used the fact that $(\tilde{\bY}_\eps, \bY'_S)$  is equal in distribution to  $\bY_{\eps, S}$. 

Next, observe that if $(\bX, \bZ)$ satisfies the assumption of Theorem~\ref{thm:MI_bound} then so does $(\bX, \tilde{\bZ})$. Accordingly, we can apply  Proposition~\ref{prop:relative_entropy_variance} to obtain 
\begin{align*}
 V( P_{\bY_{\eps, S}  , \tilde{\bZ} } \, \| \, P_{\bY_{0, 0},\bZ}  ) & \le  n d \rho^2 \left( \|\eps \| + d \rho^2 \|S\| \right)  +   3 n d^2 \rho^4 \left( \|\eps \|  + 2 d \rho^2  \|S\|   \right)^2.
\end{align*}
In view of \eqref{eq:Vn_bound}, this leads to the upper bound
\begin{align*}
\overline{\cV}_n
 & \le  \frac{ 4 }{n} \left[d \rho^2 \left( \delta + d \rho^2 \|S\| \right)  +   3 d^2\rho^4 \left( \delta + 2 d \rho^2  \|S\|   \right)^2\right] .
\end{align*}
Combining with Theorem~\ref{thm:general_bound} and making the specification
\begin{align}
\beta_n = n^{-1/5} d\rho^2  \|S\|, \quad \delta_n = n^{-2/5} d\rho^2  \|S\|,
\end{align}
leads to the stated result. 
%\begin{align}
%\big | \cI_n(R, S)   -   \hat{\cI}_n(R, S) \big |
%%\le  \frac{\beta d\rho^2}{2}  +\|S\| \left[  \sqrt{ \frac{2   d^3 \rho^6 }{ \beta n}   }  + \frac{2 \delta d^2 \rho^4 }{ 2\beta} +  \frac{2 d^2}{ \delta^2 }  \overline{\cV}_n \right] \\
%%  \le  \frac{C}{n^{1/5} }  \bigg[  d^2\rho^4 \|S\|   +\sqrt{ \frac{  d^2 \rho^4 \|S\|  } { n^{2/5} }   }     +  \frac{  d \rho^2 \left( d \rho^2 \|S\| \right)  +    d^2 \rho^4 \left( d \rho^2  \|S\|   \right)^2 }{  \rho^4  \|S\|   }   \bigg]\\
%%  \le  \frac{C}{n^{1/5} }  \bigg[  d^2\rho^4 \|S\|   +\sqrt{ \frac{  d^2 \rho^4 \|S\|  } { n^{2/5} }   }  + d^2     + d^4 \rho^4  \|S \|  \bigg]
%     \le  \frac{C d^2 }{n^{1/5} }  \left(    1    + d^2 \rho^4  \|S \|   \right) .
%\end{align}

\subsection{Proof of Theorem~\ref{thm:general_bound}}\label{proof:thm:general_bound}

The proof has two main steps. First, we use adaptive interpolation to bound the difference in terms of a remainder term. Then, we show that under an appropriate perturbation, this error can be made small. The first step is summarized by the following result, which is proved in Section~\ref{proof:prop:approx_bound}.

  \begin{prop}\label{prop:approx_bound} Assume that $(\ex{ \|\bX\|^2})^{1/2}  \le  \sqrt{n} \, \rho$. %  Assume that $\bX$ is an $n \times d$ random matrix with finite second moment and $\bZ \in \cZ$ is an observation of $\bX$ defined on the same space. 
For each $(R,S)  \in \psd^d \times \psd^{d^2}$ there exists a function $\Phi \colon [0,1] \times \psd^d \to \psd^d$ with the following properties: 
\begin{enumerate}[(i)]
\item For each $\eps \in \psd^d$, 
\begin{align}
 \big | \cI_n(R, S)   -   \hat{\cI}_n(R, S) \big |  \le  \frac{d \rho^2 \|\eps\|   }{2}    + \frac{ \|S\|}{2n^2}  \int_0^1 \ex{ \| \bO_{t,\eps} - \ex{ \bO_{t, \eps}} \|_\fro^2}  \, \dd t, \label{eq:approx_bound}
\end{align}
where $\bO_{t, \eps}$ is the overlap matrix associated with the observations $(\bY_{R + \Phi(t,\eps), (1-t) S}, \bZ)$.

\item For each $t \in [0,1]$, the mapping $\eps \mapsto \Phi(t, \eps)- \eps$ is order-preserving. In other words, 
\begin{align*}
\eps_1\preceq \eps_2 \quad \implies \quad  \Phi(t, \eps_1) - \eps_1 \preceq   \Phi(t, \eps_2) -\eps_2.
\end{align*}
Furthermore, $0\preceq \Phi(t,\eps)  -\eps \preceq  2 d\rho^2 \|S\| I_d$. 
\end{enumerate}
\end{prop}

%\begin{figure*}[!b]
%% ensure that we have normalsize text
%\normalsize
%% Store the current equation number.
%\setcounter{MYtempeqncnt}{\value{equation}}
%% Set the equation number to one less than the one
%\hrulefill
%\setcounter{equation}{54}
%\begin{align}
% \cI_n( \phi(1) , 0) -  \cI_n(\phi(0) , S )  =  \frac{1}{2 n }   \int_0^1\left\{   \langle \phi'(t) , \MMSE(\bX \mid \cG_t)  \rangle - \frac{1}{n}  \left \langle  S ,  \MMSE(\bX^{\otimes 2}  \mid \cG_t)  \right \rangle  \right\} \, \dd t . \label{eq:In_interp}
%\end{align}
%% Restore the current equation number.
%\setcounter{equation}{\value{MYtempeqncnt}}
%% IEEE uses as a separator
%%\hrulefill
%% The spacer can be tweaked to stop underfull vboxes.
%%\vspace*{4pt}
%\end{figure*}

Proposition~\ref{prop:approx_bound} provides a family of  bounds indexed by $\eps$.  Using the results in Section~\ref{sec:overlap} we can bound the average with respect to a suitable distribution on $\eps$. 
Specifically, for $\beta > 0$ let $\mu_\beta$ be the probability measure of the random matrix in \eqref{eq:eps_random}. Then, for each $t \in [0,1]$, we can apply Theorem~\ref{thm:overlap_concentration} with $\psi_t(\eps) = \Phi(t,\eps) - \eps$ and  $\bZ_{t} = (\bY_{R,(1-t) S}, \bZ)$ to obtain
 \begin{align}
\frac{1}{n^2}  \int  \ex{ \| \bO_{t,\eps} - \ex{ \bO_{t, \eps}} \|_\fro^2}  \, \mu_{\beta}(\eps) 
  \le  2 \sqrt{ \frac{2   d^3 \rho^6 }{ \beta n}   }  + \frac{4 \delta d^2 \rho^4 }{\beta } + \frac{ 4 d^2 }{ \delta^2 n^2}  \overline{V}_t,  \label{eq:overlap_concentration2}
\end{align}
where
\begin{align*}
\overline{V}_t & =   \max_{\eps, \tilde{\eps}   \in \psd^d \, : \, \|\eps\| \le \gamma , \|\tilde{\eps} \| \le \delta }   V\left( P_{\tilde{\bY}_{\tilde{\eps}}, \bY_{\eps   }, \bZ_{t} }  \,  \| \,  P_{\tilde{\bW}, \bY_{\eps  } , \bZ_{t} }  \right),
\end{align*}
with $\gamma  =  \beta  + 2 d \rho^2 \|S\|$.  Finally, recognizing that $(\bY_{\eps}, \bZ_t)$ can be replaced by $(\bY_{R + \eps, (1-t) S}, \bZ)$ yields
\begin{align*}
\frac{1}{n^2} \overline{V}_t \le  \max_{\eps, \tilde{\eps}   \in \psd^d \, : \, \|\eps\| \le \gamma , \|\tilde{\eps} \| \le \delta } \cV_n(\tilde{\eps}  \mid R + \eps , (1-t) S).
\end{align*}
Taking the maximum of this bound with respect to $0 \le t \le 1$ and combining \eqref{eq:approx_bound} and \eqref{eq:overlap_concentration2} completes the proof of Theorem~\ref{thm:general_bound}.

\section{Adaptive interpolation: Proof of Proposition~\ref{prop:approx_bound}}\label{proof:prop:approx_bound}

In view of the reduction described in \eqref{eq:MI_reduction} we will assume without loss of generality that $R =0$.  Also, we use $\cT, \cT^*$ to denote the linear operators defined as a function of $S$ according to Lemma~\ref{lem:lin_op}. Recall that for $M_1, M_2 \in \sym^d$,  $\langle S, M_1 \otimes M_2 \rangle = \langle \cT(M_1), M_2 \rangle = \langle M_1 ,\cT^*(M_2) \rangle$.

\subsection{Interpolation}
We begin with a decomposition of the mutual information in terms of a generic interpolation function.  Let $\phi \colon [0,1] \to \psd^d$ be given by  $\phi(t) =  \phi(0) + \int_0^t \phi'(u)\, \dd u$ for some integrable function $\phi': [0,1] \to \psd^d$. Then, the interpolation path is described by the matrix pair $(\phi(t), (1-t) S)$, which transitions from the point $(\phi(0),S)$ at time $t=0$ to the point $(\phi(1), 0)$ at time $t=1$.

For $t \in [0,1]$, let $\cG_t$  be the sigma-algebra generated by the observations $(\bY_{\phi(t), (1-t) S}, \bZ)$. By the matrix version of the I-MMSE relation \eqref{eq:I_RS_gradient} and the fundamental theorem of calculus, the change in mutual information along the path can be expressed as% in \eqref{eq:In_interp}. \stepcounter{equation}
\begin{align}
 \cI_n( \phi(1) , 0) -  \cI_n(\phi(0) , S )  =  \frac{1}{2 n }   \int_0^1\left\{  \langle \phi'(t) , \MMSE(\bX \mid \cG_t)  \rangle - \left \langle \tfrac{1}{n}  S ,  \MMSE(\bX^{\otimes 2}  \mid \cG_t)  \right \rangle  \right\} \, \dd t . \label{eq:In_interp}
\end{align}
For the purposes of analysis, it is useful to restate this expression with respect to the relative entropy formulation. We use $\bO_t$ to denote the overlap matrix of $\bX$ associated with $\cG_t$ and we define the normalized expectation
\begin{align}
Q_t \coloneqq \frac{1}{n} \ex{\bO_t}.% = 2 \nabla_R \cD_n(\phi(t), (1- t) S) .
% \frac{1}{n} \left( \ex{ \bX^\top \bX} - \MMSE(\bX \mid \cG_t)  \right).
\end{align}
By the property $(A^{\otimes 2} )^\top B^{\otimes 2 } = (A^\top B)^{\otimes 2}$ the overlap matrix corresponding to $\bX^{\otimes 2}$ can then  be expressed as $\bO_t^{\otimes 2}$. Using this notation, \eqref{eq:In_interp} can be rewritten according to
\begin{align*}
 \cD_n( \phi(0) , S) -  \cD_n(\phi(1) , 0 ) & =  \frac{1}{2n}   \int_0^1\left\{  \left \langle \tfrac{1}{n} S , \ex{ \bO_t^{\otimes 2} }  \right \rangle -  \langle \phi'(t) , \ex{ \bO_t}  \rangle  \right\} \, \dd t . \\
& =  \frac{1}{2}   \int_0^1\left\{  \left \langle S , Q_t^{\otimes2}   \right \rangle -  \langle \phi'(t) , Q_t  \rangle  \right\} \, \dd t  + \cR(\phi) . %\\
%& \quad +   \frac{1}{2n^2}   \int_0^1  \left \langle \tfrac{1}{n} S , \ex{ \bO_t^{\otimes 2} } -   \ex{ \bO_t}^{\otimes 2}  \right \rangle  \rangle  \, \dd t .
\end{align*}
In the second line, we have added and subtracted the term $ \left \langle S , Q_t^{\otimes2}   \right \rangle$ and so the remainder is given by
%where the remainder term is given by
 \begin{align*}
 \cR(\phi) \coloneqq\frac{1}{2n^2}   \int_0^1  \left \langle S , \ex{ \bO_t^{\otimes 2} } -   \ex{ \bO_t}^{\otimes 2}  \right \rangle  \, \dd t .
 \end{align*}
Recalling the shorthand notation  $\cD_n(R) \coloneqq \cD_n(R, 0)$ and also making the substitution $\langle S, Q_t^{\otimes 2} \rangle = \langle \cT(Q_t) , Q_t\rangle $ leads to the decomposition 
\begin{align}
\cD_n(\phi(0),S) = \cA(\phi) + \cR(\phi) \label{eq:cD_AR}
\end{align}
where 
\begin{align}
\cA(\phi) \coloneqq  \cD_n(\phi(1)  ) +  \frac{1}{2}   \int_0^1  \left \langle \cT(Q_t) - \phi'(t)   , Q_t  \right \rangle  \, \dd t . \label{eq:cA_alt}
\end{align}

This decomposition holds  generally for any interpolation function $\phi$. The basic idea in the adaptive interpolation method is to specify the interpolation function recursively as a function of the expected overlap in such a way that the term $\cA(\phi)$ can be identified with the desired  single-letter formula. In view of \eqref{eq:cA_alt} one might be tempted to consider the specification $\phi'(t) = \cT(Q_t)$ which cancels the term in the integral, but this specification does not yield the desired result. Instead, we will consider paths satisfying the property
\begin{align}
\phi'(t) = \cT(Q_t )  + \cT^*( \psi'(t))
\end{align}
for some integrable function $\psi' \colon [0,1] \to \psd^d$. Evaluating \eqref{eq:cA_alt} under this choice leads to 
\begin{align}
\cA(\phi)  =\cD_n\left(\phi(0) + \cT(  {\textstyle \int_0^1 Q_t \, \dd t} )  + \cT^*( \psi(1))  \right )  
  -  \frac{1}{2}  \int_0^1 \langle \cT(Q_t ),  \psi'(t)   \rangle  \, \dd t  ,%+ \remainder(S, \phi) , 
\label{eq:cA_alt2} 
\end{align}
where $\psi(1) \coloneqq \int_0^1 \psi'(t)\, \dd t$. Note that this expression is reminiscent of  the approximation formula given in  Proposition~\ref{prop:Dh_RS_alt}.

\subsection{Lower bound} 

We construct an interpolation function that depends on $S \in \psd^{d^2}$, an initial value $\eps \in \psd^d$, and an auxiliary parameter $\psi  \in \cQ$, whose value will be specified later. Consider the initial value problem
\begin{align}
\phi'(t)  =  \cT\left(  Q_t \right)  + \cT^*(\psi), \qquad \phi(0) = \eps, \label{eq:IVP_LB}
\end{align}
where $\cT$ and $\cT^*$ are defined as a function of $S$ according to Lemma~\ref{lem:lin_op}.  This is a first order ordinary differential equation because $Q_t$ is a function of  $\phi(t)$. The existence and uniqueness of a solution are established in the following result.

\begin{lemma}\label{lem:IVP_LB} The initial value problem  \eqref{eq:IVP_LB} has a unique solution $\phi_{\eps,\psi}(t)$ defined on $[0,1]$. Furthermore, for  $t \in [0,1]$,
\begin{enumerate}[(i)]
\item the derivative is positive semidefinite and bounded according to $\| \phi'_{\eps, q}(t) \| \le 2 \|S \| \ex{ \frac{1}{n} \|\bX\|_\fro^2}$,
\item % For each $t \in [0,1]$, 
the mapping $(\eps, \psi)  \mapsto \phi_{\eps,\psi}(t) - \eps$ is  continuous and order-preserving. 
 \end{enumerate}
\end{lemma}
\begin{proof}
Let  $\cM  \colon  [0,1] \times \psd^d \to  \psd^d$  denote the normalized expected overlap associated with the observations $(\bY_{R, (1-t) S} , \bZ)$, i.e., %be defined by%according to 
\begin{align}
\cM(t, \tilde{R}) \coloneqq 2 \nabla_R \cD_n( \tilde{R}, (1-t) S) . \label{eq:cM}
%\cM(t, R) \coloneqq \frac{ \ex{ \bX^\top \bX} -  \MMSE\left(\bX \mid \bY_{R, (1-t) S}  ,\bZ  \right)}{n} \label{eq:cM}
%\cM(t, r) : = \MMO\left(\bX \mid \bX r^{1/2} + \bW,   \sqrt{1-t} \bX^{\otimes 2}S^{1/2}  + \bW', \bZ \right). \label{eq:cM}
\end{align}
Then, $Q_t = \cM(t, \phi(t) )$ and the initial value problem in \eqref{eq:IVP_LB} can be restated as %\nbr{Maybe change notation since $F$ is used for free energy}
\begin{align}
\phi'(t)  =  F(t, \phi(t),\psi) , \qquad \phi(0) = \eps, \label{eq:IVP_LB_alt}
\end{align}
where $F : [0,1] \times \psd^d  \times \cQ \to \psd^d$ is given by $F(t, \phi, \psi) = \cT\left( \cM(t, \phi) \right)  + \cT^*(\psi)$. The function $F$ is differentiable on the interior of its domain by the I-MMSE relation. Furthermore,  by  Lemma~\ref{lem:lin_op}, it satisfies the bound 
\begin{align*}
\| F(t, \phi, q) \| & =  \|  \cT( \cM(t, \phi) ) + \cT^*(\psi)  \| \notag\\
&  \le \|S\| ( \|\cM(t, \phi) \|_*  + \|\psi\|_*) \\
& \le  2 \|S\|  \ex{\tfrac{1}{n} \|\bX\|_\fro^2},
\end{align*} 
where we have used the fact that $\psi \in \cQ$. 
 %\nbr{need to connect operator norm to norm of $S$}.  
Therefore, by the Picard--Lindel\"{o}f theorem~\cite[Chapter II, Theorem~1.1]{hartman:2002} there exists a unique solution $\phi_{\eps,\psi}(t)$  defined on the interval  $[0,1]$. Furthermore, $\phi_{\eps,\psi}(t)$ is a continuous function of the parameters  $(\eps, \psi)$  \cite[Chapter V, Theorem~3.1]{hartman:2002}.

Next, observe that  $r \mapsto \cM(t,r)$ is order-preserving by the data processing inequality for MMSE (Proposition~\ref{prop:DPI}), and  $\cT , \cT^*$ are order-preserving by Lemma~\ref{lem:lin_op}. Therefore, 
 $(\phi,\psi) \to F(t, \phi, \psi)$ is order-preserving because it is the composition of order-preserving maps.  For any $(\eps_1, \psi_1 )\preceq  (\eps_2, \psi_2)$ we can apply the comparison result in Lemma~\ref{lem:ODE_comp} to conclude that $\phi_{\eps_1,\psi_1}(t) - \eps_1 \preceq \phi_{\eps_2, \psi_2}(t) - \eps_2$ for all $t \in [0,1]$, and this establishes the order-preserving in part (ii). 
 \end{proof}

Having established the existence of the path satisfying \eqref{eq:IVP_LB}, we may consider the representation in \eqref{eq:cA_alt2}, which becomes:
 \begin{align*}
\cA( \phi_{\eps, \psi}  )& =\cD_n\left( \eps + \cT\left({\textstyle \int_0^1 Q_t \, \dd t}\right) +\cT^*(\psi ) \right )  -  \frac{1}{2} \left \langle \cT\left( {\textstyle \int_0^1Q_t \, \dd t}\right), \psi  \right \rangle.
\end{align*}
In this expression, $\cT(\int_0^1 Q_t \, \dd t)$ is a positive semidefinite matrix that is defined implicitly by the interpolation path parameterized by $(\eps, \tilde{Q} )$. Minimizing over all possible values of this matrix 
leads to the following lower bound
 \begin{align}
%\cD 
\cA( \phi_{\eps,\psi} ) & \ge \inf_{\tilde{R} \in \psd^d} \left\{ \cD_n(  \eps  + \tilde{R} +\cT^*(\psi) )  -  \frac{1}{2} \langle \tilde{R}  , \psi \rangle  \right\}  \notag \\
& \ge \inf_{\tilde{R} \in \psd^d} \left\{ \cD_n(  \tilde{R}  )   -  \frac{1}{2} \langle \tilde{R}  , \psi \rangle  +\frac{1}{2} \langle \eps + \cT^*(\psi) , \psi \rangle \right\}  \notag \\
& \ge \inf_{\tilde{R} \in \psd^d} \left\{ \cD_n(  \tilde{R}  )   -  \frac{1}{2} \langle \tilde{R}  , \psi \rangle  +\frac{1}{2} \langle  \cT^*(\psi) , \psi \rangle \right\},  \label{eq:cD_LB_c}
%& \ge \inf_{\tilde{R} \in \psd^d} \left\{ \cD_n(  \tilde{R} +\cT^*(\psi) )  -  \frac{1}{2} \langle \tilde{R}  , \psi \rangle  \right\} , \label{eq:cD_LB_c}
%&  \ge \inf_{\tilde{R} \in \sym^d} \left\{\cD_n(  \eps  + \tilde{R} +\cT^*(q) )  -  \frac{1}{2} \langle \tilde{R}  , q \rangle  \right\} \\
%&  =\frac{1}{2}  \langle  \eps + \cT^*(q)  , q\rangle  +   \inf_{\tilde{R} \in \sym^d} \left\{ D_0(   \tilde{R}  )  -  \frac{1}{2} \langle \tilde{R}  , q \rangle  \right\} \\
%& =\frac{1}{2}  \langle  \eps  , q\rangle +  \langle S, q^{\otimes 2} \rangle - D_0^*\left( q  \right) ,\label{eq:cD_LB_c}
\end{align}
where the second step holds because we have enlarged the set over which the infimum is taken by making the change of variables $\tilde{R}  \mapsto \tilde{R} - \cT^*(Q)$, and the last step holds because both $\eps$ and $\psi$ are positive semidefinite. 
%where the second step holds because $\hat{\cD}_n$ is order preserving. 

This bound holds for every choice of the parameter $\psi$.  By the compactness of $\cQ$ and the  upper-semicontinuity of \eqref{eq:cD_LB_c}, it follows from the extreme value theorem~\cite[Theorem~1.29]{bauschke:2017} that the maximum over $\psi$  is attained at some point $\psi^\star \in  \cQ$. Evaluating \eqref{eq:cD_LB_c} at this point and recalling the definition of $\hat{\cD}_n$ in \eqref{eq:Dh_RS} yields $ \cA(  \phi_{\eps,\psi^\star} )  \ge \hat{\cD}_n(0, S)$. In view of  \eqref{eq:cD_AR} this implies the lower bound
% \begin{align}
%\cA(  \phi_{\eps,Q^\star} ) 
%%& \ge \adjustlimits \max_{Q \in \cQ}  \inf_{\tilde{R} \in \psd^d} \left\{ \cD_n(  \tilde{R} +\cT^*(Q) )  -  \frac{1}{2} \langle \tilde{R}  , Q  \rangle  \right\} \\
%& \ge \hat{\cD}_n(0, S).  \label{eq:cD_LB_g}
%\end{align}
 \begin{align*}
\cD_n(\eps , S) )  & \ge \hat{\cD}_n(0, S) + \cR( \phi_{\eps}), 
\end{align*}
where $\psi_\eps \coloneqq \psi_{\eps, \psi^\star}$. Combining with the bound $\cD_n(\eps, S) \le \cD_n(\eps, S)  + \frac{1}{2} \langle \eps, \ex{\tfrac{1}{n} \bX^\top \bX}\rangle$, which follows from the convexity of $\cD_n$ and the I-MMSE relation, we arrive at
 \begin{align}
\cD_n(0 , S)   & \ge \hat{\cD}_n(0, S) - \frac{1}{2} \langle \eps, \ex{ \tfrac{1}{n} \bX^\top \bX}\rangle + \cR( \phi_{\eps}).  \label{eq:cD_LB_g}
\end{align}

\subsection{Upper bound} 

Similar to the previous section, we begin by defining an interpolation function that depends on the pair $(S, \eps)$ and an auxiliary  parameter $\eta \in \cQ$, whose value will be specified later. Consider the initial value problem with coupled equations:
\begin{subequations}
\label{eq:IVP_UB}
\begin{alignat}{3}
\phi'(t)& =  \cT( Q_t  )  +  \cT^*( \psi'(t)), & \qquad & \phi(0) = \eps \\
\psi'(t)& = 2 \nabla \cD_n(  \eps + \cT( Q_t )  +  \cT^*(\eta)) , & \qquad & \psi(0) = 0.\label{eq:IVP_UBb}
 %\cM\left( 1,  \eps + \cT( Q_t )  +  \cT^*(\eta) \right), & \qquad & \psi(0) = 0.
\end{alignat} 
\end{subequations}
%where we recall that by I-MMSE relation \eqref{eq:I_RS_gradient} and \eqref{eq:OtoMMSE}, the gradient of the relative entropy in the linear channel satisfies $2 \nabla \cD_n(R)  = \ex{ \frac{1}{n} \ex{ \bX^\top \bX}} - \frac{1}{n} \MMSE( \bX \mid \bY_R, \bZ)$.  
This is a coupled system of first order ordinary differential equations because $Q_t$ is a function of $\phi(t)$. %\nbr{I am tempted to change $\tilde{q}$ to $\psi$ to improve readability}

\begin{lemma}\label{lem:IVP_UB} The initial value problem  \eqref{eq:IVP_LB} has unique solutions $\phi_{\eps,\eta}(t)$ and $\psi_{\eps,\eta}(t)$ defined on $[0,1]$. Furthermore, for each $t \in [0,1]$,
\begin{enumerate}[(i)]
\item the derivative is positive semidefinite and bounded according to  $\| \phi'_{\eps, q}(t) \| \le 2 \|S \| \ex{ \frac{1}{n} \|\bX\|_\fro^2}$,
% $ 0 \preceq \phi'_{\eps, \eta}(t) \preceq \cT(\ex{ \frac{1}{n} \bX^\top \bX}) + \cT^*(\ex{ \frac{1}{n} \bX^\top \bX}) $,
\item the mappings $(\eps, \eta)  \mapsto \phi_{\eps,\eta}(t) - \eps$  and $(\eps, \eta)  \mapsto \psi_{\eps,\eta}(t)$  are  continuous and order-preserving. 
 \end{enumerate}
\end{lemma}
\begin{proof}
Let $\cM(t, \psi)$ be defined as in \eqref{eq:cM}. The initial value problem in \eqref{eq:IVP_UB} can be expressed as
\begin{subequations}
\begin{alignat*}{3}
\phi'(t)& =F_1( t, \phi(t),  \eps, \eta) , & \qquad & \phi(0) = \eps \\
\psi'(t)  &= F_2(t, \phi(t), \eps, \eta) ,  && \psi(0) = 0% \\
\end{alignat*} 
\end{subequations}
where $F_k : [0,1] \times \psd^d \times \psd^d \times \cQ \to \psd^d$ is given by
\begin{subequations}
\begin{align*}
F_1(t, \phi,\eps, q) &=  \cT(  \cM(t, \phi)  )  +  \cT^*( F_2(t, \phi, \eps,  \eta))\\
F_2(t, \phi,\eps, q) &= 2 \nabla \cD_n\left(\eps +  \cT( \cM(t, \phi) )  + \cT^*(\eta) \right).
\end{align*}
\end{subequations}
Following the same arguments used in the proof of Lemma~\ref{lem:IVP_LB}, one finds that these functions are differentiable on the interior of their domains and satisfy the bounds $\|  F_1(t, \phi, q) \| \le 2 \|S\| \ex{ \frac{1}{n} \|\bX\|_\fro^2}$ and $\| F_2(t, \phi, q) \| \le \| \ex{\tfrac{1}{n}  \bX^\top \bX}\|$. Therefore, by the Picard--Lindel\"{o}f theorem \cite[Chapter II, Theorem~1.1]{hartman:2002}  there exist  unique solutions $\phi_{\eps,	\eta}(t) $  and $\psi_{\eps,\eta}(t) $  defined on the interval  $[0,1]$. Furthermore, these functions are continuous in $(t, \eps,\eta)$  \cite[Chapter V, Theorem~3.1]{hartman:2002}.

Next,  for $t  \in [0,1]$ and $k =1,2$, the mapping  $(\phi, \eps, \eta) \mapsto F_k(t, \phi, \eps, \eta)$ is order-preserving because it is the composition of order-preserving maps. The order-preserving  properties of $\phi_{\eps,\eta}$ and $\psi_{\eps,\eta}$ then follow from Lemma~\ref{lem:ODE_comp} using the same approach as in the proof of Lemma~\ref{lem:IVP_LB}.
\end{proof}

Starting with \eqref{eq:cA_alt2} and introducing the additional term   $\tilde{R}_{\eps,\eta}'(t)  = \eps +  \cT(Q_t)   +\cT^*( \eta)$, we can write 
 \begin{align}
\cA( \phi_{\eps, \eta}  )
& =\cD_n\left(\phi_{\eps, \eta}(1) \right)  + \frac{1}{2} \langle \eps  + \cT^*(\eta), \psi_{\eps, \eta}(1) \rangle \notag\\
& \quad  -   \frac{1}{2} \int_0^1 \langle \tilde{R}_{\eps,\eta}'(t),  \psi'_{\eps, \eta}(t)  \rangle  \, \dd t ,  \label{eq:interp_UB_c}
\end{align}
where $\psi_{\eps, \eta}(1)   =\int_0^1 \psi_{\eps, \eta}'(t) \, \dd t$. For the path specified by \eqref{eq:IVP_UB}, the terms  $\tilde{R}'_{\eps,\eta}(t)$ and $\psi'_{\eps, \eta}(t)$ are related according to $\psi'_{\eps,\eta}(t) = 2 \nabla \cD_n(\tilde{R}'_{\eps, \eta}(t))$. Consequently, the inner product can be expressed as 
 \begin{align}
\frac{1}{2}   \langle \tilde{R}_{\eps,\eta}'(t) , \psi_{\eps,\eta}'(t) \rangle & =    \cD_n(\tilde{R}_{\eps,\eta}'(t) ) +   \cD_n^*(\tilde{q}_{\eps,\eta}'(t) ),  \label{eq:interp_UB_d}
\end{align}
where $\cD_n^*(Q) = \sup_{R \in \psd^d} \left\{ \frac{1}{2} \langle R, Q \rangle - \cD_n(R) \right\}$ is the convex conjugate of $\cD_n(R)$. 

We now consider two bounds due to Jensen's inequality. First,  by the convexity of $\cD_n^*$, 
\begin{align}
\int_0^1 \cD_n^*(\psi_{\eps,\eta}'(t) ) \, \dd t  \ge  \cD_n^*( \psi_{\eps, \eta}(1) ) .  \label{eq:interp_UB_e}
\end{align}
Second, by the convexity of $\cD_n$, 
\begin{align}
\cD_n(\phi_{\eps,\eta}(1)) & = \cD_n\left(\eps +{ \textstyle \int_0^1  \cT(Q_t  ) \, \dd t } + \cT^*(\psi_{\eps,\eta}(1)) \right) \notag \\
& \le \int_0^1 \cD_n \big(\eps + \cT( Q_t )+ \cT^*(\psi_{\eps,\eta}(1)) \big) \, \dd t \label{eq:interep_UB_e}\\
&  =  \int_0^1 \cD_n \left(\tilde{R}_{\eps,\eta}'(t) + \cT^*(\psi_{\eps,\eta}(1) - \eta) \right) \, \dd t.  \label{eq:interp_UB_f}
\end{align}
Combining \eqref{eq:interp_UB_c},  \eqref{eq:interp_UB_d},  \eqref{eq:interp_UB_e}, and \eqref{eq:interp_UB_f} yields
 \begin{align}
 \cA( \phi_{\eps, \eta}  ) & \le  \frac{1}{2} \langle \eps  + \cT^*(\psi_{\eps, \eta}(1)), \psi_{\eps, \eta}(1) \rangle-  \cD_n^*( \psi_{\eps, \eta}(1) )  \notag \\
%& \le  \frac{1}{2} \langle \eps  + \cT^*(\eta), \psi_{\eps, \eta}(1) \rangle-  \cD_n^*( \psi_{\eps, \eta}(1) ) \notag \\
%& \quad + \int_0^1 \left\{ \cD_n \left(\tilde{R}_{\eps,\eta}'(t) + \cT^*(\psi_{\eps,\eta}(1) - \eta) \right)  -   \cD_n(\tilde{R}_{\eps,\eta}'(t) )  \right\}  \, \dd t \\
%&  \frac{1}{2} \langle \eps  + \cT^*(\psi_{\eps, \eta}(1)), \psi_{\eps, \eta}(1) \rangle-  \cD_n^*( \psi_{\eps, \eta}(1) ) \notag\\
& + \frac{1}{2} \langle \cT^*(\psi_{\eps,\eta}(1) - \eta)  , \psi_{\eps, \eta}(1)  \rangle  \notag \\
& + \int_0^1 \left\{ \cD_n \big(\tilde{R}_{\eps,\eta}'(t) + \cT^*(\psi_{\eps,\eta}(1) - \eta) \big )  -    \cD_n(\tilde{R}_{\eps,\eta}'(t) )  \right\}  \, \dd t .
 \label{eq:interp_UB_g}
\end{align}
Recalling the definition of $\cD_n^*$, the first term in \eqref{eq:interp_UB_g} satisfies
\begin{align}
\MoveEqLeft 
  \frac{1}{2} \langle \eps  + \cT^*(\psi_{\eps, \eta}(1)), \psi_{\eps, \eta}(1) \rangle-  \cD_n^*( \psi_{\eps, \eta}(1) )  \notag \\
& =   \inf_{\tilde{R} \in \psd^d} \left\{\cD_n( \tilde{R} ) +   \frac{1}{2} \langle \eps  -  \tilde{R} , \psi_{\eps, \eta}(1)  \rangle + \frac{1}{2}\langle S , \psi_{\eps, \eta}(1)^{\otimes 2} \rangle \right\}  \notag \\
& \le   \inf_{\tilde{R} \in \psd^d} \left\{\cD_n( \tilde{R} ) -   \frac{1}{2} \langle   \tilde{R} , \psi_{\eps, \eta}(1)  \rangle + \frac{1}{2}\langle S , \psi_{\eps, \eta}(1)^{\otimes 2} \rangle \right\}   + \frac{1}{ 2} \langle \eps, \ex{\tfrac{1}{n}  \bX^\top \bX } \rangle \notag \\
& \le   \hat{\cD}_n(0, S) + \frac{1}{ 2} \langle \eps, \ex{\tfrac{1}{n}  \bX^\top \bX } \rangle, \label{eq:interp_UB_h}
\end{align}
where the third step holds because $ \psi_{\eps, \eta}(1)  \preceq \frac{1}{n} \ex{\bX^\top \bX}$ and the last step follows from the definition of $\hat{\cD}_n$ in  \eqref{eq:Dh_RS}.

The second and third terms in  \eqref{eq:interp_UB_g} are equal to zero whenever $\psi_{\eps, \eta}(1)  = \eta$. Thus, to specify an interpolation function with the desired properties we need to show that there exists $\eta^\star \colon \psd^d \to \cQ$  such that:
\begin{enumerate}[(i)]
\item  $\eta^\star(\eps)$ is the solution $\eta$ to the fixed point equation  $\psi'_{\eps, \eta}(1)  = \eta$ and thus $\cA(\phi_{\eps, \eps^\star(\eta)})$ is bounded from above by \eqref{eq:interp_UB_h};
\item  $\eps \mapsto \eta^\star(\eps)$ is order-preserving and thus, by Lemma~\ref{lem:IVP_UB},  $ \eps \mapsto   \phi_{\eps, q^\star(\eps)}(t) - \eps$ is order-preserving. 
\end{enumerate}
The existence of such a mapping is established in the following result. Let  $\Psi(\eps, \eta) \coloneqq \psi_{\eps,\eta}(1)  = \int_0^1 \psi'_{\eps, \eta}(t) \, \dd t$ where $\psi'_{\eps,\eta}(t)$ is the solution to the initial value problem \eqref{eq:IVP_UB} and define the fixed-point set  $\FP(\eps) =  \{ \eta \in \cQ \, : \,  \Psi(\eps, \eta) = \eta \}$.

 \begin{lemma}\label{lem:FP}
 For all  $\eps \in \psd^d$, the set $\FP(\eps)$ is nonempty and there exists a greatest element $\eta^\star(\eps)$ such that $\eta \preceq \eta^\star(\eps)$ for all $\eta \in \FP(\eps)$. Furthermore, the mapping $\eta^\star \colon \psd^d \to \cQ$ is order-preserving in the Loewner order. 
\end{lemma}
\begin{proof}
 For each $\eps \in \psd^d$, the mapping  $\eta \mapsto \Psi(\eps,\eta) $ is continuous by Lemma~\ref{lem:IVP_UB} and maps a convex compact subset of Euclidean space to itself.  Therefore,  we can apply Brouwer's fixed-point theorem to conclude that $\FP(\eps)$ is nonempty. 

Next we show that  there is a greatest element. If $\FP(\eps)$ is single-valued then we are done. Otherwise, define the set $K(\eps) = \bigcap_{q \in \FP(\eps)} \{ \eta \in \cQ \, : \, \eta \succeq  q  \}.$ 
This set is nonempty because $\cQ$ has a largest element and it is convex and compact because it is the intersection of convex and compact sets. The mapping $(\eps, \eta) \mapsto \psi_{\eps,\eta}(t)$ is order-preserving by  Lemma~\ref{lem:IVP_UB} and this implies that $(\eps, \eta) \mapsto \Psi(\eps,\eta)$ is order-preserving. Consequently,  $\eta  \in K(\eps)$ implies that $\Psi(\eps,\eta) \succeq \eta$ for all $ \eta \in \FP(\eps)$, and so  $\eta \mapsto \Psi(\eps,\eta)$ maps $K(\eps)$ to itself. By  Brouwer's fixed-point theorem,  there must be a solution on $K(\eps)$. From the definition of $K(\eps)$, this solution is the greatest element of $\FP(\eps)$, which is necessarily unique. 

Finally, we show that the greatest element, denoted by  $\eta^\star(\eps)$, is order-preserving.  Fix any $\eps_1, \eps_2 \in \psd^d$  with $\eps_1 \preceq \eps_2$ and let $\eta_k = \eta^\star(\eps_k)$ for $k=1, 2$.   For each $\eta \succeq \eta_1$, the fact that $\Psi(\eps,\eta)$ is order-preserving means that $\Psi(\eps_2, \eta) \succeq  \Psi(\eps_1, \eta_1) = \eta_1$ and so $\eta \mapsto \Psi(\eps_2, \eta)$ maps the closed convex set $\{ \eta \in \cQ \, : \, \eta \succeq \eta_1\}$ to itself. By Brouwer's fixed-point theorem, there exists a solution on this set and by the assumption that $\eta_2$ is the greatest element it follows that $\eta_1 \preceq \eta_2$. 
\end{proof}

In view of Lemma~\ref{lem:FP}, we may consider the interpolation path given by $\psi_{\eps} \coloneqq \psi_{\eps, \eta^\star(\eps)}$. By  \eqref{eq:cD_LB_g} and \eqref{eq:interp_UB_g}, we obtain the upper bound
 \begin{align}
\cD_n(0, S) )  & \le \hat{\cD}_n(0, S)+ \frac{1}{ 2} \langle \eps, \ex{\tfrac{1}{n}  \bX^\top \bX } \rangle + \cR( \phi_{\eps}).  \label{eq:cD_UB_h}
\end{align}

\begin{remark}
In comparison to previous incarnations of the adaptive interpolation method\cite{barbier:2019c,barbier:2018,barbier:2019f, barbier:2019e, barbier:2020a}, the new components introduced in this paper include: 1) the use of general positive semidefinite matrix-valued perturbations satisfying the order-preserving property, 2)  the introduction of the coupled first order differential equations to define the adaptive path for the upper bound;  and 3) the fixed-point arguments used to establish the existence of the order-preserving mapping $\eta^\star(\eps)$. Moreover, a crucial step in the bounding technique is that the Jensen's inequality in \eqref{eq:interp_UB_e} is applied to the overlap but not the auxiliary path $\psi'_{\eps,\eta}(t)$, and this allows for the cancellation of the second and third terms in \eqref{eq:interp_UB_g}.

The order-preserving assumption can be seen as a natural extension of regularity conditions used in earlier versions of the adaptive interpolation method (see e.g.,~\cite[Supplement, Definition~1]{barbier:2019}), which require that the mapping $\eps \mapsto  \Phi(t, \eps)$ is a $\cC^1$ diffeomorphism, whose Jacobian is greater  than or equal to 1. Unlike the conditions used previously, the order-preserving property does not require continuity. 
\end{remark}

\subsection{Bounding the remainder} 

In the previous two sections, we showed that for all $S$, there exist functions $\phi^L_{\eps}$ and $\phi_\eps^U$ such that
 \begin{align*}
- \frac{1}{ 2} \langle \eps, \ex{\tfrac{1}{n}  \bX^\top \bX } \rangle -  \cR( \phi^L_{\eps}) & \le 
\cD_n(0, S) ) -  \hat{\cD}_n(0, S)  \le  \frac{1}{ 2} \langle \eps, \ex{\tfrac{1}{n}  \bX^\top \bX } \rangle + \cR( \phi^U_{\eps}). 
\end{align*}
This double inequality implies that
\begin{align*}
 | \cD_n(0, S) ) -  \hat{\cD}_n(0, S) |  \le  \frac{\|\eps\| }{ 2}  \ex{ \tfrac{1}{n} \|\bX\|_\fro^2}  
 + \max\left\{ \left|\cR( \phi^U_{\eps}) \right|, \left|  \cR( \phi^U_{\eps})\right|  \right\}.
\end{align*} 
 To bound the remainder term, we write
 \begin{align*}
|  \langle S, \ex{ \bO_{t}^{\otimes 2}}   - \ex{ \bO_{t}}^{\otimes 2}  \rangle |  & =  |  \ex{  \left \langle S,  (  \bO_{t}- \ex{ \bO_{t}} )^{\otimes 2} \right \rangle } |  \\
&  \le \|S\|  \ex{  \| (\bO_{t}   - \ex{ \bO_{t}})^{\otimes 2} \|_\ast } \\
&  =  \|S\|\ex{  \| \bO_{t}   - \ex{ \bO_{t}} \|^2_\fro },  \label{eq:Reps_bnd}
\end{align*}
where $\|\cdot\|_*$ denotes the nuclear norm.  The second step follows from Jensen's inequality and the matrix H\"{o}lder inequality and the last step holds because $\|M^{\otimes 2}\|_* = \| M \|_{\fro}^2$ for any $M \in \reals^{m \times n}$. Hence, for $\phi \in \{ \phi_\eps^L, \phi^U_\eps\}$, 
\begin{align*}
 \left|  \cR( \phi^U_{\eps})\right|   \le \frac{\|S\|}{2 n^2} \int_0^1 \ex{  \| \bO_{t}   - \ex{ \bO_{t}} \|^2_\fro }\, \dd t. 
\end{align*}
 This completes the proof of Proposition~\ref{prop:approx_bound}.

\section{Variance inequality: Proof of Theorem~\ref{thm:var_inq}}\label{proof:thm:var_inq}

There are two steps in this proof. Section~\ref{sec:var_inq} establishes the variance inequality for the case $d=1$, and Section~\ref{sec:basis_decomp} then extends this result to the matrix-valued setting $d > 1 $.

\subsection{Variance inequality from pointwise I-MMSE}\label{sec:var_inq}
This section considers a continuous time representation of the free energy. This approach is inspired by the pointwise I-MMSE relations in \cite{weissman:2010, venkat:2012}.  Let $\bX$ be an $n \times 1$ vector and  define the observation process%Let $(\bX, \bZ)$ be distributed according to a probability measure on  an $\reals^n \times \cZ$ where $\bX$ is the unknown signal and $\bZ$ is the observation. Define the random process $(\bY_t)_{t \ge 0}$ according to 
\begin{align}
\bY_t = t \bX + \bW_t,
\end{align}
where $(\bW_t)_{t \ge 0}$ is $\reals^n$-valued standard Brownian motion that is independent of $(\bX, \bZ$). We use $\cG_t $ to denote the filtration generated by the observations $(\bY_s , \bZ )_{ 0 \le s \le t}$.

Using this notation, the conditional expectation of a function of $\bX$ can be expressed as
\begin{align*}
 \ex{ f(\bX) \mid \cG_t }  = \frac{  \int f(\bx)  \exp\left\{  \langle \bx, \bY_t   \rangle  -  \frac{t}{2}  \| \bx\|^2 \right \}  \, \dd P_{\bX \mid \bZ}(\bx \mid \bZ )} { \int  \exp\left\{  \langle \bx, \bY_t \rangle  -  \frac{t}{2}  \| \bx\|^2 \right \}  \, \dd P_{\bX \mid \bZ}(\bx \mid \bZ)}.
\end{align*}
The innovations process $(\bar{\bY}_t)_{t \ge 0}$ is defined by
\begin{align*}
\bar{\bY}_t : = \bY_t  - \int_0^t  \ex{ \bX \mid \cG_s}   \, \dd s.
\end{align*}
An important property of the innovations process  is that it is a standard Brownian motion under $P_{\bY_t \mid \bZ}$.  

We study two different processes that are adapted to $\cG_t$.  The first is the free energy: 
\begin{align}
F_t& : = \log \frac{ \dd P_{\bY_t , \bZ}}{  \dd P_{\bW_t , \bZ}},%\\
%& =  \log \int_{\reals^{n} } \exp\left\{ \langle \bx  , \bY_t  \rangle  - \frac{t}{2}  \|\bx\|_2^2  \right\}    \, \dd P_{\bX \mid \bZ }(\bx \mid \bZ),
\end{align}
and the second is the conditional overlap:
\begin{align}
\hat{\bO}_t &: = \|  \ex{ \bX \mid \cG_t} \|^2.
\end{align}
The means of these processes are related to the mutual information and MMSE according to
\begin{align}
 \ex{ F_t} &= D( P_{\bY_t, \bZ} \, \| \, P_{\bW_t,  \bZ} ), \label{eq:Ft_mean}  \\
 %& = \frac{t}{2} \ex{ \|\bX\|^2} - I(\bX ; \bY_t  \mid \bZ)  \label{eq:Ft_mean} \\
\bEx[ \hat{\bO}_t] &=  \ex{ \|\bX\|^2} - \mmse(\bX \mid \bY_t, \bZ) . \label{eq:Oht_mean}
  %= \frac{t}{2} \ex{ \|\bX\|^2} - I(\bX ; \bY_t)  \\
\end{align} 
%D_{\bX\mid \bZ}(t)$ and  $ \bEx[\hat{\bO}_t ]= \frac{d}{d t} 2 D_{\bX \mid \bZ}(t) $ where $D_{\bX \mid \bZ }(t)$ is the relative entropy function.   
As a consequence of Girsanov's theorem,  the free energy can be expressed as \cite[Equation (116)]{venkat:2012}
\begin{align*}
F_t & =\int_0^t \langle \ex{ \bX \mid \cG_s},  \dd \bY_{s} \rangle    - \frac{1}{2}  \int_0^t \hat{\bO}_s  \, ds .
\end{align*}
In terms of the innovations process $\bar{\bY}_t$, this expression becomes
%Writing this expression in terms of the innovations process yields
\begin{align}
F_t & = \frac{1}{2}  \int_0^t \hat{\bO}_s  \, \dd s + \int_0^t \langle  \ex{ \bX \mid \cG_s},  \dd \bar{ \bY}_{s} \rangle .\label{eq:Ft_decomp}
\end{align}
The second term is a zero-mean martingale and so taking the expectation of both sides, comparing with \eqref{eq:Ft_mean}  and \eqref{eq:Oht_mean}, and then taking the derivative in $t$ recovers the I-MMSE relationship $\frac{ \dd }{ \dd t} I(\bX; \bY_t \mid \bZ) = \frac{1}{2} \mmse( \bX \mid \bY_t, \bZ)$.

%or the relative entropy and the overlap. 

Next, we consider a further decomposition of the conditional overlap.  The following result can be viewed as special case of the Kushner--Stratonovich equation.% \nb{cite}.  

\begin{lemma}\label{lem:rep_pi_f}
For any $f : \reals^{n} \to \reals^m$ with $\ex{ \|f(\bX)\|^2} < \infty$, it holds that
\begin{align*}
\ex{ f(\bX) \mid \cG_t }  & =  \ex{ f(\bX) \mid\bZ}  + \int_0^t \cov( f(\bX), \bX\mid \cG_s)  \, \dd \bar{\bY}_{s}.
\end{align*}
\end{lemma}
\begin{proof} 
For simplicity we derive this result in the setting where there is no side information $\bZ$. The conditional expectation can be expressed as $\ex{ f(\bX) \mid \cG_t }  = g_t(\bY_t)$ where  
\begin{align}
g_t(\by)  =\frac{ \int f(\bx) \exp\left\{ \langle \by, \bx \rangle  - \frac{t}{2} \|\bx\|^2 \right \} \, \dd P_{\bX}(\bx)}{\int \exp\left\{ \langle \bx, \by \rangle  - \frac{t}{2} \|\bx\|^2 \right \} \, \dd P_{\bX}(\bx)}.
\end{align}
By It\^{o}'s formula~\cite[Theorem~8.3]{steele:2001}, 
\begin{align}
g_t(\bY_t) & = g_0(0) + \int_0^t \dot{g}_t(\bY_t)  \, \dd s  + \int_0^t  \sum_{i}  g_s^i( \bY_s) \, \dd Y_{s,i}  + \frac{1}{2} \int_0^t \sum_i g_s^{ii}(\bY_s) \, \dd s,
\end{align}
where $\dot{g}_t(\by)$ denotes the time derivative, and $g^i_t(\by)$ and $g^{ii}_t(\by)$ denote the first and second partial derivatives with respect to the $i$-th entry. Some straightforward but tedious calculations show that 
\begin{align}
\dot{g}_t(\bY_t) & = - \frac{1}{2} \cov(f(\bX) , \|\bX\|^2 \mid \bY_t)\\
g^i_t(\bY_t) %& = \pi_t\left( f(\bX)  X_i  \right) - \pi_t\left( f(\bX) \right) \pi_t\left( X_i  \right) 
& = \cov(f(\bX), X_i \mid \bY_t) \\
g^{ii}_t(\bY_t) 
%&  =  \pi_t\left( f(X)  X^2_i  \right)  -  \pi_t\left( f(X)   \right) \pi_t(X_i^2) + 2 \left\{ \pi_t\left( f(X)  X_i \right) - \pi_t\left( f(X) \pi_t\left(  X_i \right)  \right\} \pi_t\left(  X_i \right) \right) \\
& = \cov(f(\bX), X_i^2 \mid \bY_t)  + 2 \cov( f(\bX), X_i \mid \bY_t) \ex{ X_i \mid \bY_t },
\end{align}
and plugging these expression back into It\^{o}'s formula,  gives
\begin{align}
g_t(\bY_t) & =  \ex{ f(\bX)}  + \int_0^t  \sum_{i}  \cov(f(\bX), X_i \mid \bY_t)  \, \dd Y_{s,i}   -  \int_0^t \sum_i \cov( f(\bX), X_i \mid \bY_t)\ex{ X_i \mid \bY_t }  \, \dd s.
\end{align}
Writing this expressions in terms of the innovations process gives the desired result. 
\end{proof}

\begin{lemma}\label{lem:rep_Oh} The conditional overlap can be expressed as 
\begin{align}
\hat{\bO}_t & =  \| \ex{ \bX\mid \bZ}\|^2 + \int_0^t \| \cov(\bX \mid \cG_s)\|_\fro^2\,  \dd s   + 2 \int_0^t  \langle  \cov( \bX \mid \cG_s) \ex{ \bX \mid \cG_s}  , \dd \bar{\bY}_{s}  \rangle. \label{eq:rep_Oh}
 \end{align}
\end{lemma}
\begin{proof}
 By Lemma~\ref{lem:rep_pi_f}, the conditional mean satisfies  $\ex{\bX \mid \cG_t}   =\ex{ \bX \mid \bZ}  +  \int_0^t  \cov(\bX \mid \cG_s)  \, \dd \bar{\bY}_{s}$.  Applying It\^{o}'s formula with the mapping  $\bx \mapsto \|\bx\|^2$ gives the stated result. 
\end{proof}

Lemma~\ref{lem:rep_Oh} can be viewed as a second-order pointwise I-MMSE relation providing a link between the first and second derivatives of the free energy. Indeed, the second term in \eqref{eq:rep_Oh} is a zero-mean martingale and so taking the expectation of both sides, comparing with \eqref{eq:Oht_mean},  and taking the derivative in $t$ recovers the identity $\frac{\dd}{\dd t} \mmse(\bX \mid  \bY_t ,  \bZ)  =  - \ex{ \| \cov(\bX \mid \bY_t, \bZ)\|_\fro^2}$.

Using these pointwise decompositions of the free energy and conditional overlap we can derive an inequality linking the variance of the conditional overlap to the variance of the free energy. The following result is the special case of Theorem~\ref{thm:var_inq} for the vector-valued  setting ($d=1)$ where the overlap matrix is a scalar. 

\begin{prop}[Variance inequality]\label{prop:var_inq}
For all $\delta > 0$, 
\begin{align*}
\var(\hat{\bO}_t) & \le \frac{4}{\delta^2}   \var\left( F_{t+\delta} - F_t\right)   +2 \bEx[\hat{\bO}_t]   \left( \bEx[\hat{\bO}_{t+\delta}]   -  \bEx[\hat{\bO}_t]   \right).
\end{align*}
\end{prop}
\begin{proof}
We begin with the decomposition 
\begin{align}
\var(\hat{\bO}_t)  %= \cov(\hat{\bO}_t, \hat{\bO}_t )  
= \frac{2}{\delta}  \cov\left( \hat{\bO}_t ,  F_{t+\delta} - F_t \right)  
  - \cov\left( \hat{\bO}_t  , \frac{2}{\delta} \left(  F_{t+\delta} -F_t \right)   - \hat{\bO}_t \right) . \label{eq:var_inq_b}
\end{align}
By the Cauchy-Schwarz inequality, the first covariance term satisfies
\begin{align}
\cov\left( \hat{\bO}_t , F_{t+\delta} -F_t \right) \le \sqrt{  \var( \hat{\bO}_t )  \var\left( F_{t+\delta} - F_t\right) }.  \label{eq:var_inq_c}
\end{align}
To bound the second term, we use \eqref{eq:Ft_decomp} to write
\begin{align*}
\frac{2}{\delta} \left(  F_{t+\delta} -F_t \right)   - \hat{\bO}_t  & = \frac{1}{\delta} \int_{t}^{t+\delta} \left( \hat{\bO}_{s}  - \hat{\bO}_t \right) \, \dd s  + \frac{2}{\delta} \int_{t}^{t+\delta}  \langle \bX_s - \ex{ \bX_{s} \mid  \cG_s } ,  \dd \bar{ \bY}_{s} \rangle.
\end{align*}
%Here, we recall that the innovations term has zero-mean under $\cG_t$, and so 
By the martingale property, the second term on the right-hand side is uncorrelated with $\hat{\bO}_t$ and thus
\begin{align*}
\cov\left( \hat{\bO}_t , \frac{2}{\delta} \left( F_{t+\delta} -F_t \right)   - \hat{\bO}_t \right)  
 & =  \cov\left( \hat{\bO}_t , \frac{1}{\delta} \int_t^{t+\delta} ( \hat{\bO} _{ s} - \hat{\bO}_t)  \, \dd s   \right)\\
 &   =\frac{1}{\delta} \int_t^{t+\delta}   \cov\left( \hat{\bO}_t ,  \hat{\bO}_{s}   - \hat{\bO}_t \right) \, \dd s,
\end{align*}
where the second step is justified by Fubini's theorem.  Next, the decomposition in  Lemma~\ref{lem:rep_Oh} leads to  
\begin{align}
 \hat{\bO}_{s}   - \hat{\bO}_t & = \int_t^{s}  \| \cov(\bX \mid \cG_{u})\|_\fro^2  \, \dd u  + 2 \int_t^{s}     \langle  \cov( \bX \mid \cG_u) \ex{ \bX \mid \cG_u}  , \dd \bar{\bY}_{u}  \rangle .
\end{align} 
Similar to before, the second term is a zero-mean martingale  that is uncorrelated with $\hat{\bO}_t$, and thus
\begin{align*}
 \cov\left( \hat{\bO}_t ,  \hat{\bO}_{s}   - \hat{\bO}_t \right) &  = \cov\left( \hat{\bO}_t ,  \int_t^{s}  \| \cov(\bX \mid \cG_{u})\|_\fro^2  \, \dd u\right) \\
& = \int_s^{t}  \cov\left( \hat{\bO}_t ,   \| \cov(\bX \mid \cG_{u})\|_\fro^2 \right) \, \dd u.
\end{align*}
Now, the crux of the argument is that  because both $\hat{\bO}_t$ and $ \| \cov(\bX \mid \cG_{u})\|_\fro^2 $ are non-negative, the covariance satisfies 
\begin{align*}
 \cov\left( \hat{\bO}_t ,   \| \cov(\bX \mid \cG_{u})\|_\fro^2 \right) \ge   -  \bEx[ \hat{\bO}_t]  \ex{ \| \cov(\bX \mid \cG_{u})\|_\fro^2 } . %= \bEx[ \hat{\bO}_t] \frac{ d}{ d u} \bEx[ \hat{\bO}_{t+u}] 
\end{align*}
Combining the above displays and recognizing that $\ex{ \| \cov(\bX \mid \cG_{u})\|_\fro^2 } = \frac{ \dd}{ \dd u} \bEx[ \hat{\bO}_{u}]$ leads to 
\begin{align}
 \cov\left( \hat{\bO}_t , \frac{2}{\delta} \left( F_{t+\delta} -F_t \right)   - \hat{\bO}_t \right) 
& \ge  - \frac{1}{\delta} \int_t^{t+\delta}  \int_t^{s}   \bEx[\hat{\bO}_t]    \ex{  \| \cov(\bX \mid \cG_{u})\|_\fro^2 }   \, \dd u  \, \dd s \notag \\
 & \ge  -\bEx[\hat{\bO}_t]   \int_t^{t+ \delta}      \ex{  \| \cov(\bX \mid \cG_{u})\|_\fro^2 }   \, \dd u \notag  \\
 & =  - \bEx[\hat{\bO}_t]   \left( \bEx[\hat{\bO}_{t+\delta}]   -  \bEx[\hat{\bO}_t]   \right).   \label{eq:var_inq_g}
 \end{align}
 
In view of \eqref{eq:var_inq_b}, \eqref{eq:var_inq_c}, and \eqref{eq:var_inq_g} we have shown that
\begin{align}
\var(\hat{\bO}_t) & \le \frac{2}{\delta} \sqrt{  \var( \hat{\bO}_t )  \var\left( F_{t+\delta} - F_t\right) }   + \bEx[\hat{\bO}_t]   \left( \bEx[\hat{\bO}_{t+\delta}]   -  \bEx[\hat{\bO}_t]   \right) .
\end{align}
The final expression follows from the fact that for nonnegative numbers $x,a,b$,  the inequality  $x \le \sqrt{ bx} + a$ implies  $x \le b + 2 a$. 
\end{proof}

\subsection{Extension to matrix-valued setting}\label{sec:basis_decomp}

We now prove the variance inequality in the matrix-valued setting ($d > 1$)  where the overlap is  given by a $d\times d$ random matrix. We use the same notation as the statement of Theorem~\ref{thm:var_inq}, where  $\bO$ is the overlap matrix associated with observations $\bZ$ and $\bO_{R}$ is the overlap matrix associated with the observation pair $(\bY_R, \bZ)$. 

 To bound the Frobenius norm of the difference between  the conditional expectation of the overlap and its expectation, we use the following result.

 \begin{lemma}\label{lem:frame}
 The collection of rank-one positive semidefinite matrices  $\{B_1, \dots, B_{d^2}\}$ given in \eqref{eq:def_basis} forms a frame for the space of $d \times d$ symmetric matrices with lower and upper frame bounds $1$ and $d$, respectively. In other words, for each $M \in \sym^{d}$, 
 \begin{align}
 \|M\|_\fro^2 \le \sum_{k = 1}^{d^2}   \langle B_{k},   M  \rangle^2  \le  d \,  \| M \|_\fro^2. \label{eq:frame_bnd}
\end{align}
%Furthermore, $\sum_{k=1}^{d^2}  B_{d} =d  I_d$.% and thus $\|\sum_{ij = 1}^n B_{ij} \|_\infty \le  d$.
\end{lemma} 
 \begin{proof}
For any symmetric matrix $M = (m_{ij})$, we can write
\begin{align*}
\langle B_{\pi(i,j)} , M \rangle &= \begin{dcases}  m_{ii} , & i = j\\
\frac{1}{2} (m_{ii} + m_{jj})  + m_{ij}    & i  <  j \\
\frac{1}{2} (m_{ii} + m_{jj} ) -  m_{ij}    & i  >  j 
  \end{dcases}.
\end{align*}
where  $\pi(i,j) = i + (d-1) j$.   Therefore, %the sum of the squares is given by
\begin{align*}
\sum_k \langle B_{k} , M \rangle^2
&= \sum_{i = j } m^2_{ii} +  \frac{1}{4} \sum_{i \ne j}  (m_{ii} + m_{jj})^2  +   \sum_{i \ne j}  m_{ij}^2  = \|M\|_F^2 +  \frac{1}{4} \sum_{i \ne j}  (m_{ii} + m_{jj})^2.
\end{align*}
By the basic inequality $0 \le (a + b)^2 \le 2a + 2 b$, the second term satisfies $
0  \le \frac{1}{4} \sum_{i \ne j}  (m_{ii} + m_{jj})^2   \le  \sum_{i  \ne j} m_{ii}^2  = (d- 1) \sum_{i} m_{ii}^2 \le (d-1) \|M\|_\fro^2$, and this establishes the desired result. 
%\begin{align*}
%\| M\|_\fro^2 & \le \sum_{k} \langle B_{k} , M \rangle^2 \le   \| M\|_\fro^2 + (d -1) \sum_{i } m_{ii}^2 \le  d \| M\|_\fro^2.
%\end{align*}
%The lower bound is tight when the diagonal entries are zero and the upper bound is tight when the off-diagonal entries are zero. 
\end{proof}

In view of  Lemma~\ref{lem:frame} and the fact that the conditional overlap matrix is symmetric, we can now write
\begin{align}
 \|   \ex{ \bO  \mid \bZ}  - \ex{ \bO}    \|^2_\fro & \le   \sum_{k=1}^{d^2} \langle B_k,     \ex{\bO \mid \bZ}   - \ex{ \bO }     \rangle ^2 . \label{eq:overlap_frame_decomp}
\end{align}
Furthermore, each $B_k$ can be expressed as $B_k = a_k a_k^\top$ where $a_k$ is a $d \times 1$. Introducing the $n \times 1$ vector $\bX_k \coloneqq \bX a_k$, leads to
\begin{align}
\langle B_k,     \ex{\bO \mid \bZ}  -   \ex{\bO }  \rangle & = \hat{\bO}_k - \bEx[ \hat{\bO}_k]
%  \langle a_k a_k^\top , \ex{ \bX \mid \bZ} \ex{ \bX \mid \bZ}^\top \rangle\\
%& =\| \ex{ \bX_k \mid \bZ} \|^2.\\
%& = \hat{\bO}_k
\end{align}
where $\hat{\bO}_k = \| \ex{ \bX_k \mid \bZ} \|^2$ is the conditional expectation of the scalar overlap corresponding to $\bX_k$.  Taking the expectation of both sides of  \eqref{eq:overlap_frame_decomp} yields
\begin{align}
\ex{   \|   \ex{ \bO  \mid \bZ}  - \ex{ \bO}    \|^2_\fro } & \le   \sum_{k=1}^{d^2}\var(\hat{\bO}_k). \label{eq:overlap_frame_decomp_c}
\end{align}

%To conclude the proof of Theorem~\ref{thm:var_inq} we
We can now apply the vector-valued variance inequality in Proposition~\ref{prop:var_inq}, evaluated at $t = 0$, to obtain % This yields
\begin{align}
\var(\hat{\bO}_k)  \le \frac{4}{\delta^2}   \var\left( F_{k,\delta}  \right)  
  +2 \bEx[\hat{\bO}_k ]   \left( \bEx[ \hat{\bO}_{k,\delta}  ]   -  \bEx[\hat{\bO}_k]   \right). \label{eq:overlap_frame_decomp_d}
\end{align}
where $F_{k,\delta}$ and $\hat{\bO}_{k,d}$  are defined as in Section~\ref{sec:var_inq} with respect to the $n \times 1$ vector $\bX_k$. We can reinterpret these  these quantities the context of the matrix model by noting that they are equal in distribution to the random variables $F_{\delta B_k}$ and $\langle B_k , \ex{  \bO_{\delta B_k} \mid \bY_{\delta B_k}, \bZ}  \rangle$ associated with the observation of $\bX$ with matrix $\delta B_k$, i.e., 
\begin{align}
F_{\delta B_k} \coloneqq \log \frac{ \dd P_{\bY_{\delta B_k}, \bZ}}{ P_{\bW, \bZ}}.
\end{align}
In particular, this means that
\begin{align*}
\var(F_{k,\delta}) &= V( P_{\bY_{\delta B_k}, \bZ }\, \| \, P_{\bW, \bZ}) \\
 \bEx[ \hat{\bO}_{k,\delta} ] & = \langle B_k ,  \ex{ \bO_{\delta B_k} }  \rangle
\end{align*}
Combining these expressions with \eqref{eq:overlap_frame_decomp_c} and \eqref{eq:overlap_frame_decomp_d} completes the proof of Theorem~\ref{thm:var_inq}.

\section{Overlap concentration: Proof of Theorem~\ref{thm:overlap_concentration} }\label{proof:thm:overlap_concentration} 

This section gives the proof of the bound on the squared deviation of the overlap given in Theorem~\ref{thm:overlap_concentration}.  Recall that $\bX$ is an $n \times d$ random matrix satisfying $( \ex{ \|\bX\|^4})^{1/4} \le \sqrt{n} \, \rho$ and  $\bO_\eps$ is the $d \times d$ overlap matrix associated with observations $(\bY_\eps, \bZ)$.  The main object of interest is the averaged squared deviation of the overlap given by
\begin{align}
\Delta &\coloneqq  \int \ex{ \| \bO_{\eps + \psi(\eps)}  - \ex{ \bO_{\eps + \psi(\eps)} } \|_\fro^2} \, \dd \mu_{\beta}(\eps) ,%\\
%& =  \int \ex{ \| \bO_{\eps}  - \ex{ \bO_{\eps } } \|_\fro^2} \, \dd \tilde{\mu}_{\beta}(\eps),
\end{align}
where $\psi \colon \psd^d \to \psd^d$ is an arbitrary order-preserving function and  $\mu_\beta$ is the probability measure of the random matrix in \eqref{eq:eps_random}.  Equivalently, we can write
\begin{align}
\Delta & = \int \ex{ \| \bO_{\eps}  - \ex{ \bO_{\eps} } \|_\fro^2} \, \dd \tilde{\mu}_{\beta}(\eps),
\end{align}
where $\tilde{\mu}_{\beta}$ is the pushforward measure under the mapping $\eps \mapsto \eps + \psi(\eps)$. 

\subsection{Orthogonal decomposition}
We begin with the decomposition
 \begin{align}
\ex{\| \bO_{\eps}  - \ex{ \bO_\eps}\|_\fro^2}  &  = \ex{\|\bO_{\eps}  -\ex{ \bO_{\eps} \mid \cG_\eps } \|_\fro^2}  + \ex{ \|\ex{ \bO_{\eps} \mid \cG_\eps } - \ex{ \bO_{\eps}}\|_\fro^2},  \label{eq:Oeps_decomp}
\end{align}
where $\cG_\eps$ denotes the sigma-algebra generated by the observation pair $(\bY_\eps, \bZ)$.  The first term is a measure of the correlation between the rows in the conditional distribution.  In Section~\ref{proof:eq:O_decomp1}, it is shown that 
%, which is measure of the  correlation in the conditional distribution, can be bounded  according to
\begin{align}
\ex{ \|  \bO_\eps  - \ex{ \bO_\eps \mid \cG_\eps} \|_\fro^2} 
%& \le  2 \sqrt{d\,    \ex{ \| \bX^\top \bX\|_\fro^2} \ex{    \left\|  \cov(\gvec(\bX) \mid \cG_{\eps} )\right\|_\fro^2}   } ,
& \le  2 n d  \rho^2 \left(  \ex{    \left\|  \cov(\gvec(\bX) \mid \cG_{\eps} )\right\|_\fro^2} \right)^{1/2}  , \label{eq:O_decomp1}
% & \le 2  \sqrt{ \ex{ \| \bX\|_2^4}}  \sqrt{ \sum_{i,j=1}^n  \ex{  \gtr( \cov(X_i, X_j \mid \cG))^2 } }.
%&\le  2 \sqrt{d}  \sqrt{ \ex{ \| \bX\|_2^4}} \sqrt{ \ex{ \| \cov( \gvec(\bX) \mid \cG)\|_2^2}}.
\end{align}
where $\gvec(\bX)$ denotes the $nd \times 1$ vector obtained by stacking the columns in $\bX$. The second term in \eqref{eq:Oeps_decomp} is a measure of the variability  in conditional expectation with respect to the observations. For $\delta > 0$, applying the variance inequality in Theorem~\ref{thm:var_inq} along with the bound $\ex{ \bO_{\eps} }\preceq \ex{ \bX^\top \bX} \preceq n \rho^2 I_d $ and that $\gtr(B_k) = 1$  yields
\begin{align}
 \ex{ \|\ex{ \bO \mid \cG_\eps} - \ex{ \bO}\|_\fro^2} 
%& \le \frac{4}{ \delta^2} \sum_{k=1}^{d^2} V\left( P_{\bY_{\eps + \delta B_k}, \bZ}  \,  \| \,  P_{\bY_\eps, \bZ}  \right)  \notag\\
& \le \frac{4}{ \delta^2} \sum_{k=1}^{d^2} V\left( P_{\tilde{\bY}_{\delta B_k} , \bY_{\eps}, \bZ}  \,  \| \,  P_{\tilde{\bW}, \bY_\eps, \bZ}  \right)   +   2n \rho^2  \sum_{k=1}^{d^2}  \langle B_k, \bEx[ \bO_{\eps + \delta B_k} ] -  \bEx[ \bO_\eps  ]   \rangle. \label{eq:O_decomp2}
\end{align} 
Combining \eqref{eq:Oeps_decomp}, \eqref{eq:O_decomp1}, and \eqref{eq:O_decomp2} and then integrating with respect to $\tilde{\mu}$ yields
\begin{align}
\Delta & \le  \Delta_1 + \Delta_2 + \Delta_3,   \label{eq:Delta_bnd}
\end{align}
where
\begin{align}
\Delta_1 & = 2  d n \rho^2  \left(  \int \ex{   \left\|  \cov(\gvec(\bX) \mid \cG_\eps)\right\|_\fro^2}   \, \dd \tilde{\mu}_\beta(\eps)  \right)^{1/2} \notag  \\
\Delta_2 & =\frac{4}{ \delta^2}   \sum_{k=1}^{d^2}  \int V\left( P_{\tilde{\bY}_{\delta B_k} , \bY_{\eps}, \bZ}  \,  \| \,  P_{\tilde{\bW}, \bY_\eps, \bZ}  \right)  \, \dd \tilde{\mu}_\beta(\eps) \notag  \\
\Delta_3 & =    2 \rho^4   \sum_{k=1}^{d^2}   \int \langle B_k, \bEx[ \bO_{\eps + \delta B_k} ] -  \bEx[ \bO_\eps  ]   \rangle \, \dd \tilde{\mu}_\beta(\eps) . \label{eq:Delta_3}
\end{align}
%These terms are bounded separately the he following sections. 

%\subsection{Bound on first term in \eqref{eq:Delta_bnd}} 

\subsection{Average over perturbation}
Let $\cU = [0,1]^{d^2 + 1}$ and for $k = 0, 1\dots, d^{2}$ define  $Q_k , Q_k' \colon \cU \times \cU \to [0, \infty)$, according to 
%Fix any points $u_1, \dots, u_{d^2} \in [0,1]$ and define the functions 
\begin{align*}
Q_k(u, \tilde{u}) & \coloneqq  \langle B_k,  \ex{ \bO_{\eps(u) + \psi(\eps(\tilde{u}))}} \rangle\\
Q'_k(u, \tilde{u}) & \coloneqq \frac{\partial  }{\partial u_k} \langle B_k,  \ex{ \bO_{\eps(u) + \psi(\eps(\tilde{u}))}} \rangle
\end{align*}
where  $B_0 \coloneqq  I_d$ and %$\eps(\cdot)$ is defined in \eqref{eq:eps_u}. 
\begin{align}
\eps(u) \coloneqq  \frac{\beta u_0}{2}  B_0 + \frac{\beta}{2d} \sum_{k=1}^{d^2} u_k B_k. \label{eq:eps_u}
\end{align}
%where $B_1, \dots, B_{d^2}$  is the collection of $d \times d$ positive semidefinite matrices  defined as in \eqref{eq:def_basis} and the entries of $u = (u_0, \dots, u_{d^2})$ are independent random variables distributed uniformly on $[0,1]$.  
Here, the existence of the partial derivative is justified by the linearity of $u \mapsto \eps(u)$ and the differentiability of the MMSE with the matrix in the linear model channel~\cite{payaro:2009}. Furthermore, the following properties hold:
%Then, we have the following properties
\begin{enumerate}[(i)]
\item Monotonicity:  $Q_k(u,\tilde{u})$ is order-preserving,  and thus non-decreasing in each of its arguments, because it is the composition of order-preserving mappings. 
\item Boundedness: $0 \le  Q_k(u,\tilde{u}) \le n \rho^2 $ because $0 \preceq \ex{\bO_{\eps} } \preceq  \ex{ \bX^\top \bX}\preceq n \rho^2 I_d$ and $\gtr(B_k) = 1$. 
\end{enumerate}

With these definitions in hand, we are now ready to bound the integral appearing in $\Delta_1$. 
%
% For the expression for the Hessian of the mutual information the vector valued linear channel~\cite{lamarca:2009}, it follows that
%\begin{align}
%\nabla_{R}  \MMSE(\bX \mid \bY_R  \mid \bZ) & = -   \sum_{i,j=1}^n \ex{ \cov(X_i, X_j \mid \bY, \bZ)^{\otimes 2} }.
%\end{align}
By the second order I-MMSE relation for the vector-valued linear Gaussian channel with scalar signal-to-noise ratio \cite{payaro:2009} we can write
\begin{align*}
\frac{ \dd }{ \dd t} \mmse(\bX \mid \bY_{R + t I_d}  , \bZ)  = - \ex{   \left\|  \cov(\gvec(\bX) \mid  \bY_{R + t I_d},   \bZ)\right\|_\fro^2}.
\end{align*}
In conjunction with the relationship between the MMSE and the overlap in \eqref{eq:OtoMMSE} and the chain rule for differentiation, this implies that the partial derivative with respect to $u_0$ can be expressed as
\begin{align}
Q_0'(u, \tilde{u})%& = \left \langle I_{d}^{\otimes 2}   , \sum_{i,j=1}^{n}  \cov(X_i, X_j  \mid \cG_{\eps(u) + \psi(\eps(\tilde{u}))})^{\otimes 2}  \right \rangle \\
& = \frac{\beta}{2} \ex{   \left\|  \cov(\gvec(\bX) \mid \cG_{\eps(u) + \psi(\eps(\tilde{u}))})\right\|_\fro^2}. 
\end{align}
From the definition of  $\tilde{\mu}_\beta$, it then follows that
\begin{align}
 \int \ex{   \left\|  \cov(\gvec(\bX) \mid  \cG_{\eps } \right\|_\fro^2}  \, \dd \tilde{\mu}_{\beta}(\eps) & = \frac{2}{\beta}  \int_{\cU}  Q_0'(u, u) \, \dd u. \label{eq:Delta_bnd_T1a}
\end{align}
%
%\frac{1}{\alpha } \int_0^\alpha \ex{    \left\|  \cov(\gvec(\bX) \mid \cG_{\eps(u) + \psi(\eps(u))})\right\|_\fro^2} \, \dd  u_0  & = \frac{1}{\alpha } \int_0^\alpha Q_0'(u, u) \, \dd  u_0.
%\end{align} 
Now, the crux of the argument is that because $Q_0(u,\tilde{u})$ is non-decreasing in both $u_0$ and $\tilde{u}_0$, the integral of $Q'_0(u,u)$ with respect to $u_0$ can be bounded from above as follows: 
\begin{align*}
\int_0^1   Q'_0(u, u)  \, \dd u_0 & \le \int_0^1  \frac{\partial}{\partial u_0}   Q_0(u, u)  \, \dd u_0\\
&  = Q_0(u ,u) \Big \vert_{u_0  = 1}  - Q_0(u ,u) \Big \vert_{u_0  = 0}\\
& \le nd \rho^2,
\end{align*}
where the last step follows from the bounds on $Q_0(u,\tilde{u})$. Combining this inequality with \eqref{eq:Delta_bnd_T1a}  leads to 
\begin{align*}
% \int \ex{   \left\|  \cov(\gvec(\bX) \mid  \cG_{\eps + \psi(\eps)} \right\|_\fro^2}  \, \dd \mu_{ \beta}(\eps) 
  \int \ex{   \left\|  \cov(\gvec(\bX) \mid  \cG_{\eps } \right\|_\fro^2}  \, \dd \tilde{\mu}_{\beta}(\eps)
 & \le \frac{ 2 nd \rho^2, } {\beta},
\end{align*}
and thus, by Jensen's inequality, $\Delta_1 \le (2 nd \rho^2)^{3/2} / \sqrt{\beta}$.
%\begin{align}
%\Delta_1 \le 2 n d \rho^2 \sqrt{ \frac{ 2 nd \rho^2, } {\beta}}.
%\end{align}

For the term $\Delta_2$, noting that  $\eps \sim \mu_\beta$ satisfies $\|\eps\| \le \beta$ almost surely and  $\| \delta B_k\| = \delta$ leads to 
\begin{align*}
%V\left( P_{\bY_{\eps + \psi(\eps)  + \delta B_k}, \bZ}  \,  \| \,  P_{\bY_\eps, \bZ}  \right) \le \overline{V},
V\left( P_{\tilde{\bY}_{\delta B_k} , \bY_{\eps + \psi(\eps) }, \bZ}  \,  \| \,  P_{\tilde{\bW}, \bY_{\eps + \phi(\eps)}, \bZ}  \right)  \le \overline{V},
 %\coloneqq\max_{ \eps, \tilde{\eps} \in \psd^d \, : \, \|\eps\| \le \beta, \|\tilde{\eps}\| \le \delta} V\left( P_{\bY_{\eps + \psi(\eps) + \tilde{\eps} }, \bZ}  \,  \| \,  P_{\bY_\eps + \psi(\eps) , \bZ}  \right).
%\Delta_2 & \le \frac{4d^2}{\delta^2} \max_{ \eps, \tilde{\eps} \in \psd^d \, : \, \|\eps\| \le \beta, \|\tilde{\eps}\| \le \delta} V\left( P_{\bY_{\eps + \psi(\eps) + \tilde{\eps} }, \bZ}  \,  \| \,  P_{\bY_\eps + \psi(\eps) , \bZ}  \right).
\end{align*}
for all $\eps$ in the support of  $\tilde{\mu}_\beta$. Hence, $\Delta_2 \le  \frac{4d^2}{\delta^2}  \overline{V}$. 

Finally, we consider the term $\Delta_3$. From \eqref{eq:eps_u} we can write $\eps(u)  + \delta_k B_k =  \eps(  u +  \alpha e_k)$  
%\begin{align}
%\eps(u)  + \delta_k B_k =  \eps(  u +  \alpha e_k),
%\end{align}
where $e_k$ denotes a standard basis vector in  $(1 + d^2)$ dimensions (with the indexing starting at zero) and $\alpha  =  2 d \delta / \beta$. Consequently, $ \langle B_k,  \ex{ \bO_{\eps(u) + \psi(\eps(\tilde{u})) + \delta B_k }} \rangle  = Q_{k}( u +  \alpha e_k, \tilde{u} )$, 
%\begin{align}
% \langle B_k,  \ex{ \bO_{\eps(u) + \psi(\eps(\tilde{u})) + \delta B_k }} \rangle & = Q_{k}( u +  \alpha e_k, \tilde{u} ).
%\end{align}
and so the integral in $\Delta_3$ can be expressed as
\begin{align}
\int \langle B_k, \bEx[ \bO_{\eps + \delta B_k} ] -  \bEx[ \bO_\eps  ]   \rangle \, \dd \tilde{\mu}_{\beta}(\eps) 
 & =  \int_{\cU} \left\{   Q_k(u +\alpha e_k , u) -  Q_k(u, u) \right\}  \, \dd u. \label{eq:Delta_bnd_T3a}
\end{align}
Since $v \mapsto Q_k(u + v  e_k , u)$ is differentiable, the integrand can be expressed as %we can write the difference as an integral and interchanging the order of integration as follows: 
\begin{align} 
 \int_0^1  \left\{ Q_k(u + 2\delta e_k , u) -  Q_k(u, u) \right\} \, \dd u_k  & =  \int_0^1  \int_0^{\alpha}  Q'_k(u +v , u) \, \dd v  \, \dd u_k \notag \\
& =  \int_0^{\alpha}   \int_0^1    Q'_k(u +v , u)   \, \dd u_k   \, \dd v. \label{eq:Delta_bnd_T3b}
\end{align} 
Similar to before, we use the fact that $Q_k(u +v, \tilde{u} )$ is non-decreasing in both $u_0$ and $\tilde{u}_0$ to bound the term on the inside
\begin{align*}
 \int_0^1    Q'_k(u +v , u)   \, \dd u_k & \le  \int_0^1   \frac{\partial}{\partial u_k}   Q_k(u +v , u)   \, \dd u_k \le 
 n \rho^2, % \langle B_k, \ex{ \bX^\top \bX} \rangle
\end{align*}
where the last step follows from the bounds on $Q_k(u,\tilde{u})$. Plugging this inequality back into  \eqref{eq:Delta_bnd_T3b} then yields
%Combining this inequality with  \eqref{eq:Delta_bnd_T3a}  then yields
\begin{align}
%\int \langle B_k, \bEx[ \bO_{\eps + \delta B_k} ] -  \bEx[ \bO_\eps  ]   \rangle \, \dd \tilde{\mu}_{\beta}(\eps)
\int_0^1  \left\{ Q_k(u + 2\delta e_k , u) -  Q_k(u, u) \right\} \, \dd u_k 
  \le \alpha n\rho^2. % \frac{\langle B_k, \ex{ \bX^\top \bX} \rangle}{\beta},
\end{align}
In view of \eqref{eq:Delta_3} and  \eqref{eq:Delta_bnd_T3a}, we conclude that  $\Delta_3 \le  4 n^2 d \delta \rho^4/ \beta$.
%\begin{align}
% \Delta_3 %& =   \sum_{k=1}^{d^2} \langle B_k,\ex{ \bX^\top \bX}  \rangle  \int \langle B_k, \bEx[ \bO_{\eps + \delta B_k} ] -  \bEx[ \bO_\eps  ]   \rangle \, \dd \mu_{\alpha, \beta} 
% \le    \sum_{k=1}^{d^2} \langle B_k,\ex{ \bX^\top \bX}  \rangle^2 \le d \| \ex{ \bX^\top \bX}\|_\fro^2. 
%\end{align}
%where the last step follows from the upper frame bound in Lemma~\ref{lem:frame}. 

Combining the upper bounds on $\Delta_k$, $k =1 ,2,3$ with \eqref{eq:Delta_bnd} completes the proof of Theorem~\ref{thm:overlap_concentration}.

\subsection{Proof of Inequality~\eqref{eq:O_decomp1}}\label{proof:eq:O_decomp1} 
To simplify the notation we make the dependence on $\eps$ implicit and write $\bO$ and $\cG$. Recall that the overlap matrix is given by $\bO = (\bX')^\top \bX''$ where $\bX', \bX''$ are conditionally independent draws from the condition distribution of $\bX$ given $\cG$. Introducing the notation $\hat{\bX} \coloneqq \ex{ \bX \mid \cG}$ and noting that $\ex{ \bO \mid \cG} = \hat{\bX}^\top \hat{\bX}$ leads to the  orthogonal decomposition 
\begin{align}
\bO - \ex{ \bO \mid \cG} %=  (\bX')^\top \bX'' - \hat{\bX}^\top \hat{\bX} 
=  ( \bX' - \hat{\bX})^\top \bX''  +  \hat{\bX}^\top( \bX''  - \hat{\bX}),
\end{align}
where the expected inner product of the terms on the right-hand side is equal to zero.  Using  the identities $\|A^\top B\|_F^2 = \langle A^\top B  , A^\top  B\rangle= \langle  A A^\top , B B^\top \rangle$  and recalling that $(\bX, \bX', \bX'')$ are i.i.d.\ conditional on $\cG$ leads to
 %$\langle A^\top B, C^\top D \rangle = \langle CA^\top, D B^\top \rangle$ 
% and recalling $(\bX, \bX', \bX'')$ is exchaingable under $\cG$
  %that $(\bX_k ,  \hat{\bX})$  is equal in distribution to $(\bX , \hat{\bX})$ for $k =1, 2$ leads to 
\begin{align*}
 \ex{ \|  \bO  - \ex{ \bO \mid \cG}\|_\fro^2  } 
& =  \ex{ \|  ( \bX' - \hat{\bX})^\top \bX'' \|_\fro^2 }   + \ex{ \|  \hat{\bX}^\top  ( \bX'' - \hat{\bX}) \|_\fro^2 }  \\
%&  =   \ex{ \langle   ( \bX_1 - \hat{\bX})^\top \bX_2,  ( \bX_1 - \hat{\bX})^\top \bX_2   \rangle}  +  \ex{ \langle   ( \bX_2 - \hat{\bX})^\top \hat{\bX},  ( \bX_2 - \hat{\bX})^\top  \hat{\bX}   \rangle}  \\
&  =   \ex{ \langle   ( \bX' - \hat{\bX}) ( \bX' - \hat{\bX})^\top , \bX''(\bX'')^\top   \rangle}  \notag \\
& \quad  +  \ex{ \langle  ( \bX'' - \hat{\bX})^\top   ( \bX'' - \hat{\bX})^\top ,  \hat{\bX}  \hat{\bX}^\top   \rangle}  \\
&  =   \ex{ \langle   \bE,  \bX \bX^\top   \rangle}  +  \ex{ \langle  \bE ,  \bX' (\bX'')^\top    \rangle} ,
\end{align*}
where $\bE = \ex{ (\bX- \hat{\bX}) (\bX - \hat{\bX})^\top \mid \cG}$. From here, two applications of the Cauchy-Schwarz inequality  yields
\begin{align}
%\ex{ \langle   \bE,  \bX \bX^\top   \rangle} 
 \ex{ \|  \bO  - \ex{ \bO \mid \cG}\|_\fro^2  } 
&  \le  \sqrt{ \ex{\| \bE\|_\fro^2}}   \left(\sqrt{ \ex{\| \bX \bX^\top \|_\fro^2}}  + \sqrt{ \ex{\| \bX' (\bX'')^\top  \|_\fro^2}}    \right)  \notag \\
&  \le 2   \sqrt{ \ex{\| \bE\|_\fro^2 }  \ex{\| \bX^\top \bX \|_\fro^2}} \notag\\
&  \le 2  n \rho^2 \sqrt{ d\,   \ex{\| \bE \|_\fro^2}} , \label{eq:O_to_Oh_bnd_d}
\end{align}
where the last line follows from  $( \ex{ \|\bX\|^4})^{1/4} \le \sqrt{n} \rho$, which implies that $\ex{ (\bX^\top \bX)^2 } \preceq n^2 \rho^4 I_d$. 
%where the second step follows from $ \ex{\| \bX^\top \bX \|_\fro^2} \le d  n^2 \rho^2$ and
%% the submultiplicativity of the Frobenius, which gives $\ex{\| \bX' (\bX'')^\top  \|_\fro^2} \le \ex{\| \bX'\|_\fro^2} \ex{ \|\bX''  \|_\fro^2}$
%\begin{align}
% \ex{\| \bX' (\bX'')^\top \|_\fro^2 } %& =  \ex{ \langle \bX_1 \bX^\top_2,  \bX_1 \bX^\top_2 \rangle} \\
%%& =  \ex{ \langle( \bX')^\top \bX',  (\bX'')^\top \bX'' \rangle} \\
%& \le  \sqrt{  \ex{ \|  (\bX')^\top \bX'\|^2_\fro}   \ex{ (\bX'')^\top \bX'' \|_\fro^2} }  \\
%&  =  \ex{\| \bX^\top  \bX \|_\fro^2 } .
%\end{align}

Finally, letting $\bX = (X_1, \dots, X_n)^\top $  and $\hat{\bX} = (\hat{X}_1, \dots, \hat{X}_n)^\top$, allows us to write
\begin{align}
 \| \bE  \|^2_\fro & =  \sum_{i,j=1}^n  \ex{  (X_i - \hat{X}_i)^\top ( X_j - \hat{X}_j) \mid \cG }^2 \notag \\
 & =  \sum_{i,j=1}^n   \left \langle I_d,  \cov(X_i, X_j \mid \cG)  \right \rangle^2 \notag \\
& \le  \sum_{i,j=1}^n d  \left\|  \cov(X_i, X_j \mid \cG)\right\|_\fro^2 , \label{eq:O_to_Oh_bnd_g}%\
%& = d \lim_{\eps \downarrow 0}    \frac{ \dd}{ \dd \gamma }    \langle I_d , \MMO( \bX \mid  \cG,  \sqrt{ \gamma} \bX + \bW) \rangle \Big \vert_{\gamma = \eps}. \label{eq:O_to_Oh_bnd_g}
\end{align}
where the third step is the Cauchy-Schwarz inequality.  Plugging \eqref{eq:O_to_Oh_bnd_g}
back into \eqref{eq:O_to_Oh_bnd_d} completes the proof of \eqref{eq:O_decomp1}.

\appendix

\section{Comparison lemma for differential equations}

There exists a large literature on differential inequalities; see e.g.,~\cite[Chapter~III]{hartman:2002}. The next result is related to a result by  Lasota et al.\ \cite[Lemma~4]{lasota:1970}. A self-contained proof is provided for completeness.  %\nb{GR: See EG Theorem 2.2.5 Pachpatte attributed to  Lasota et al. (1971)}

\begin{lemma}[Comparison Lemma]\label{lem:ODE_comp}
Let $\phi(t)$ and $\psi(t)$ be $\sym^d$-valued functions that are differentiable on $(0,1)$ with %defined on $[0,1]$ sat : [0,1] \to \sym^d$ satisfying
\begin{alignat*}{3}
\phi'(t) &\preceq F(t, \phi(t)) , &\quad &\phi(0)  \preceq u\\
\psi'(t) &\succeq G(t, \psi(t)), &&  \psi(0)  \succeq v,
\end{alignat*}
where $u,v \in \sym^d$ satisfy $u \preceq v$ and $F, G \colon [0,1] \times \sym^d \to \sym^d$ satisfy  $F(t, \phi) \preceq G(t, \psi) $  whenever $\phi \preceq \psi$.  Furthermore, suppose that  $ |G(t, \phi) - G(t,\psi) |  \le L(t) \| \phi - \psi\|$ where $L(t)$ is integrable on $[0,1]$. % and also $F(t, \phi) \preceq G(t, \psi) $  whenever $\phi \preceq \psi$. 
 Then, 
\begin{align}
\phi(t)  - u \preceq  \psi(t)   -  v   \quad \text{for all $t \in [0,1]$.}
\end{align}
\end{lemma}
\begin{proof}
The desired result can be stated equivalently as $h(t) \le 0$ for all $t\in [0,1]$ where $h(t) = \lambda_\mathrm{max}( \phi(t)  - u  - \psi(t)  + v  )$. Using the variational representation of the maximum eigenvalue $\lambda_\mathrm{max}(M) = \sup_{ \| u \| \le 1} u^\top  M u$
 %$h(t) = \sup_{ \| u \| \le 1} u^\top (\phi(t) - a  - \psi(t)  + b) u$
   and then applying the envelope theorem \cite[Theorem~2]{milgrom:2002}, it follows that $h(t)$ is absolutely continuous with
$ h(b)   =  h(a) +  \int_a^b  \lambda_\mathrm{max}( \phi'(t) - \psi'(t) ) \, \dd t$. The integrand can be bounded from above using
\begin{align*}
\lambda_\mathrm{max}( \phi'(t) - \psi'(t))  & \le  \lambda_\mathrm{max}( F(t, \phi(t))- G(t, \psi(t))) \\
&   \le \lambda_\mathrm{max}\left( G(s,   \psi(t) + h(t) I_d  ) -G(t, \psi(t)) \right) \\
&\le  L(t)   | h(t) |, 
\end{align*}
where the first inequality follows from the assumptions on $\phi(t), \psi(t)$ and the remaining steps follow from the assumptions on $F,G$ and the basic inequality $\phi(t)  \preceq \psi(t) + h(t) I_d$. Thus, we have shown that 
\begin{align*}
 h(b)   \le  h(a) +  \int_a^b  L(t) | h(t)| \,\dd t.
 \end{align*}
Now, suppose for the sake of contradiction that there exists $b \in ( 0, 1]$ such that $h(b) > 0$.  By the continuity of $h$ and the fact that $h(0) \le 0$,  there exists $a \in [0,b)$ such that $h(a) = 0$ and  $h(t) \ge 0$ for all $t \in [a,b]$.  Applied to this interval, the upper bound described above implies that
$h(b)  \le  \int_a^b  L(t)  h(t) \, \dd t$.  By Gr\"onwall's inequality \cite[Chapter~III, Theorem~1.1]{hartman:2002} it follows that $h(b) \le 0$, which is a contradiction. %Therefore, we conclude that $h(t)$ is nonpositive on $[0,1]$ and the proof is complete. 
\end{proof}

\section{Gradients of convex functions}

%We begin with an elementary result about the limits of the derivative of convex function. 
%
%\begin{lemma}[Griffiths lemma{\cite[Lemma~1]{hepp:1973}}]  Let   $f_n$ be a sequence of proper convex functions on $\reals$ that converges pointwise to $f$ on $I \coloneqq (a,b)$. Then, the limit $f$ is convex and  for each $x \in I$, 
%\begin{align}
%f^-(x) \le \liminf_{n \to \infty} f_n^- \le \limsup_{n \to \infty} f_n^+ \le f^+(x),
%\end{align}
%where the superscripts $-$ and $+$ denote the left- and right-derivatives, respectively. 
%\end{lemma} 
%%\begin{proof}
%%Fixe $x \in I$. For each $y > 0$  such that $x \pm y \in I$, pointwise convergence implies that
%%\begin{align}
%% \limsup_{n \to \infty} f_n^+ \le f^+(x+ y) 
%%\end{align}
%%\end{proof}

\begin{lemma}\label{lem:convex_gradients}
Let $f_n$ be a sequence of  convex functions defined on $\sym^d_+$ that converges pointwise to a limit $f$. Then, $f$ is convex and for all $x \in \psd^d$ and $y \in \sym^d$ such that $x + t y \in \psd^d$ for small enough $t  > 0$, it follows that
\begin{align}
\limsup_{n \to \infty} \max_{u \in \partial f_n(x)} \langle u, y \rangle %& =  \limsup_{n \to \infty}  \inf_{t >0} \frac{ f_n(x + t y)  - f_n(x) }{ t}\\
%& \le \inf_{t >0}  \limsup_{n \to \infty}   \frac{ f_n(x + t y)  - f_n(x) }{ t}\\
%& =  \inf_{t >0}     \frac{ f(x + t y)  - f(x) }{ t}\\
& \le   \max_{u \in \partial f(x)} \langle u, y \rangle. 
\end{align}
\end{lemma} 
\begin{proof} The directional derivative of a convex function  satisfies $f'(x;y)  = \inf_{t >0} t^{-1}( f(x + t y)  - f(x))$  \cite[Proposition~17.2]{bauschke:2017}. 
%\begin{align}
%f'(x;y)  = \inf_{t >0} \frac{ f(x + t y)  - f(x) }{ t}. 
%\end{align}
For each $t > 0$, pointwise convergence implies that $\frac{1}{t} ( f_n(x + ty) - f_n(x)) \to \frac{1}{t} ( f(x + ty) - f(x)) $ and taking the infimum over both sides yields $ \limsup_{n \to \infty}  f'_n(x;y)  \le   f'(x;y)$. To obtain the desired expression we apply the max formula \cite[Theorem~17.18]{bauschke:2017}, which states that $ f'(x ; y) =  \max_{u \in \partial f(x)} \langle u, y \rangle  $.
\end{proof}

\section{Concentration of free energy}

\begin{prop}\label{prop:relative_entropy_variance}
Consider the assumptions of Theorem~\ref{thm:MI_bound}. For all  $(R, S) \in  \psd^d \times \psd^{d^2}$, %the relative entropy variance satisfies
\begin{align}
 V( P_{ \bY_{R , S}  , \bZ } \, \| \, P_{\bY_{0,0}, \bZ}  ) \le  n d \rho^2 \left( \|R\| + d \rho^2 \|S\| \right)   +   3 n d^2 \rho^4 \left( \|R \|  + 2 d \rho^2  \|S\|   \right)^2.
\end{align}
\end{prop}
\begin{proof}
To simplify notation we define $\overline{\bW} = (\bW , \bW')$. The free energy is given by
\begin{align}
F &\coloneqq \log \frac{ \dd P_{\bY_{R,S} , \bZ}}{ \dd (P_{\bY_{0,0}, \bZ})} (\bY_{R,S} ,  \bZ) . 
\end{align}
Expressing $\bY$ as a function of the tuple $(\overline{\bW} , \bX)$ gives an alternative representation
\begin{align*}
F&  =   \log \int \exp\left\{  - \cH_{\bu}(\overline{\bW}, \bX) \right\}  \, \dd P_{\bX \mid \bZ}(\bu \mid \bZ), 
\end{align*} 
where the ``Hamiltonian'' is given by 
\begin{align*} 
\cH_{\bu}(\overline{\bW}, \bX) &\coloneqq  \frac{1}{2} \|\bu R^{1/2} \|_\fro^2  + \frac{1}{2} \|\bu^{\otimes 2}  ( \tfrac{1}{n} S)^{1/2}  \|_\fro^2   -  \langle \bW ,  \bu R^{1/2}  \rangle - \langle \bW'  , \bu^{\otimes 2} (\tfrac{1}{n}  S)^{1/2} \rangle  \notag \\ 
& \quad  -  \langle \bX,  \bu R  \rangle  -  \langle \bX^{\otimes 2} ,  \bu^{\otimes 2}  \tfrac{1}{n} S \rangle .
\end{align*} 
Note that the gradient of $F$ with respect to  $( \overline{\bW}, \bX)$ satisfies
\begin{align}
\nabla F  & =  -   \int \nabla  \cH_{\bu}(\overline{\bW}, \bX) \, \dd \mu(\bu)   \label{eq:Fgradient}
\end{align} 
where $\mu$ is the probability measure defined by 
\begin{align*}
\dd \mu(\bu) \propto \exp\{ -\cH_{\bu}( \overline{\bW}, \bX)\} \, \dd P_{\bX \mid \bZ}(\bu \mid \bZ).
\end{align*}

We use the following decomposition of the variance
\begin{align}
\var(F) & = \ex{ \left( F - \ex{ F \mid \bX, \bZ}\right)^2}  + \ex{ \left(  \ex{ F \mid \bX, \bZ} - \ex{ F \mid \bZ} \right)^2}  + \ex{ \left(  \ex{ F \mid \bZ} - \ex{ F} \right)^2}. \label{eq:VarF_decomp}
\end{align} 
The terms on the right-hand side correspond to the variation with respect to the noise $\overline{\bW}$, the signal $\bX$, and the side information $\bZ$, respectively.

To bound the first term in \eqref{eq:VarF_decomp} we use the Gaussian Poincar\'{e} inequality \cite[Theorem~3.20]{boucheron:2013}, which yields
\begin{align}
\ex{ \left(F - \ex{ F \mid \bX, \bZ }\right)^2} &  \le \ex{  \|\nabla_{\bW} F \|_\fro^2} +  \ex{  \|\nabla_{\bW'}   F    \|_\fro^2} \notag \\
& = \ex{  \left\| \int \bu  R^{1/2} \, \dd \mu(\tilde{\bx})  \right\|_\fro^2} + \frac{1}{n}   \ex{ \left \| \int \bu^{\otimes 2}  S^{1/2}  \, \dd \mu(\tilde{\bx})   \right  \|_\fro^2} \notag \\
%& \le \ex{  \left\| \int \bX  R^{1/2} \, \dd \mu(\tilde{\bx})  \right\|_\fro^2} + \frac{1}{n}   \ex{ \left \| \int \bu^{\otimes 2}  S^{1/2}  \, \dd \mu(\tilde{\bx})   \right  \|_\fro^2} \notag \\
%& =  \langle R , \ex{ \bX^\top \bX} \rangle  + \frac{1}{n} \langle S , \ex{ (\bX^\top \bX)^{\otimes 2} }\rangle  \notag \\
& \le  \ex{ \int  \left\|  \bu  \right\|_\fro^2   \, \dd \mu(\bu) } \|R\| +  \frac{1}{n} \ex{ \int \left \|  \bu \right  \|_\fro^4   \, \dd \mu(\bu)  }  \|S\|  \notag \\
%& \le  \ex{ \int  \left\|  \tilde{\bx}  \right\|_\fro^2   \, d \mu(\tilde{\bx}) } \|R\| +  \frac{1}{n} \ex{ \int \left \|  \tilde{\bx}^{\otimes 2}  \right  \|_\fro^2   \, d \mu(\tilde{\bx})  }  \|S\|\\
%& =  \ex{ \|\bX\|_F^2} \|R\|   + \frac{1}{n} \ex{ \|\bX \|_F^4} \|S\|  \notag \\
%& \le \langle \ex{ \bX^\top \bX} , R \rangle  + \frac{1}{n}  \langle \ex{ (\bX^\top \bX)^{\otimes 2}} , S\rangle\\
& \le  n d \rho^2 \|R\|  + n d^2 \rho^4 \|S\|. \label{eq:F_W_f}
\end{align}
The last step  follows from the assumption $\|X_i\| \le \sqrt{d}\, \rho$ almost surely, and thus every $\bu$ in the support of $\mu$ satisfies $\|\bu\|_F \le \sqrt{ n d} \rho$. 

For the second term in \eqref{eq:VarF_decomp} we use the Efron-Stein inequality~\cite[Theorem~3.1]{boucheron:2013}, which yields
\begin{align}
\ex{ \left(  \ex{ F \mid \bX, \bZ} - \ex{ F \mid \bZ} \right)^2} 
  & \le  \frac{1}{2} \sum_{i=1}^n \ex{ \left(  \ex{ F \mid \bX, \bZ} - \ex{ F^{(i)}  \mid \bX^{(i)}, \bZ } \right)^2}  \notag \\
& \le \frac{1}{2} \sum_{i=1}^n \ex{(  F- F^{(i)} )^2},  \label{eq:Fstein1}
\end{align}
where  $F^{(i)}$ and $\bX^{(i)}$ are obtained by replacing the $i$-th row of $X_i$ with an independent copy $X_i'$. For $t \in [0,1]$, let $\bX^{(i)}(t) = (1-t) \bX  + t \bX^{(i)}$ and define the interpolating free energy
\begin{align*}
F_t^{(i)}&  =   \log \int \exp\{ -\cH_{\tilde{\bx}}( \overline{\bW}, \bX^{(i)}(t)  ) \, \dd P_{\bX \mid \bZ}(\tilde{\bx} \mid \bZ),
\end{align*} 
such that $ F^{(i)}_1  = F$ and $F^{(i)}_0 = F^{(i)}$. Letting $\cT$ be the linear operator defined in Lemma~\ref{lem:lin_op}, and using the fact that $\|\cT(A^\top B)\| \le  \| A^\top B\|_* \|S\| \le \|A\|_\fro \|B\|_\fro \|S\|$ leads to 
\begin{align*}
\frac{ \dd }{ \dd t} F_t^{(i)} & =  \langle (X_i - X_i' ) u_i^\top ,   ( R +\tfrac{2}{n} \cT([ \bX^{(i)}(t)]^\top \bu) )  \rangle  \\
& \le   \| X_i - X_i' \| \|  u_i\|   \left[ \|R\|  +\tfrac{2}{n} \|\bX^{(i)}(t)\|_\fro \|\bu\|_\fro   \|S\|  \right]. 
\end{align*}
Since $\|u_i\| \le \sqrt{d} \rho$ and $\|\bX^{(i)}(t)\|_\fro,  \|\bu\|_\fro \le \sqrt{n d } \rho$ almost surely it follows that
\begin{align}
 F -  F^{(i)}  \le \|X_i - X_i' \|  \sqrt{d }\rho  \left( \|R\|  + 2 d \rho^2  \|S\|   \right). \label{eq:F_to_Fi}
 \end{align} 
%\begin{align}
% & F -  F^{(i)} \\
% & =   \int_0^1 \frac{\dd}{ \dd t} F_t^{(i)} \, \dd t  \notag \\
%& = \int_0^1  \int  \langle (X_i - X_i' ) u_i^\top ,   ( R +\tfrac{2}{n} \cT([ \bX^{(i)}(t)]^\top \bu) )  \rangle     \, \dd \mu_{t}(\bu)   \, \dd t \notag \\
%& \le  \|X_i - X_i' \| \int_0^1  \int  \| u_i  \|  \left[ \|R\|  +  \frac{ 2}{n}  \| \bX^{(i)}(t)  \|_\fro \|\bu \|_\fro \|S\|  \right]  \, \dd \mu_{t}(\bu)   \, \dd t \notag \\
%%& \le  \|X_i - X_i' \| \int_0^1  \int  \| u_i  \| \|   R + \tfrac{2}{n} \cT([  \bX^{(i)}(t)]^\top \bu)  \|     \, \dd \mu_{t}(\bu)   \, \dd t \notag \\
%& \le \|X_i - X_i' \|  \sqrt{d }\rho  \left( \|R\|  + 2 d \rho^2  \|S\|   \right). \label{eq:F_to_Fi}
%\end{align}
%In the second line, we have used the identity $\langle X^{\otimes 2} , \bu^{\otimes 2} S\rangle  = \langle (\bu^\top X)^{\otimes 2} , S\rangle  = \langle  \bX, \bu \cT( \bu^\top \bX)\rangle$ where $\cT$ is the linear operator defined in Lemma~\ref{lem:lin_op}. The last step follows from the assumptions $\|X_i \| \le \sqrt{d} \rho$ and $\|u_i\| \le \sqrt{d} \rho$. 
Applying the same argument for $F^{(i)} - F$ and then squaring gives
\begin{align}
\ex{ ( F- F^{(i)} )^2}  \le 2   \left( d \rho^2 \|R\|  + 2 d^2 \rho^4  \|S\|   \right)^2, \label{eq:F_to_Fi_square}
\end{align}
where we have used the fact that $\ex{ \|X_i - X_i'\|^2} \le \ex{ \|X_i\|^2} \le \sqrt{d} \rho$. 

Now we consider the third term in \eqref{eq:VarF_decomp}. For $\bz_1, \bz_2 \in \cZ$ let  $(\bX_1, \bX_2)$ be drawn from a joint distribution  with marginals $P_{\bX \mid \bZ =\bz_1}$ and $P_{\bX \mid \bZ =\bz_2}$,  and let $\bY_k$ be the observations associated  with  tuple $(\overline{\bW}, \bX_k)$ for $k =1,2$. Then, 
\begin{align}
 \ex{ F \mid \bZ = \bz_1}  - \ex{ F \mid \bZ = \bz_2} 
&= D( P_{\bY_1} \, \| \, P_{\bW} )  - D( P_{\bY_2} \, \| \, P_{\bW} )  \notag \\
%& = \ex{ \log \frac{ \dd P_{\bY_1}(\bY_1)  \dd P_{\bW}(\bY_2) }{ \dd P_{\bW}(\bY_1)  \dd P_{\bY_2}(\bY_2) } } \notag \\
%& = \ex{ \log \frac{ \dd P_{\bY_1}(\bY_1)  \dd P_{\bW}(\bY_2) }{ \dd P_{\bW}(\bY_1)  \dd P_{\bY_1}(\bY_2) } } + \ex{ \log \frac{ \dd P_{\bY_1}(\bY_2)   }{   \dd P_{\bY_2}(\bY_2) } } \notag \\
& =   \log \frac{ \int \exp\left\{  - \cH_{\bu}( \overline{\bW}, \bX_1) \right\}  \, \dd P_{\bX \mid \bZ = \bz_1}(\bu)}{  \int \exp\left\{  - \cH_{\bu}( \overline{\bW}, \bX_2) \right\}  \, \dd P_{\bX \mid \bZ = \bz_1}(\bu)}  \notag \\
& \quad - D( P_{\bY_2} \, \| \, P_{\bY _1} ). \label{eq:Fi_diff}
\end{align}
%Moreover,
%\begin{align}
%h(\bY_2) - h(\bY_1) & = \ex{ \log \frac{ \dd P_{\bY_1}(\bY_1)}{ \dd P_{\bY_1}(\bY_2)}} - D( P_{\bY_2} \, \| \, P_{\bY _1} ),
%\end{align}
%and so by symmetry and Jensen's inequality, 
%\begin{align}
%\left| h(\bY_2) - h(\bY_1) \right|  & \le \ex{ \left|  \log \frac{ \dd P_{\bY_1}(\bY_1)}{ \dd P_{\bY_1}(\bY_2)}\right| } +  \ex{  \left|  \log \frac{ \dd P_{\bY_2}(\bY_2)}{ \dd P_{\bY_2}(\bY_1)}  \right| }
%\end{align}
Using a slightly different version of the Efron-Stein inequality, we have
\begin{align*}
\ex{ \left(  \ex{ F \mid \bZ} - \ex{ F} \right)^2} & \le \sum_{i=1}^n  \ex{ \left(  \ex{ F  - F^{(i)} \mid \bZ, Z_i'  }   \right)_+^2}
\end{align*}
where $(x)_+ = \max(x,0)$ and  $F^{(i)}$ is obtained by replacing the $i$-th entry $Z_i$ with an independent copy $Z_i'$. Let $\bX$ be drawn according to the conditional distribution given $\bZ$ and let  $\bX^{(i)}$  be obtained by replacing the $i$-th row $X_i$ with an independent vector $X_i'$ drawn according to the conditional distribution given by $Z_i'$. From  \eqref{eq:Fi_diff}, it then follows that
%Let $(\bX_1, \bX_2)$ be given by the coupling where $\bX_1 = \bX$ and $\bX_2  =\bX^{(i)}$ where the $i$-th row $X_i$ is replaced by an  independent vector $X_i'$ drawn according to the conditional distribution given by $Z_i'$.  From  \eqref{eq:Fi_diff}, it follows that%the bounds given above, it follows that
\begin{align*}
\ex{ F  - F^{(i)} \mid \bZ, Z_i'  } 
&\le \ex{   \log \frac{ \int \exp\left\{  - \cH_{\bu}( \overline{\bW}, \bX) \right\}  \, \dd P_{\bX \mid \bZ }(\bu \mid \bZ)}{  \int \exp\left\{  - \cH_{\bu}( \overline{\bW}, \bX^{(i)} ) \right\}  \, \dd P_{\bX \mid \bZ }(\bu \mid \bZ)} \mid \bZ, Z_i'}  .%- D( P_{\bY_2} \, \| \, P_{\bY _1} ).
%%
%& \le \frac{1}{2} \langle R,    \ex{ X_i X_i^\top \mid X_i}  \rangle  + \frac{1}{2} \langle S ,    \ex{ (\bX^\top \bX)^{\otimes 2}  \mid \bZ }- \ex{( (\bX^{(i)})^\top \bX^{(i)} )^{\otimes 2}  \mid \bZ, Z_i'  } \rangle  \notag\\
%& \quad + \ex{ \log \frac{ \dd P_{\bY,  \bZ}(\bY, \bZ)}{ \dd P_{\bY,  \bZ}(\bY^{(i)}, \bZ)}}\\
%& \le  \ex{  \left|  \log \frac{ \dd P_{\bY,  \bZ}(\bY, \bZ)}{ \dd P_{\bY,  \bZ}(\bY^{(i)}, \bZ)}  \right| \mid \bZ, Z_i' }  +\ex{ \left|  \log \frac{ \dd P_{\bY,  \bZ}(\bY^{(i)}, \bZ^{(i)})}{ \dd P_{\bY,  \bZ}(\bY, \bZ^{(i)})} \right| \mid \bZ, Z_i' } \notag\\
%& \quad +  d \rho^2 \|R\|_\mop  + 2 d^2 \rho^4 \|S\|_\mop.
\end{align*} 
This  term can be bounded using the same steps as in \eqref{eq:F_to_Fi}, and this leads to
\begin{align*}
%\ex{ \left(  \ex{ F  - F^{(i)} \mid \bZ, Z_i'  }   \right)_+^2}
\ex{ \left(  \ex{ F \mid \bZ} - \ex{ F} \right)^2}  \le  2n  \left( d \rho^2 \|R\|  + 2 d^2 \rho^4 \|S\|\right)^2. 
\end{align*} 
Combining this bound with the  decomposition \eqref{eq:VarF_decomp} and the bounds in \eqref{eq:F_W_f},   \eqref{eq:Fstein1} and \eqref{eq:F_to_Fi_square} completes the proof. 
\end{proof}

\bibliographystyle{IEEEtran}
%\bibliographystyle{apalike}
%\bibliography{short_names,bibliography_sbm} 
%\bibliography{/Users/galenreeves/Dropbox/Library/short_names,/Users/galenreeves/Dropbox/Library/library} 
%\bibliography{mtp_arxiv.bbl}
\bibliography{matrix_tensor_product_arxiv.bbl}
%\end{singlespace}

\end{document}